\documentclass[]{article}
\usepackage{authblk}

\usepackage{url}
\usepackage{verbatim} 
\usepackage{caption}
\usepackage{algpseudocode}
\usepackage{algorithm}
\usepackage{amsmath}
\usepackage{Definitions}
\usepackage{color}
\usepackage{enumitem}
\usepackage{subfigure} 
\usepackage{mathtools}
\usepackage{tikz}
\usepackage{fullpage}

\usepackage{color}

\newcommand{\xhdr}[1]{\vspace{1.7mm}\noindent{{\bf #1.}}}
\newcommand{\explain}[2]{\underset{\mathclap{\overset{\uparrow}{#2}}}{#1}}
\newcommand{\explainup}[2]{\overset{\mathclap{\underset{\downarrow}{#2}}}{#1}}

\newcommand{\coevolve}{\textsc{Coevolve}}

\begin{document} 

\title{COEVOLVE: A Joint Point Process Model for Information Diffusion and Network Co-evolution}

\author[1]{Mehrdad Farajtabar}
\author[1]{Yichen Wang}
\author[2]{Manuel Gomez Rodriguez}
\author[1]{Shuang Li}
\author[1]{\\Hongyuan Zha}
\author[1]{Le Song}

\affil[1]{Georgia Institute of Technology, \{mehrdad,yichen.wang,sli370\}@gatech.edu, \{zha,lsong\}@cc.gatech.edu}
\affil[2]{Max Plank Institute for Software Systems, manuelgr@mpi-sws.org}

\date{}

\begin{small}
\maketitle
\end{small}

\begin{abstract}

Information diffusion in online social networks is affected by the underlying network topology, but it also has the power to change it. Online users are constantly crea\-ting new links when exposed to new information sources, and in turn these links are alternating the way information spreads. However, these two highly intertwined stochastic processes, information diffusion and network evolution, have been predominantly studied \emph{separately}, ignoring their co-evolutionary dynamics. 

We propose a temporal point process model, \coevolve, for such joint dyna\-mics, allowing the intensity of one process to be modulated by that of the other. This model allows us to efficiently simulate interleaved diffusion and network events, and generate traces obe\-ying common diffusion and network patterns observed in real-world networks. Furthermore, we also develop a convex optimization framework to learn the parameters of the model from historical diffusion and network evolution traces. We experimented with both synthetic data and data ga\-thered from Twitter, and show that our model provides a good fit to the data as well as more accurate predictions than alternatives.
\end{abstract}

\section{Introduction}
\label{sec:intro}

\setlength{\abovedisplayskip}{4pt}
\setlength{\abovedisplayshortskip}{1pt}
\setlength{\belowdisplayskip}{4pt}
\setlength{\belowdisplayshortskip}{1pt}
\setlength{\jot}{3pt}

Online social networks, such as Twitter or Weibo, have become large information networks where people share, discuss and search for information of personal interest as well as breaking news~\cite{KwaLeeParMoo10}. 
In this context, users often forward to their \emph{followers} information they are exposed to via their \emph{followees}, triggering the emergence of information \emph{cascades} that travel through the network~\cite{CheAdaDowKleLes14},
and constantly create new links to information sources, triggering changes in the network itself over time. 
Importantly, recent empirical studies with Twitter data have shown that both information diffusion and network evolution are coupled and network changes are often triggered by information diffusion~\cite{AntDov13,WenRatPerGonCasBonSchMenFla13,MyeLes14}. 

%


While there have been many recent works on mo\-de\-ling information diffusion~\cite{GomLesKra10,GomBalSch11,DuSonRodZha13,CheAdaDowKleLes14,FarGomDuZamZhaSon15} and network evolution~\cite{ChaZhaFal2004, LesBacKumTom08,LesChaKleFaletal10}, most of them treat these two stochastic processes independently and separately, ignoring the influence one may have on the other over time. 
%
%
Thus, to better understand information diffusion and network evolution, there is an urgent need for joint probabilistic models of the two processes, which are largely inexistent to date.

In this paper, we propose a probabilistic generative model, \coevolve, for the joint dynamics of information diffusion and network evolution. Our model is based on the framework of temporal point processes, which explicitly characterizes the continuous time interval between events, and it consists of two interwoven and interdependent components, as shown in Figure~\ref{fig:coevolution-framework}:
\begin{itemize}
  \item[I.] {\bf Information diffusion process.} We design an ``identity revealing'' multivariate Hawkes process~\cite{Liniger09} to capture the mutual excitation behavior of retweeting events, where the intensity of such events in a user is boosted by previous events from her time-varying set of followees. Although Hawkes processes have been used for information diffusion before~\cite{BluBecHelKat12, IwaShaGha13, ZhoZhaSon13, ZhoZhaSon13b, FarDuGomValZhaSon14, LinAdaRya14, DuFarAhmSmoSon15, ValGom15}, the key innovation of our approach is to ex\-pli\-cit\-ly model the excitation due to a par\-ti\-cu\-lar source node, hence revealing the identity of the source. Such design reflects the reality that information sources are explicitly acknowledged, and it also allows a particular information source to acquire new links in a rate according to her ``informativeness''.   
  \item[II.] {\bf Network evolution process.} We model link creation as an ``information driven'' survival process, and couple the intensity of this process with retweeting events. Although survival processes have been used for link creation before~\cite{HunSmyVuAsu11,VuHunSmyAsu11}, the key innovation in our model is to incorporate re\-tweeting\- events as the driving force for such processes. Since our model has captured the source identity of each retweeting event, new links will be targeted toward information sources, with an intensity proportional to their degree of excitation and each source'{}s influence. 
\end{itemize}


\begin{figure}[t]
        \centering
        \begin{tabular}{c} 
       		 \includegraphics[width=0.4\textwidth]{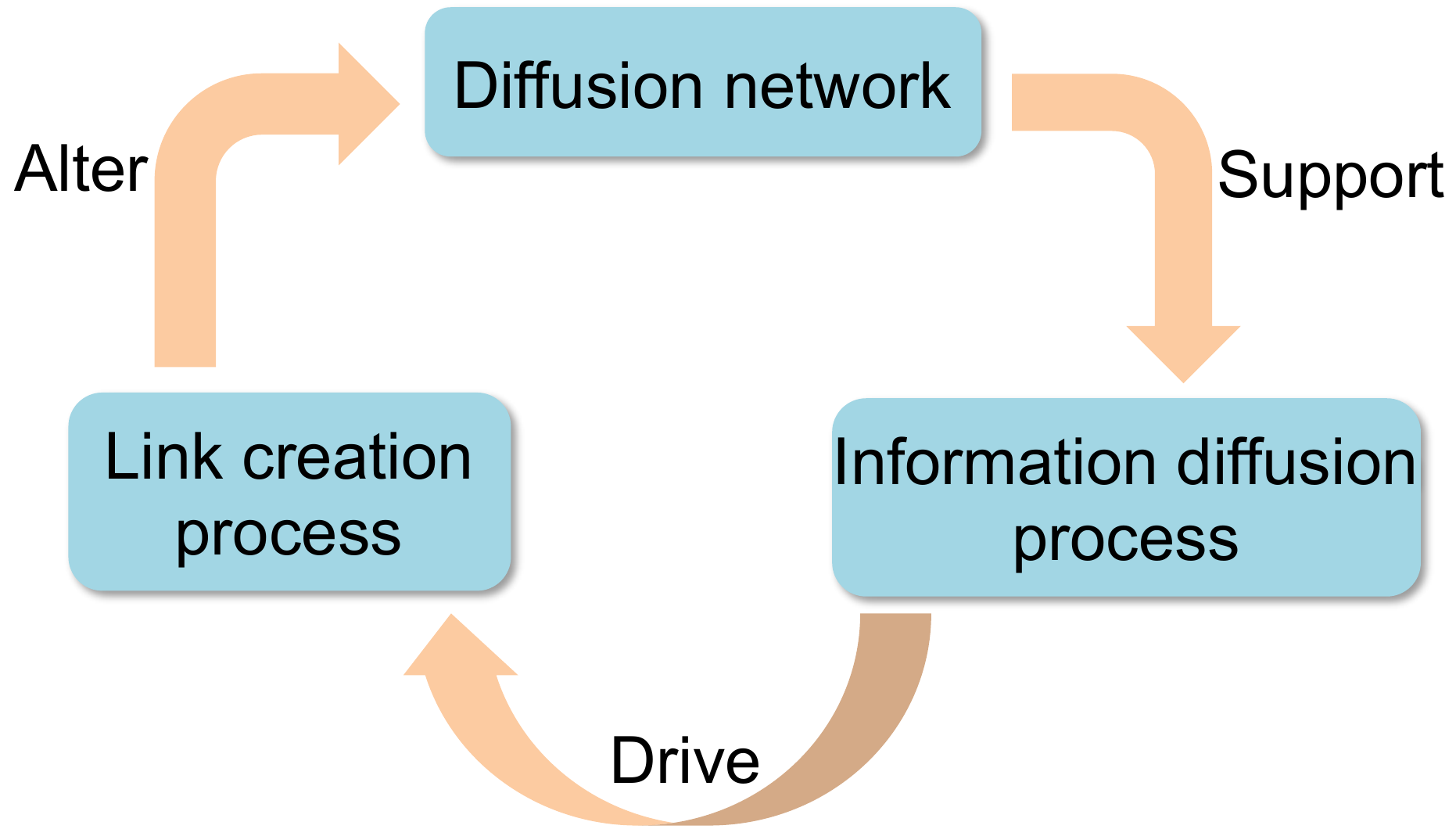} 
        \end{tabular}
        \vspace{-1mm}
        \caption{Illustration of how information diffusion and network structure processes interact}
        \label{fig:coevolution-framework}
\end{figure}

Our model is designed in such a way that it allows the two processes, information diffusion and network evolution, unfold simultaneously in the same time scale and exercise bidirectional influence on each other, allowing sophisticated coevolutionary dynamics to be generated, as illustrated in Figure \ref{fig:coevolution-simple}. 

Importantly, the flexibility of our model does not prevent us from efficiently simulating diffusion and link events from the model and learning its parameters from real world data: 
\begin{itemize}
  \item {\bf Efficient simulation.} We design a scalable sampling procedure that exploits the sparsity of the generated networks. Its complexity is $O(n d \log m)$, where $n$ is the number of events, 
  $m$ is the number of users and $d$ is the maximum number of followees per user.
  \item {\bf Convex parameters learning.} We show that the model parameters that maximize the joint likelihood of observed diffusion and link creation events can be efficiently found via convex optimization.
\end{itemize}
Then, we experiment with our model and show that it can produce coevolutionary dynamics of information diffusion and network evolution, and generate retweet and link events that obey common information diffusion patterns (\eg, cascade structure, size and depth), static network patterns (\eg, node degree) and temporal network patterns (\eg, shrinking diameter) described in related literature~\cite{LesKleFal05,LesChaKleFaletal10,GoeWatGol12}. 
Finally, we show that, by modeling the coevolutionary dynamics, our model provides significantly more accurate link and diffusion event predictions than alternatives in large scale Twitter dataset~\cite{AntDov13}. 

The remainder of this article is organized as follows. We first proceed by building sufficient background on the temporal point processes framework in Section~\ref{sec:background}. Then, we introduce our joint model of information diffusion and network structure co-evolution in 
Section~\ref{sec:model}. 
Sections~\ref{sec:simulation} and \ref{sec:estimation} are devoted to answer two essential questions: how can we generate data from the model? and how can we efficiently learn the model parameters from historical event data? Any generative model should be able to answer the above questions.
In Sections~\ref{sec:properties},~\ref{sec:synthetic}, and~\ref{sec:real} we perform empirical investigation of the properties of the model, we evaluate the accuracy of the parameter estimation in synthetic data, and we evaluate the performance of the proposed model in real-world dataset, respectively. 
Section~\ref{sec:related} reviews the related work and Section~\ref{sec:extensions} discusses some extensions to the proposed model.
Finally, the paper is concluded in Section~\ref{sec:discussion}.
\begin{figure}[t]
        \centering
        \begin{tabular}{c} 
        \includegraphics[width=0.65\textwidth]{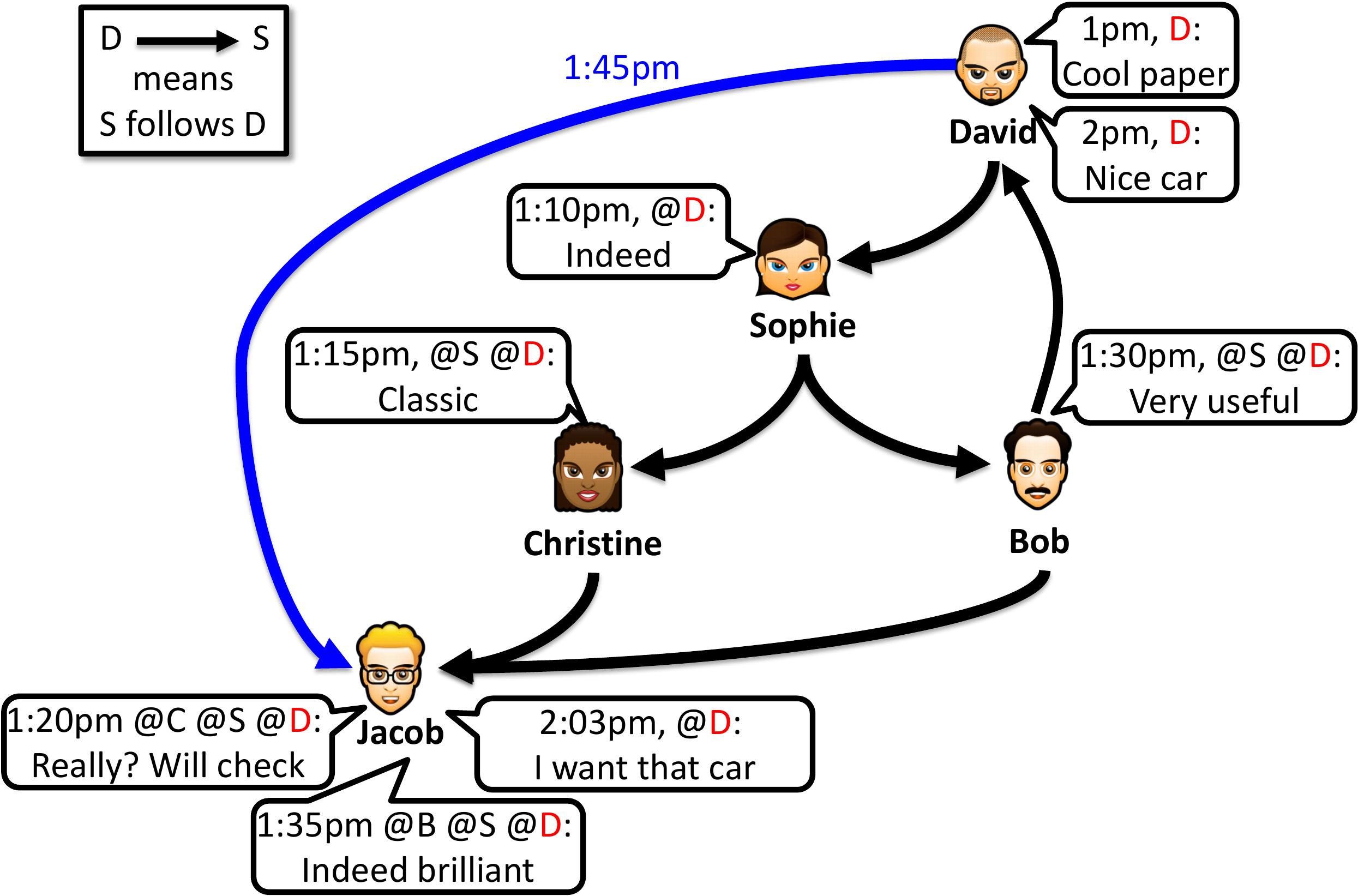} 
        \end{tabular}
        \caption{Illustration of information diffusion and network structure co-evolution: David's tweet at 1:00 pm about a paper is retweeted by Sophie and Christine respectively at 1:10 pm and 1:15 pm to reach out to Jacob. Jacob retweets about this paper at 1:20 pm  and 1:35 pm and then finds David a good source of information and decides to follow him directly at 1:45 pm. Therefore, a new path of information to him (and his downstream followers) is created. As a consequence, a subsequent tweet by David about a car at 2:00 pm directly reaches out to Jacob without need to Sophie and Christine retweet.}
        \label{fig:coevolution-simple}
\end{figure}

%
%

%

\section{Background on Temporal Point Processes}
\label{sec:background}
A temporal point process is a random process whose rea\-li\-za\-tion consists of a list of discrete events localized in time, $\cbr{t_i}$ with $t_i \in \RR^+$ and $i \in \ZZ^+$. Many different types of data produced in online social networks can be represented as temporal point processes, such as the times of retweets and link creations. A temporal point process can be equivalently re\-pre\-sen\-ted as a counting process, $N(t)$, which records the number of events before time $t$. Let the history $\Hcal(t)$ be the list of times of events $\cbr{t_1, t_2, \ldots, t_n}$ up to but not in\-clu\-ding time $t$. Then, the number of observed events in a small time window $[t, t+dt)$ of length $dt$ is 
\begin{align}
d N(t) = \sum_{t_i \in \Hcal(t)} \delta(t-t_i) \, dt,
\end{align} 
and hence $N(t) = \int_0^t dN(s)$, where $\delta(t)$ is a Dirac delta function. %
More generally, given a function $f(t)$, we can define the convolution with respect to $dN(t)$ as 
\begin{align}
  \label{eq:f_conv_dc}
  f(t)\, \star\, dN(t) := \int_{0}^t f(t-\tau)\, dN(\tau) = \sum\nolimits_{t_i \in \Hcal(t)} f(t-t_i). 
\end{align}
The point process representation of temporal data is fundamentally different from the discrete time representation typically used in social network analysis. It directly models the time interval between events as random variables, avoids the need to pick a time window to aggregate events, and 
allows temporal events to be modeled in a fine grained fashion. Moreover, it has a remarkably rich theoretical support~\cite{AalBorGje08}.  
 \begin{figure}[t]
        \centering
        \begin{tabular}{c} 
       		 \includegraphics[width=0.48\textwidth]{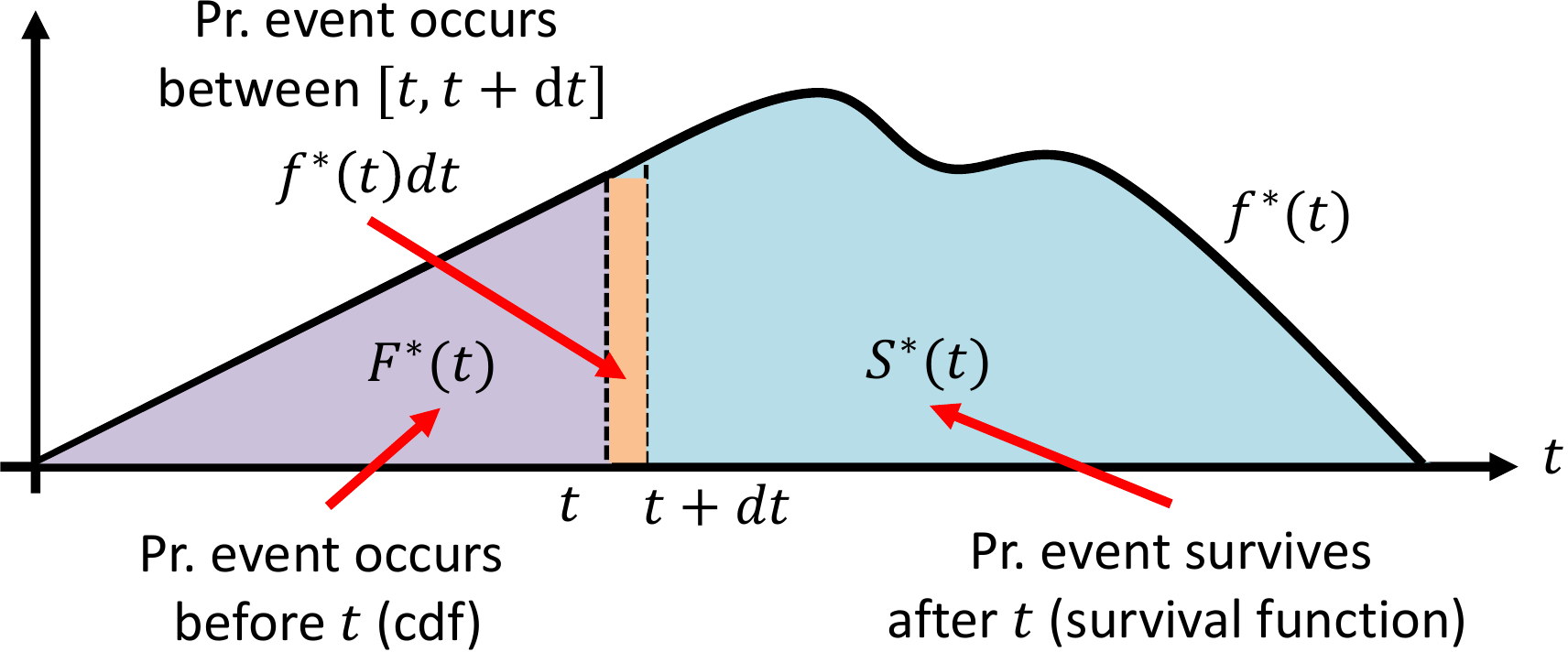} 
        \end{tabular}
        \vspace{-2mm}
        \caption{Illustration of the conditional density function, the conditional cumulative density function and the survival function}
        \label{fig:intensity-probability-survival}
                \vspace{-2mm}
\end{figure}

An important way to characterize temporal point processes is via the conditional intensity function --- a stochastic model for the time of the next event given all the times of previous events. Formally, the conditional intensity function $\lambda^*(t)$ (intensity, for short) is the conditional probability of observing an event in a small window $[t, t+dt)$ given the history $\Hcal(t)$, \ie,
\begin{align}
  \label{eq:intensity}
  \lambda^*(t)dt := \PP\cbr{\text{event in $[t, t+dt)$}|\Hcal(t)} = \EE[dN(t) | \Hcal(t)], 
\end{align}
where one typically assumes that only one event can happen in a small window of size $dt$ and thus $dN(t) \in \cbr{0,1}$. 
Then, given the observation until time $t$ and a time $t' \geqslant t$, we can also characterize the conditional probability that no event happens until $t'$ as
\begin{align}
S^*(t') = \exp\left(-{\scriptsize \int_t^{t'}} \lambda^*(\tau) \, d\tau\right),
\end{align}
the (conditional) probability density function that an event occurs at time $t'$ as 
 \begin{align}
 f^*(t') = \lambda^*(t')\, S^*(t'),
 \end{align}
and the (conditional) cumulative density function, which accounts for the probability that an event happens before time $t'$:
 \begin{align}
  F^*(t') = 1- S^*(t') = \int_t^{t'} f^*(\tau) \, d \tau.
 \end{align}
Figure~\ref{fig:intensity-probability-survival} illustrates these quantities.
Moreover, we can express the log-likelihood of a list of events $\cbr{t_1,t_2,\ldots,t_n}$ in an observation window $[0, T)$ as 
\begin{align}
  \label{eq:loglikehood_fun}
  \Lfra = \sum_{i=1}^n \log \lambda^*(t_i) - \int_{0}^T \lambda^*(\tau)\, d\tau,
   \quad  T\geqslant t_n.  
\end{align}
This simple log-likelihood will later enable us to learn the parameters of our model from observed data.
 \begin{figure}[t]
        \centering
        \begin{tabular}{l r}
       		 a) Poisson process \quad\quad\quad & \includegraphics[width=0.45\textwidth]{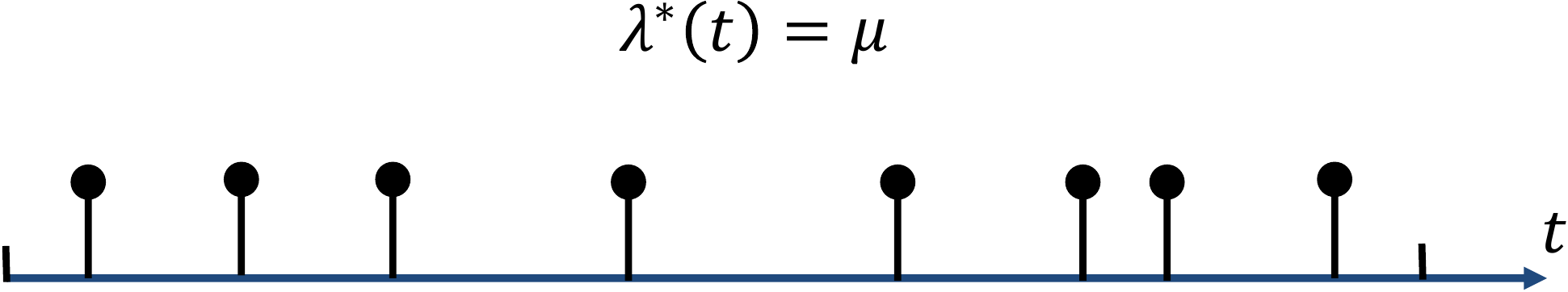} \vspace{4mm} \\
		 b) Hawkes process & \includegraphics[width=0.45\textwidth]{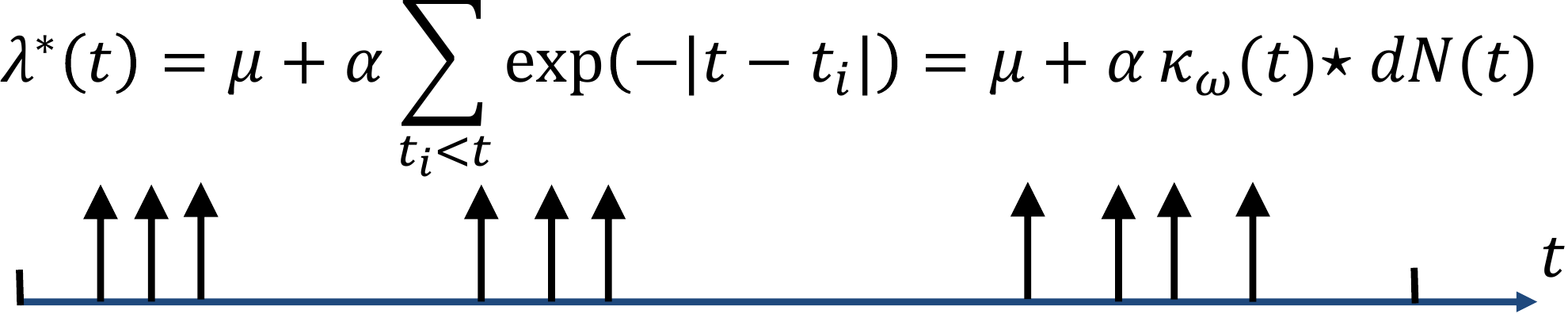}  \vspace{4mm}  \\
		 c) Survival process & \includegraphics[width=0.45\textwidth]{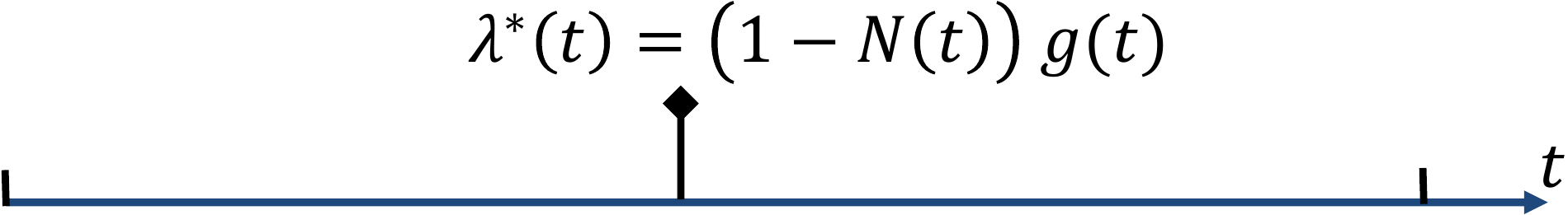}  \vspace{3mm}  \\
        \end{tabular}
                \vspace{-2mm}
        \caption{Three types of point processes with a typical realization}
        \label{fig:processes}
\end{figure}
Finally, the functional form of the intensity $\lambda^*(t)$ is often designed to capture the phenomena of interests. Some useful functional forms we will use are~\cite{AalBorGje08}: 
\begin{itemize}
 \item[(i)] {\bf Poisson process.} The intensity is assumed to be independent of the history $\Hcal(t)$, but it can be a nonnegative time-varying function,~\ie, 
 \begin{align}
 \lambda^*(t) = g(t)\geqslant 0.
 \end{align}
 \item[(ii)] {\bf Hawkes Process.} The intensity is history dependent and models a mutual excitation between events,~\ie,
 \begin{align}
   \label{eq:hawkes}
   \lambda^*(t) = \mu + \alpha \kappa_\omega(t) \star dN(t) = \mu + \alpha \sum\nolimits_{t_i \in \Hcal(t)} \kappa_{\omega}(t-t_i),
 \end{align}
 where,
 \begin{align}
 \kappa_{\omega}(t):= \exp(-\omega t)\II[t\geqslant 0]
 \end{align}
 is an exponential triggering kernel and $\mu\geqslant 0$ is a baseline intensity independent of the history. Here, the occurrence of each historical event increases the intensity by a certain amount 
 determined by the kernel and the weight $\alpha \geqslant 0$, making the intensity history dependent and a stochastic process by itself. 
 In our work, we focus on the exponential kernel, however, other functional forms, such as log-logistic function, are possible, and the general properties of our model do not depend on this particular 
 choice.
 \item[(iii)] {\bf Survival process.} There is only one event for an instantiation of the process,~\ie,
 \begin{align}
   \label{eq:survival_process}
   \lambda^*(t) = (1 - N(t)) g(t), 
 \end{align}
 where $g(t)\geqslant 0$ and the term  $(1 - N(t))$ makes sure $\lambda^*(t)$ is $0$ if an event already happened before $t$.    
\end{itemize}
Figure~\ref{fig:processes} illustrates these processes. Interested reader should refer to~\cite{AalBorGje08} for more details on the framework of temporal point processes.

\section{Generative Model of Information Diffusion and Network Evolution}
\label{sec:model}
In this section, we use the above background on temporal point processes to formulate \coevolve, our probabilistic model for the joint dynamics of information diffusion and network evolution.

\subsection{Event Representation}
 \begin{figure}[t]
        \centering
        \begin{tabular}{c c}
           \hspace{-4mm}
       	   \includegraphics[width=0.51\textwidth]{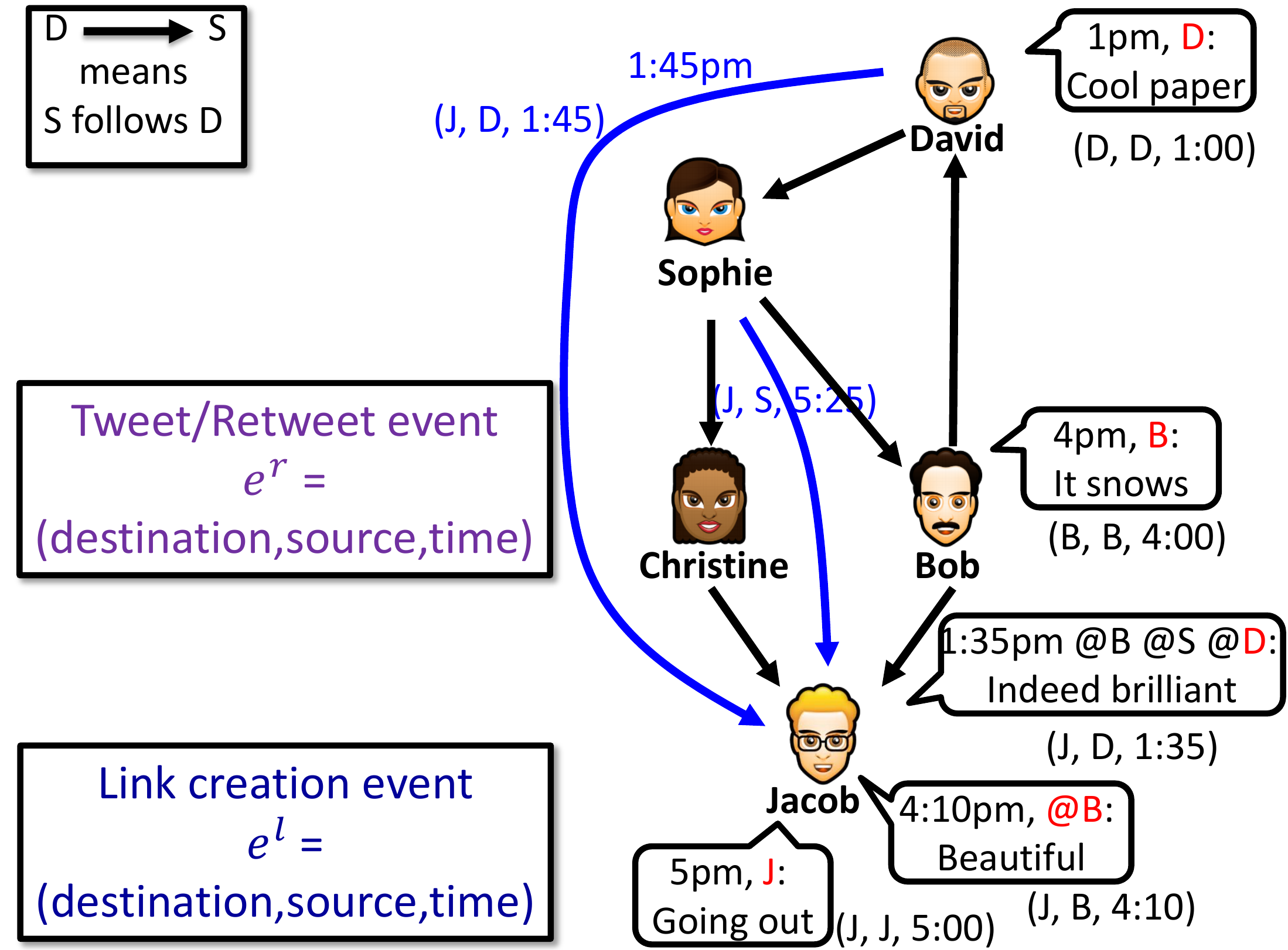}  & \hspace{-5mm}
	  \includegraphics[width=0.51\textwidth]{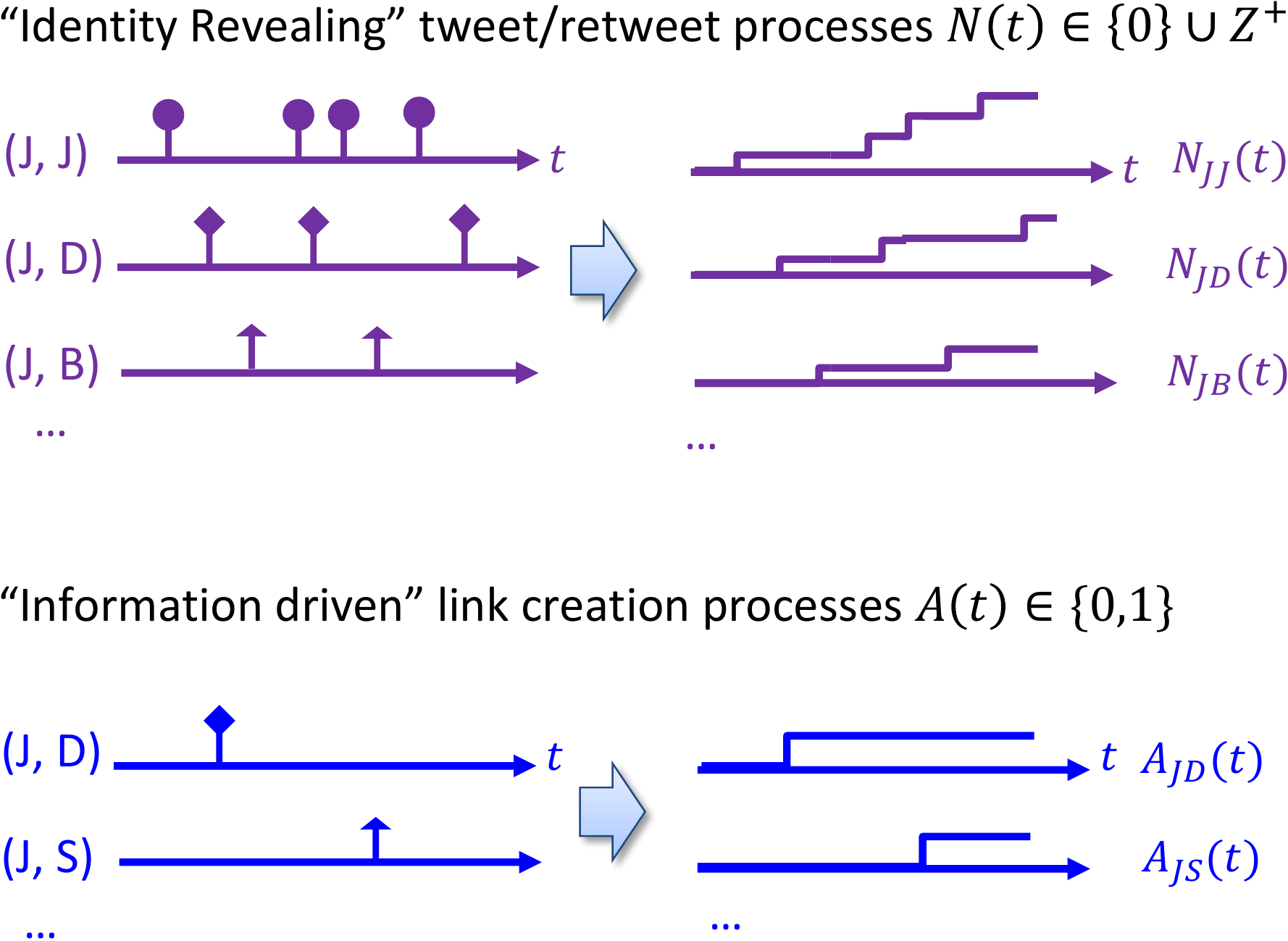} \hspace{-5mm} \\
	  	  \vspace{1mm} & 	  \vspace{1mm}\\
	  a) Event representation & b) Point and counting processes
        \end{tabular}
        \caption{Events as point and counting processes. Panel (a) shows a trace of events generated by a tweet from David followed by new links Jacob creates to follow David and Sophie. Panel (b) shows the associated points in time and the counting process realization. }
        \label{fig:events}
\end{figure}
We model the generation of two types of events: tweet/retweet events, $e^r$, and link creation events, $e^l$. Ins\-tead of just the time $t$, we record each event as a triplet,
as illustrated in Figure~\ref{fig:events}(a):
\begin{align}
  \label{eq:event}
  e^r~~~\text{or}~~~e^l ~~:=~~ (~\explain{u}{\text{destination}},~~~\explainup{s}{\text{source}},~~~\explain{t}{\text{time}}~).
\end{align}

{\bf For retweet event}, the triplet means that the destination node $u$ retweets at time $t$ a tweet originally posted by source node $s$. Recording the source node $s$ reflects the real world scenario that information sources are ex\-pli\-cit\-ly acknowledged. Note that the occurrence of event $e^r$ does \emph{not} mean that $u$ is directly retweeting from or is connected to $s$. This event can happen when $u$ is re\-twee\-ting a me\-ssage by another node $u'$ where the original information source $s$ is acknowledged. Node $u$ will pass on the same source acknowledgement to its followers (\eg, ``I agree @a @b @c @s''). Original tweets posted
by node $u$ are allowed in this notation. In this case, the event will simply be $e^r=(u,u,t)$. Given a list of retweet events up to but not including time $t$, the history $\Hcal_{us}^r(t)$ of retweets by $u$ due to 
source $s$ is 
\begin{align}
\Hcal_{us}^r(t) = \cbr{e_i^r = (u_i,s_i,t_i) | u_i = u~\text{and}~s_i = s}.
\end{align}
The entire history of retweet events is denoted as 
\begin{align}
\Hcal^r(t):=\cup_{u,s \in [m]} \Hcal^{r}_{us}(t)
\end{align}

{\bf For link creation event}, the triplet means that destination node $u$ creates at time $t$ a link to source node $s$, \ie, from time $t$ on, node $u$ starts following node $s$. To ease the exposition, we restrict ourselves to the case 
where links cannot be deleted and thus each (directed) link is created only once. However, our model can be easily augmented to consider multiple link creations and deletions per node pair, as discussed in Section~\ref{sec:extensions}. 
We denote the link creation history as $\Hcal^l(t)$. 

\subsection{Joint Model with Two Interwoven Components}

Given $m$ users, we use two sets of counting processes to record the generated events, one for information diffusion and another for network evolution. More specifically,
\begin{itemize}
  \item[{I}.] Retweet events are recorded using a matrix $\Nb(t)$ of size $m\times m$ for each fixed time point $t$. The $(u,s)$-th entry in the matrix, $N_{us}(t) \in \cbr{0} \cup \ZZ^+$, counts the number of retweets of $u$ due to source $s$ up to time $t$. These counting processes are ``identity revealing", since they keep track of the source node that triggers each retweet. The matrix $\Nb(t)$ is typically less sparse than $\Ab(t)$, since $N_{us}(t)$ can be nonzero even when node $u$ does not directly follow $s$. We also let $d\Nb(t):=\rbr{~dN_{us}(t)~}_{u,s\in [m]}$. 
  \item[{II}.] Link events are recorded using an adjacency matrix $\Ab(t)$ of size $m\times m$ for each fixed time point $t$. The $(u,s)$-th entry in the matrix, $A_{us}(t) \in \cbr{0,1}$, indicates whether $u$ is directly following $s$. Therefore,
  $A_{us}(t) = 1$ means the directed link has been created before $t$. For simplicity of exposition, we do not allow self-links. The matrix $\Ab(t)$ is typically sparse, but the number of nonzero entries can change over time. We also define $d\Ab(t):=\rbr{~dA_{us}(t)~}_{u,s\in [m]}$.
\end{itemize}
%
%
%
%
%
Then, the interwoven information diffusion and network evolution processes can be characterized using\- their respective intensities 
\begin{align}
& \EE[d\Nb(t)\, |\,  \Hcal^r(t) \cup \Hcal^l(t)] = \Gammab^*(t) \, dt \\
& \EE[d\Ab(t)\,  |\,  \Hcal^r(t) \cup \Hcal^l(t)] = \Lambdab^*(t) \, dt,
\end{align} 
where, 
\begin{align}
& \Gammab^*(t) = (~\gamma_{us}^*(t)~)_{u,s\in [m]} \\
&\Lambdab^*(t) = (~\lambda_{us}^*(t)~)_{u,s\in [m]}.
\end{align}
The sign $^*$ means that the intensity matrices will depend on the joint history, 
$\Hcal^r(t) \cup \Hcal^l(t)$, 
and hence their evolution will be coupled. 
By this coupling, we make: (i) the counting processes for link creation to be ``information driven" and (ii) the evolution of the linking structure to change the information diffusion process. 
%
%
In the next two sections, we will specify the details of these two intensity matrices. 

\subsection{Information Diffusion Process} \label{sec:intensity-information-diffusion}
We model the intensity, $\Gammab^*(t)$, for retweeting events using multivariate Hawkes process~\cite{Liniger09}:
\begin{equation}
  \gamma_{us}^*(t) = \II[u = s] \, \eta_{u} + \II[u \neq s]\,              \beta_{s} \sum\nolimits_{v \in \Fcal_u(t)}  \kappa_{\omega_1}(t) \star  \rbr{A_{uv}(t)\,\, dN_{vs}(t)}, 
\label{eq:activity-occurrence}     
\end{equation}  
%
where $\II[\cdot]$ is the indicator function and $\Fcal_u(t):=\cbr{v \in [m] : A_{uv}(t) = 1}$ is the current set of follo\-wees of $u$. 
The term $\eta_u \geqslant 0$ is the intensity of original tweets by a user $u$ on his own initiative, becoming the source of 
a cascade, and the term $\beta_s \sum_{v \in \Fcal_u(t)} \kappa_{\omega}(t) \star  \rbr{A_{uv}(t)\,\, dN_{vs}(t)}$ models the propagation 
of peer influence over the network, where the triggering kernel $\kappa_{\omega_1}(t)$ models the decay of peer influence over time. 

Note that the retweeting intensity matrix $\Gammab^*(t)$ is by itself a stochastic process that depends on the time-varying network topology, the non-zero entries in $\Ab(t)$, whose growth is controlled by the network evolution process in 
Section~\ref{sec:intensity-network-evolution}. Hence the model design captures the influence of the network topology and each source'{}s influence, $\beta_{s}$, on the information diffusion process.
More specifically, to compute $\gamma_{us}^*(t)$, one first finds the current set $\Fcal_u(t)$ of followees of $u$, and then aggregates the retweets of these followees that are due to source $s$. Note that these followees may or may not \emph{directly} follow 
source $s$.
Then, the more frequently node $u$ is exposed to retweets of tweets originated from source $s$ via her followees, the more likely she will also retweet a tweet originated from source $s$. Once node $u$ retweets due to source $s$, 
the corresponding $N_{us}(t)$ will be incremented, and this in turn will increase the likelihood of triggering retweets due to source $s$ among the followers of $u$. 
Thus, the source does \emph{not} simply broadcast the message to nodes directly following her but her influence propagates through the network even to those nodes that do not directly follow her. 
Finally, this information diffusion model allows a node to repeatedly generate events in a cascade, and is very different from the independent cascade or linear threshold models~\cite{KemKleTar03} which allow at 
most one event per node per cascade. 
 \begin{figure}[t]
        \centering
        \begin{tabular}{c c c}
           \hspace{-4mm}
       	   \includegraphics[width=0.395\textwidth]{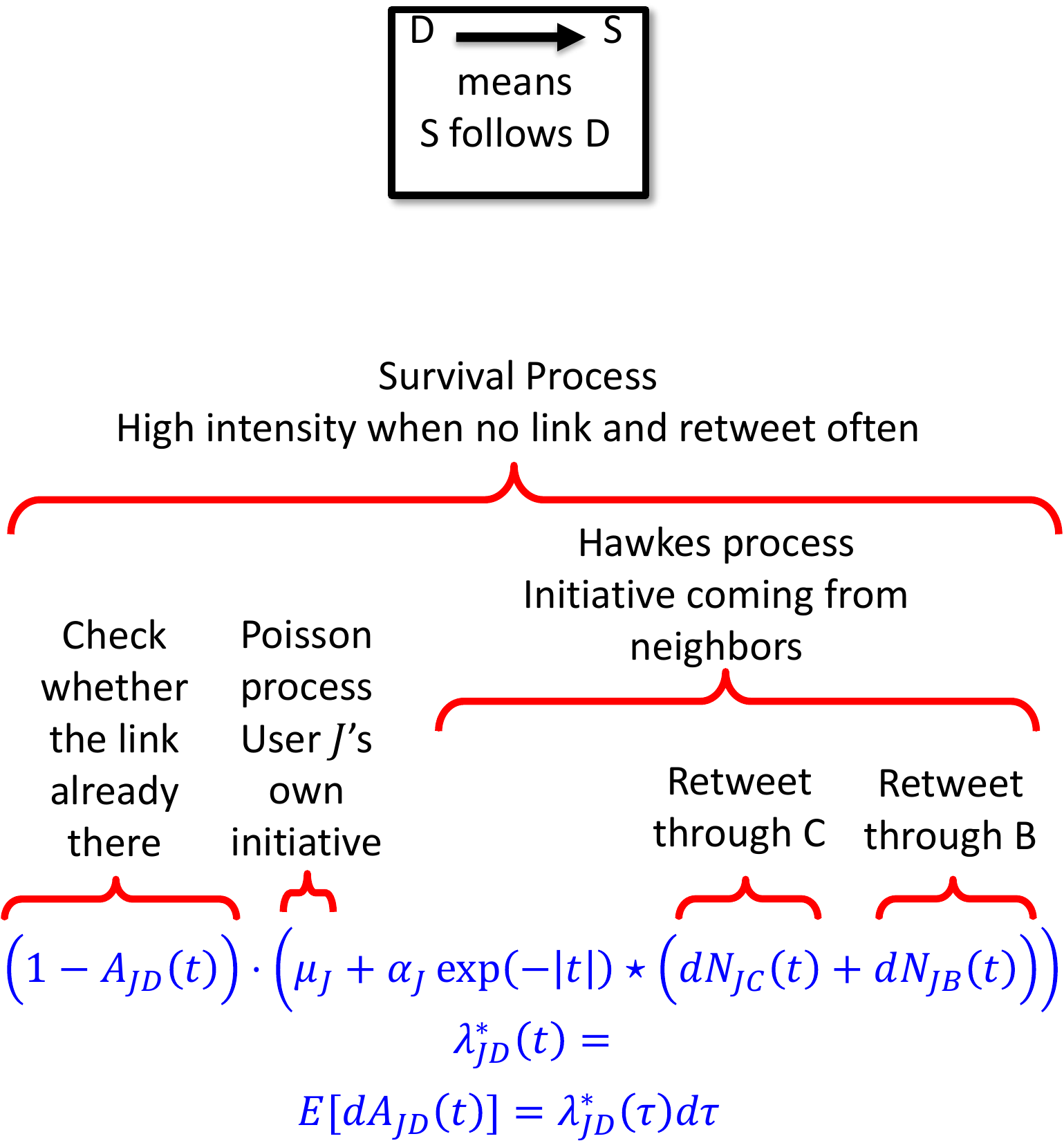} \hspace{-4mm} &
	    \includegraphics[width=0.19\textwidth]{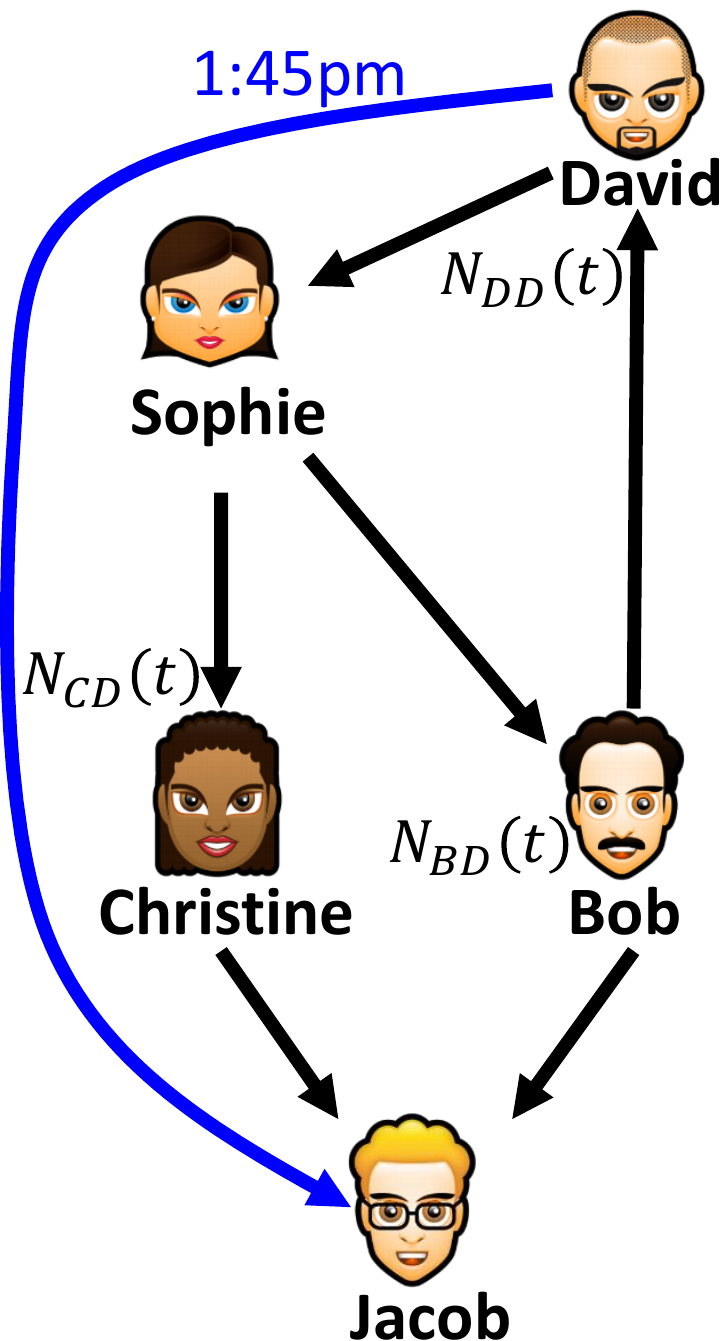} \hspace{-4mm} & \hspace{-4mm}
	   \includegraphics[width=0.365\textwidth]{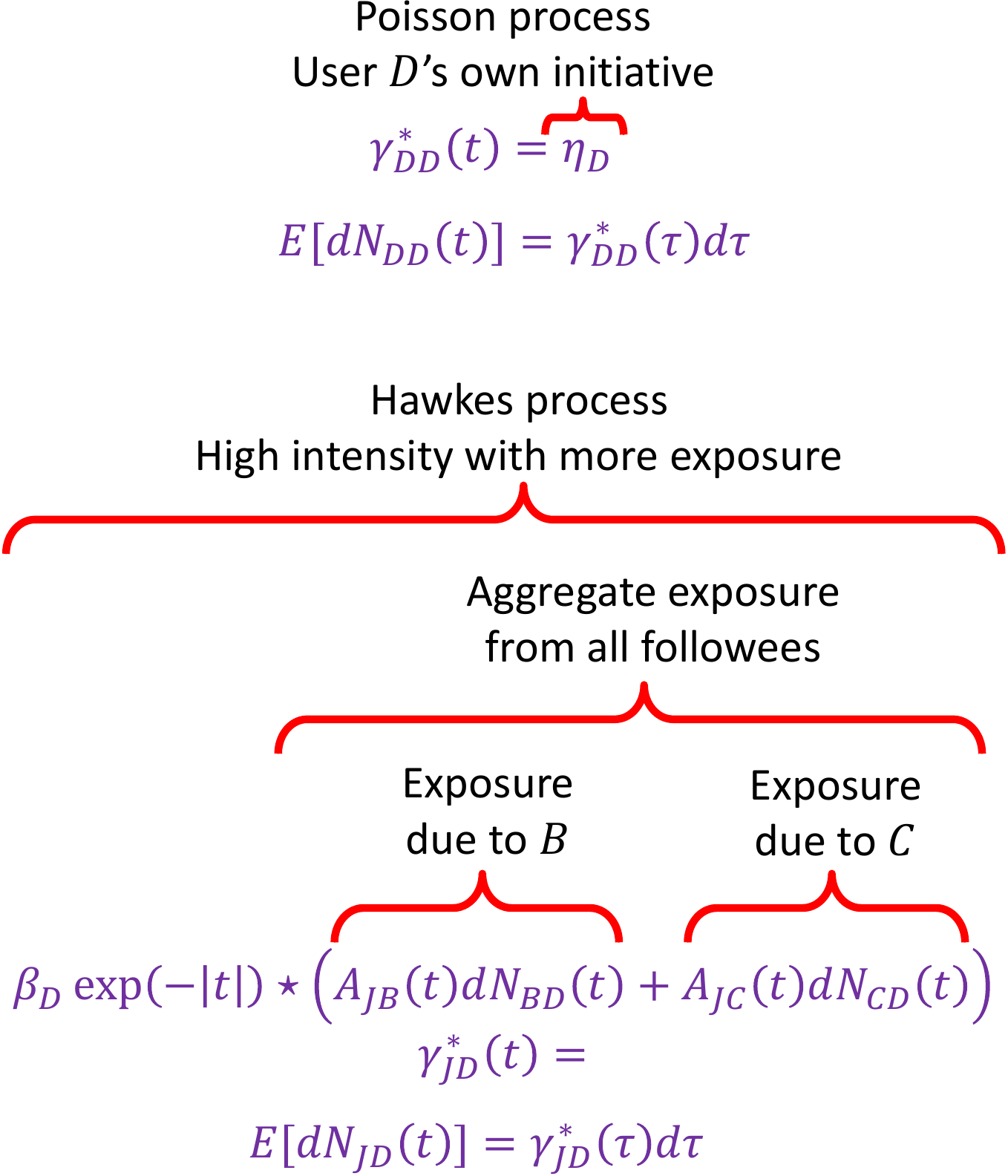} \\
	   a) Link creation process
	   & b) Social network
	   & c) Information diffusion process
        \end{tabular}
        \caption{The breakdown of conditional intensity functions for 1) information diffusion process of Jacob retweeting posts originated from David $N_{JD}(t)$; 2) information diffusion process of David tweeting on his own initiative $N_{DD}(t)$; 3) link creation process of Jacob following David $A_{JD}(t)$}
        \label{fig:intervowen}
\end{figure}

\subsection{Network Evolution Process} \label{sec:intensity-network-evolution}

In our model, each user is exposed to information through a time-varying set of neighbors. By doing so, information diffusion affects network evolution, increasing the practical 
 application of our model to real-world network datasets.
 The particular definition of exposure (\eg, a retweet'{}s neighbor) depends on the type of historical information that is available.
 Remarkably, the flexibility of our model allows for different types of diffusion events, which we can broadly classify into two categories. 

 In the first category, events corresponds to the times when an information cascade hits a person, for example, through a retweet from one of her 
 neighbors, but she does not explicitly like or forward the associated post. 
Here, we model the intensity, $\Lambdab^*(t)$, for link creation using a combination of survival and Hawkes process:
\begin{equation}
\label{eq:network-evolve-intensity-2}
  \lambda_{us}^*(t) = (1-A_{us}(t))\left(\mu_u + \alpha_u  \sum_{v \in \Fcal_u(t)} \kappa_{\omega_2}(t)\star dN_{vs}(t)\right),
\end{equation}
where the term $1-A_{us}(t)$ effectively ensures a link is crea\-ted only once, and after that, the corresponding intensity is set to zero. 
The term $\mu_u \geqslant 0$ denotes a baseline intensi\-ty, which mo\-dels when a node $u$ decides to follow a source $s$ spontaneously at her own initiative.
The term $\alpha_u \kappa_{\omega_2}(t) \star dN_{vs}(t)$ corresponds to the retweets by node $v$ (a followee of node $u$) which are originated from source $s$. 
The triggering kernel $\kappa_{\omega_2}(t)$ models the decay of interests over time.

In the second category, the person decides to explicitly like or forward the associated post and influencing events correspond to the times when she does so.
In this case, we model the intensity, $\Lambdab^*(t)$, for link creation as:
\begin{equation}
\label{eq:network-evolve-intensity}
  \lambda_{us}^*(t) = (1-A_{us}(t))(\mu_u + \alpha_u \,\kappa_{\omega_2}(t)\star dN_{us}(t)),
\end{equation}
%
where the terms $1-A_{us}(t)$, $\mu_u \geqslant 0$, and the decaying kernel  $\kappa_{\omega_2}(t)$ play the same role as the corresponding ones in Equation~\eqref{eq:network-evolve-intensity-2}.
The term $\alpha_u \kappa_{\omega_2}(t) \star dN_{us}(t)$ corresponds to the retweets of node $u$ due to tweets ori\-gi\-na\-lly published by source $s$.
The higher the corresponding retweet intensity, the more likely $u$ will find information by source $s$ useful and will create a \emph{direct} link to $s$.  

In both cases, the link creation intensity $\Lambdab^{*}(t)$ is also a stochastic process by itself, which depends on the retweet events, be it the retweets by the neighbors of node $u$ or
the retweets by node $u$ herself, respectively.
Therefore, it captures the influence of retweets on the link creation, and closes the loop of mutual influence between information diffusion and network topology. Figure~\ref{fig:intervowen} illustrates these two 
interdependent intensities.

Intuitively, in the latter category, information diffusion events are more prone to trigger new connections, because, they involve the target and source nodes in an explicit interaction, however, they are 
also less frequent. 
Therefore, it is mostly suitable to large event datasets, as the ones we generate in our synthetic experiments. 
In contrast, in the former category, information diffusion events are less likely to inspire new links but found in abundance. Therefore, it is more suitable for smaller datasets, as the ones we use in our real-world
experiments.
Consequently, in our synthetic experiments we used the latter and in our real-world experiments, we used the former. More generally, the choice of exposure event should be made based on the type and amount 
of available historical information. 

Finally, note that creating a link is more than just adding a path or allowing information sources to take shortcuts during diffusion. The network evolution makes fundamental 
changes to the diffusion dynamics and stationary distribution of the diffusion process in Section~\ref{sec:intensity-information-diffusion}. As shown 
in~\cite{FarDuGomValZhaSon14}, given a fixed network structure $\Ab$, the expected retweet intensity $\mub_{s}(t)$ at time $t$ due to source $s$ 
will depend of the network structure in a nonlinear fashion, \ie,
\begin{align}
  \mub_{s}(t):=\EE[\Gammab_{\cdot s}^*(t)]=(e^{(\Ab-\omega_1\Ib)t} + \omega_1(\Ab-\omega_1\Ib)^{-1}(e^{(\Ab-\omega_1\Ib)t}-\Ib))\,\etab_s,
\end{align}
where $\etab_s\in \RR^m$ has a single nonzero entry with value $\eta_s$ and $e^{(\Ab-\omega_1\Ib)t}$ is the matrix exponential.  
When $t\rightarrow\infty$, the stationary intensity $\bar \mub_s=(\Ib - \Ab / \omega)^{-1}\, \etab_s$ is also nonlinearly related to the network structure. 
Thus, given two network structures $\Ab(t)$ and $\Ab(t')$ at two points in time, which are different by a few edges, the effect of these edges on the information 
diffusion is not just an additive relation. 
Depending on how these newly created edges modify the eigen-structure of the sparse matrix $\Ab(t)$, their effect on the information diffusion dynamics can be 
very significant.

\begin{figure}[t!]
        \centering
        \begin{tabular}{c c} 
        \includegraphics[width=0.47\textwidth]{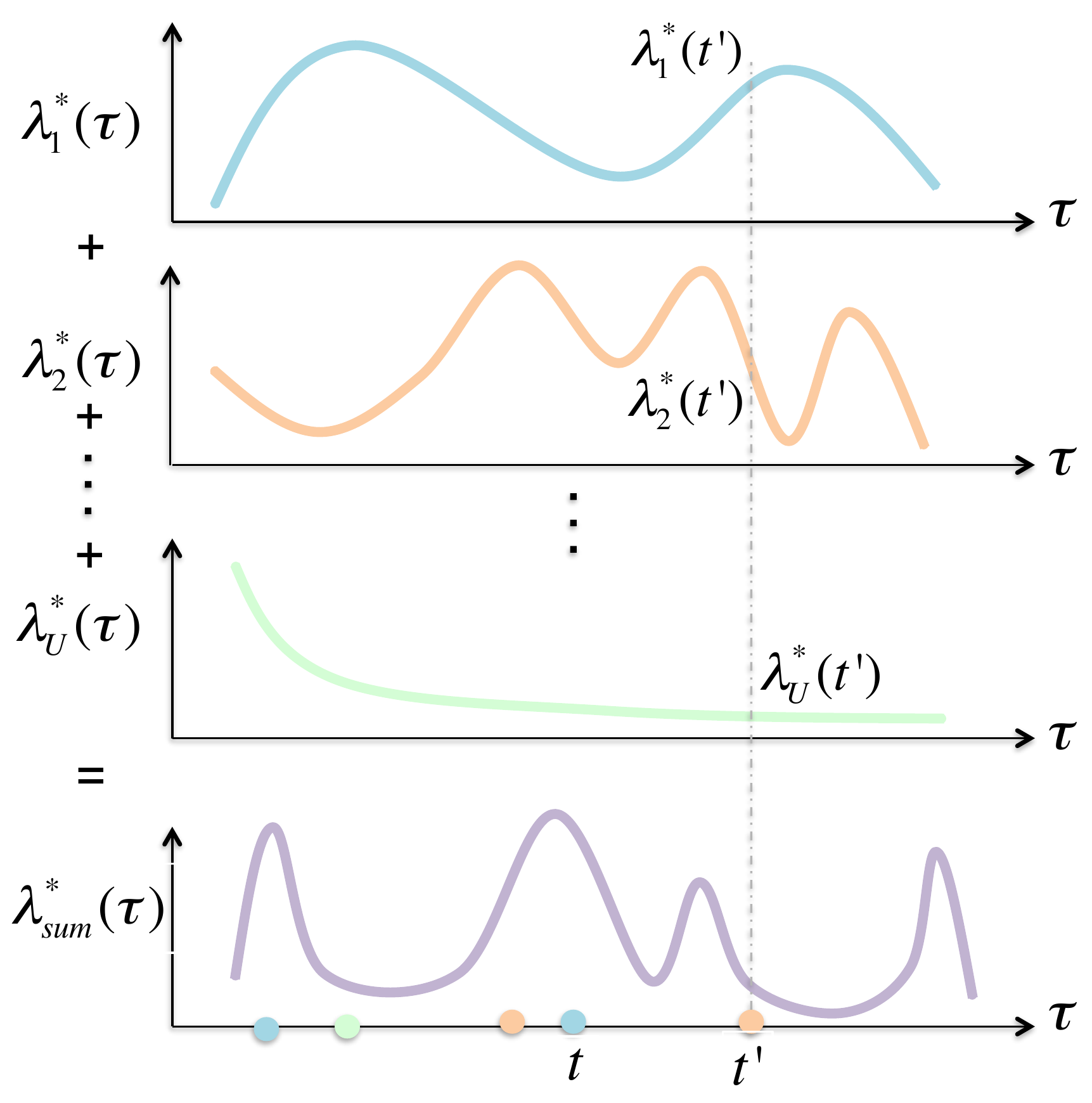} &
           \includegraphics[width=0.465\textwidth]{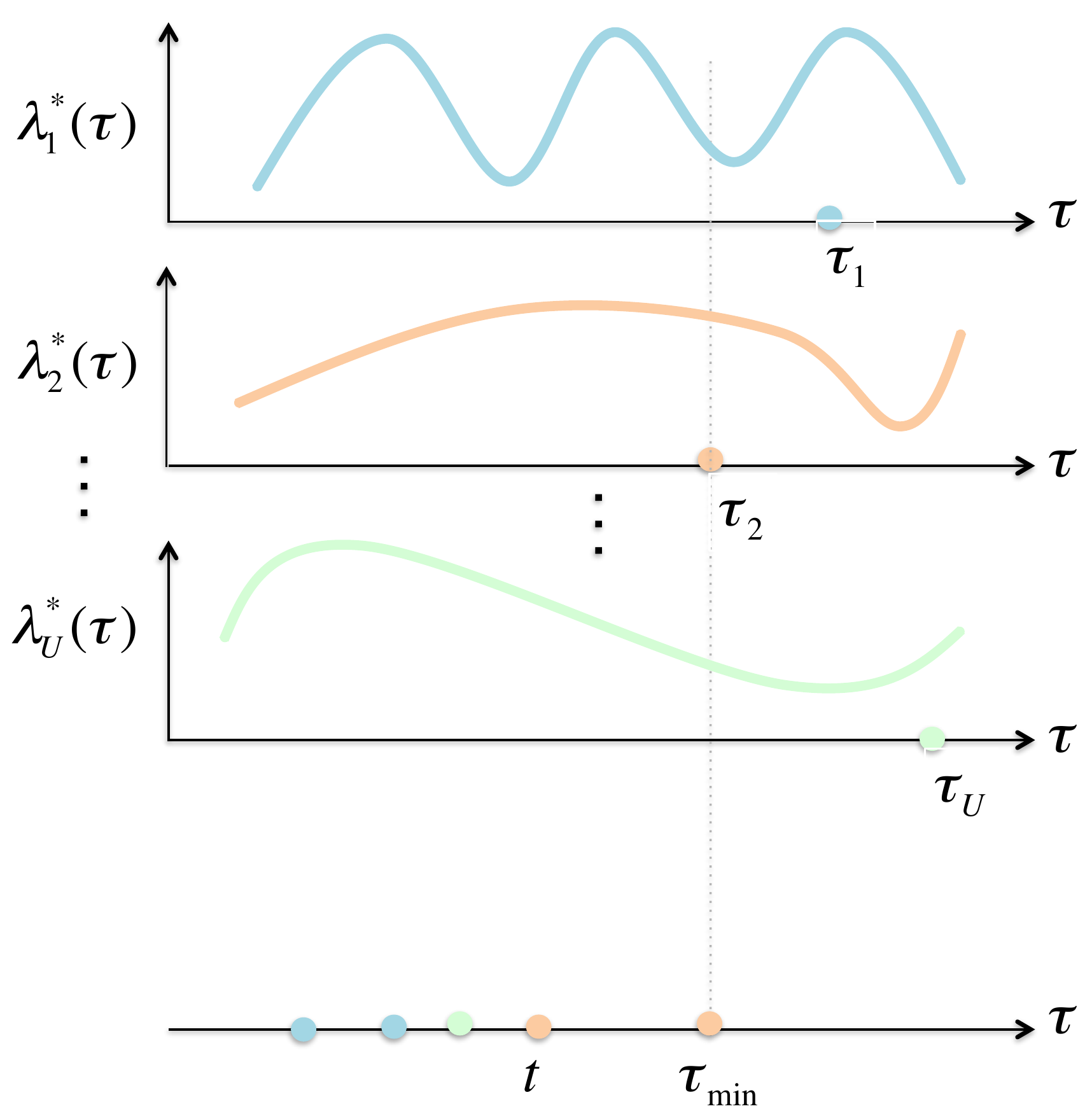} \\
          (a) Ogata'{}s algorithm & (b) Proposed algorithm
        \end{tabular}
 \caption{Ogata'{}s algorithm vs our simulation algorithm in simulating $U$ interdependent point processes characterized by intensity functions $\lambda_1(t), \ldots, \lambda_U(t)$.
 Panel (a) illustrates Ogata'{}s algorithm, which first takes a sample from the process with intensity equal to sum of individual intensities and then assigns it to the proper dimension proportionally to its contribution to the 
 sum of intensities. Panel (b) illustrates our proposed algorithm, which first draws a sample from each dimension independently and then takes the minimum time among them.}
        \label{fig:simulation}
\end{figure}

\begin{algorithm}[t]
\caption{Simulation Algorithm for \coevolve} \label{alg:simulation}
\begin{algorithmic}
\State {\bf Initialization:}
\State Initialize the priority queue $Q$ 
\For{ $\forall~u,s \in [m]$} 
\State Sample next link event $e^l_{us}$ from $A_{us}$ (Algorithm \ref{alg:sampling})
\State $Q.insert(e^l_{us})$
\State Sample next retweet event $e^r_{us}$ from $N_{us}$ (Algorithm ~\ref{alg:sampling})
\State $Q.insert(e^r_{us})$
\EndFor
\State \textbf{General Subroutine:} 
\State $t \gets 0$
\While{$t <T $}
\State  $e \gets $ $Q.extract\_min() $
\If{$e = (u, s, t')$ is a retweet event}
\State Update the history $\Hcal_{us}^r(t') = \Hcal_{us}^r(t) \cup \cbr{e}$ 
\For{$\forall~v~ s.t.~ u \rightsquigarrow v $}
\State Update event intensity: $\gamma_{vs}(t') = \gamma_{vs}(t'^{-}) + \beta$
\State Sample retweet event $e^r_{vs}$ from $\gamma_{vs}$ (Algorithm~\ref{alg:sampling})
\State $Q.update\_key(e^r_{vs})$
\If{ NOT  $s \rightsquigarrow v$} 
\State Update link intensity: $\lambda^*_{vs}(t') = \lambda^*_{vs}(t'^{-}) + \alpha$
\State Sample link event $e^l_{vs}$ from $\lambda_{vs}$ (Algorithm~\ref{alg:sampling})
\State $Q.update\_key(e^l_{vs})$
\EndIf
\EndFor
\Else
\State Update the history $\Hcal_{us}^l(t') = \Hcal_{us}^l(t) \cup \cbr{e}$
\State $\lambda^*_{us}(t) \gets 0~~\forall~ t > t'$
\EndIf
\State $t \gets t'$
\EndWhile
\end{algorithmic}
\end{algorithm}

\begin{algorithm}[t]
\caption{Efficient Intensity Computation} \label{alg:intensity}
\begin{algorithmic}
\small
\State {\bf Global Variabels:} 
\State Last time of intensity computation: $t$
\State Last value of intensity computation: $I$
\State {\bf Initialization:}
\State $t \gets 0$
\State $I \gets \mu$
\State \textbf{function} $get\_intensity(t')$
\State \hspace{2mm} $I' \gets (I - \mu) \exp(-\omega(t'-t)) + \mu$
\State \hspace{2mm} $t \gets t'$
\State \hspace{2mm} $I \gets I'$
\State \hspace{2mm} \textbf{return}  $I$
\State \textbf{end function}
\end{algorithmic}
\end{algorithm}

\begin{algorithm}[t]
\caption{1-D next event sampling}\label{alg:sampling} 
\begin{algorithmic}
\State \textbf{Input:} Current time: $t$
\State \textbf{Output:} Next event time: $s$
\State $s \gets t$
\State $\hat{\lambda} \gets \lambda^*(s)$ (Algorithm \ref{alg:intensity})
\While{$ s < T$}
\State  $g \sim Exponential(\lambdahat)$
\State $s \gets s + g$
\State $\lambdabar \gets \lambda^*(s)$ (Algorithm \ref{alg:intensity})
\State Rejection test: 
\State $d \sim Uniform(0,1)$
\If{$d\times \lambdahat  <  \lambdabar$ }
\State \textbf{return} $s$
\Else
\State $\lambdahat= \lambdabar$
\EndIf
\EndWhile
\State \textbf{return} $s$
\end{algorithmic}
\end{algorithm}

\section{Efficient Simulation of Coevolutionary Dynamics} 
\label{sec:simulation}

We could simulate samples (link creations, tweets and retweets) from our model by adapting Ogata'{}s thinning algorithm~\cite{Ogata81}, originally designed for multidimensional Hawkes 
processes. 
However, a naive implementation of Ogata'{}s algorithm would scale poorly, \ie, for each sample, we would need to re-evaluate $\Gammab^{*}(t)$ and $\Lambdab^{*}(t)$. Thus, to draw $n$ 
sample events, we would need to perform $O(m^2 n^2)$ operations, where $m$ is the number of nodes. 
Figure~\ref{fig:simulation}(a) schematically demonstrates the main steps of Ogata's algorithm.
Please refer to Appendix~\ref{app:sec:simulation} for further details.

Here, we design a sampling procedure that is especially well-fitted for the structure of our model.
The algorithm is based on the following key idea: if we consider each intensity function in $\Gammab^{*}(t)$ and $\Lambdab^{*}(t)$ as a separate point process and draw a sample from each, 
the minimum among all these samples is a valid sample for the multidimensional point process.

As the results of this section are general and can be applied to simulate any multi-dimensional point process model we abuse the notation a little bit and represent $U$ (possibly inter-dependent) point processes by $U$ intensity functions $\lambda^*_1, \ldots, \lambda^*_U$. In the specific case of simulating coevolutionary dynamics we have $U=m^2 + m(m-1)$ were the first and second terms are the number information diffusion and link creation processes, respectively.
Figure~\ref{fig:simulation} illustrates the way in which both algorithms differ. 
The new algorithm has the following steps:
\begin{enumerate}
\item Initialization: Simulate each dimension separately and find their next sampled event time.
\item Minimization: Take the minimum among all the sampled times and declare it as the next event of the multidimensional  process.
\item Update: Recalculate the intensities of the dimensions that are affected by this approved sample and re-sample only their next event. Then go to step 2.
\end{enumerate}
%
%
To prove that the new algorithm generates samples from the same distribution as Ogata{}'s algorithm does we need the following Lemma. It justifies step 2 of the above outline.

\begin{lemma}
Assume we have $U$ independent non-homogeneous Poisson processes with intensity $\lambda^*_1(\tau), \ldots, \lambda^*_U(\tau)$. Take random variable $\tau_u$ equal to the time of process $u${}'s first event after time $t$. 
Define
$ \tau_{min} = \min_{1 \le u \le U} \cbr{\tau_u}$ and 
$ u_{min} = \argmin_{1 \le u \le U} \cbr{\tau_u}$.
Then, 

(a) $\tau_{min}$ is the first event after time $t$ of the Poisson process with intensity $\lambda^*_{sum}(\tau)$.
In other words, $\tau_{min}$ has the same distribution as the next event ($t'$) in Ogata{}'s algorithm.

(b) $u_{min}$ follows the conditional distribution $\PP(u_{min}=u | \tau_{min} = x) = \frac{\lambda^*_U(x)}{\lambda^*_{sum}(x)}$.
{\it I.e.} the dimension firing the event comes from the same distribution as the one in Ogata{}'s algorithm. 
\end{lemma}

\begin{proof}
(a) The waiting time of the first event of a dimension $u$ is exponentially distributed\footnote{ If random variable $X$ is exponentially distributed with parameter $r$, then $f_X(x) = r \exp(-rx)$ is its probability distribution function and  $F_X(x) = 1- \exp(-rx)$ is the cumulative distribution function.}  random variable ~\cite{Ross06}; $\ie$, 
$
\tau_u - t \sim Exponential \rbr{\int_t^{t+\tau_u} \lambda^*_u(\tau) \, d  \tau}
$. 
We have:
\begin{equation}
\begin{split}
\PP (\tau_{min} \le x | x > t) & 
= 1 - \PP(\tau_{min} > x | x > t) 
= 1- \PP(\min \rbr{\tau_1, \ldots, \tau_U} > x | x > t) \\ & 
= 1 - \PP(\tau_1 > x, \ldots,  \tau_U > x | x > t) 
= 1- \prod_{u=1}^U \PP(\tau_u > x | x > t)  \\ & 
=  1- \prod_{u=1}^U \exp \rbr{- \int_t^{t+x} \lambda^*_u(\tau) \, d  \tau}
= 1- \exp \rbr{- \int_t^{t+x} \lambda^*_{sum}(\tau) \, d  \tau} .
\end{split}
\end{equation}
Therefore, $\tau_{min}-t$ is exponentially distributed with parameter $\int_t^{\tau_{min}} \lambda^*_{sum}(\tau) \, d \tau$ which can be seen as the first event of a non-homogenous poisson process with intensity $\lambda^*_{sum}(\tau)$ after time $t$.

(b) To find the distribution of $u_{min}$ we have
\begin{equation}
\begin{split}
\PP(u_{min} = u | \tau_{min} = x)  &
= \lambda^*_u(x) \exp \rbr{- \int_t^{t+x} \lambda^*_u(\tau) \, d  \tau} 
\prod_{v \neq u} \exp \rbr{- \int_t^{t+x} \lambda^*_v(\tau) \, d \tau } \\ &
= \lambda^*_u(x) \prod_{v} \exp \rbr{- \int_t^{t+x} \lambda^*_v(\tau) \, d \tau }.
\end{split}
\end{equation}
After normalization we get 
$ \PP(u_{min} = u | \tau_{min} = x) = \frac{\lambda^*_U(x)}{\lambda^*_{sum}(x)}$.

\end{proof}

Given the above Lemma, we can now prove that the distribution of the samples generated by the proposed algorithm is identical to the one generated by Ogata'{}s method.
\begin{theorem}
The sequence of samples from Ogata'{}s algorithm and our proposed algorithm follow the same distribution.
\end{theorem}
\begin{proof}
Using the chain rule the probability of observing $\Hcal_T = \cbr{(t_1, u_1), \ldots, (t_n, u_n)}$ is written as:
\begin{equation}
\PP \cbr{(t_1, u_1), \ldots, (t_n, u_n)} = \prod_{i=1}^n \PP \cbr{(t_i, u_i) | (t_{i-1}, u_{i-1}), \ldots, (t_1, u_1)}  = \prod_{i=1}^n \PP \cbr{(t_i, u_i) | \Hcal_{t_i}}
\end{equation}
By fixing the history up to some time, say $t_i$, all dimensions of multivariate Hawkes process become independent of each other (until next event happens).
Therefore, the above lemma  can be applied to show that the next sample time from Ogata{}'s algorithm and the proposed one come from the same distribution, \ie, for every $i$,  $\PP \cbr{(t_i, u_i) | \Hcal_{t_i}}$ is the same for both algorithms. Thus, the multiplication of individual terms is also equal for both. This will prove the theorem.
\end{proof}

This new algorithm is specially suitable for the structure of our inter-coupled processes. Since social and information networks are typically sparse, every time we sample a new node (or link) event from the model, only a small number of intensity functions in the local neighborhood 
of the node (or the link), will change. This number is of $O(d)$ where $d$ is the maximum number of followers/followees per node.  As a consequence, we can reuse most of the individual samples for the next overall sample. Moreover, we can find which intensity function has the 
minimum sample time in $O(\log m)$ operations using a heap priority queue.
The heap data structure will help maintain the minimum and find it in logarithmic time with respect to the number of elements therein. Therefore, we have reduced an $O(nm)$ factor in the original algorithm to $O(d \, \log m)$.

Finally, we exploit the properties of the exponential function to update individual intensities for each new sample in $O(1)$. For simplicity consider a Hawkes process with intensity $\lambda^*(t) = \mu + \sum_{t_i \in \Hcal_{t}} \alpha \, \omega \exp(-\omega(t-t_i))$. Note that both link creation and information diffusion processes have this structure. Now, let $t_{i} < t_{i+1}$ be two arbitrary times, we have
\begin{align}
\lambda^*(t_{i+1}) = (\lambda^*(t_{i}) - \mu) \exp(-\omega (t_{i+1}-t_{i})) + \mu.
\end{align}
It can be readily generalized to the multivariate case too. 
Therefore, we can compute the current intensity without explicitly iterating over all previous events. 
As a result we can change an $O(n)$ factor in the original algorithm to $O(1)$.
Furthermore, the exponential kernel also facilitates finding the upper bound of the intensity since it always lies at the beginning of one of the processes taken into consideration.
Algorithm~\ref{alg:intensity} summarizes the procedure to compute intensities with exponential kernels, and Algorithm~\ref{alg:sampling} shows the procedure to sample the next event 
in each dimension making use of the special property of exponential kernel functions.

The simulation algorithm is shown in Algorithm~\ref{alg:simulation}. By using this algorithm we reduce the complexity from $O(n^2 m^2)$ to $O(n d \log m)$, 
where $d$ is the maximum number of followees per node. 
That means, our algorithm scales logarithmically with the number of nodes and linearly with the number of edges at any point in time during the simulation. Moreover, events for new links, tweets and retweets are generated in a temporally intertwined and interleaving fashion, since every new retweet event will modify the intensity for link 
creation and vice versa.
%


\section{Efficient Parameter Estimation from Coevolutionary Events}
\label{sec:estimation}

In this section, we first show that learning the parameters of our proposed model reduces to solving a convex optimization problem and then develop an efficient, parameter-free Minorization-Maximization 
algorithm to solve such problem.

\subsection{Concave Parameter Learning Problem}

Given a collection of retweet events $\Ecal=\{e_i^r\}$ and link creation events $\Acal = \{e_i^l\}$ recorded within a time window $[ 0, T)$, we can easily estimate the parameters needed in our model using maximum likelihood estimation. To this aim, 
we compute the joint log-likelihood  $\Lfra$ of these events using Equation~\eq{eq:loglikehood_fun}, \ie,
\begin{equation}
\begin{split}
\Lfra(\cbr{\mu_u}, \cbr{\alpha_u}, \cbr{\eta_u}, \cbr{\beta_s})  =  &\underbrace{\sum_{e_i^r \in \Ecal} \log\rbr{\gamma^{*}_{u_i s_i}(t_i)} 
  - \sum_{u,s \in [m]} \int_0^T \gamma^{*}_{us}(\tau)\, d \tau}_{\text{tweet / retweet}} \\
 & + \underbrace{\sum_{e_i^l \in \Acal} \log\rbr{\lambda_{u_i s_i}^{*}(t_i)} -\sum_{u,s\in[m]} \int_0^{T}  \lambda_{u s}^{*}(\tau) \, d\tau}_{\text{links}}. 
\label{eq:data_loglikelihood} 
\end{split}
\end{equation}
For the terms corresponding to retweets, the log term sums only over the actual observed events while the integral term actually sums over all possible combination of destination and source pairs, even if there is no event between a particular pair 
of destination and source. For such pairs with no observed events, the corresponding counting processes have essentially survived the observation window $[0,T)$, and the term $-\int_0^T \gamma_{us}^{*}(\tau)d\tau$ simply corresponds to the 
log survival probability. The terms corresponding to links have a similar structure. 

Once we have an expression for the joint log-likelihood of the retweet and link creation events, the parameter learning problem can be then formulated as follows:
\begin{equation}
	\label{eq:optimization-main}
	\begin{array}{ll}
		\mbox{minimize}_{\cbr{\mu_u}, \cbr{\alpha_u}, \cbr{\eta_u}, \cbr{\beta_s}} & - \Lfra(\cbr{\mu_u}, \cbr{\alpha_u}, \cbr{\eta_u}, \cbr{\beta_s}) 		\\
		\mbox{subject to} & \mu_u \ge 0,\quad \alpha_u \ge 0 \quad \eta_u \ge 0,\quad \beta_s \ge 0 \quad  \forall \, u,s \in [m].
	\end{array}
\end{equation}

%
%
\begin{theorem}
The optimization problem defined by Equation~\eqref{eq:optimization-main} is jointly convex.\end{theorem}
\begin{proof}
We expand the likelihood by replacing the intensity functions into Equation \eqref{eq:data_loglikelihood}:
\begin{equation}
 \label{eq:data_loglikelihood_expan} 
 \begin{split}
   \Lfra = &  \sum_{e_i^r \in \Ecal} \log \rbr{\II[u_i = s_i] \, \eta_{u_i} + \II[u_i \neq s_i]\,              \beta_{s_i} \sum\nolimits_{v \in \Fcal_{u_i(t_i)}}  \left. \Bigl( \kappa_{\omega_1}(t) \star  \rbr{A_{u_i v}(t)\,\, dN_{vs_i}(t)} \Bigr) \right\vert_{t = t_i}} \\
  & -  \sum_{u,s \in [m]} \II[u = s] \, \eta_{u} \int_0^T \, d t + \II[u \neq s]\,   \beta_{s} \sum\nolimits_{v \in \Fcal_u(t)} \int_0^T   \kappa_{\omega_1}(t) \star  \rbr{A_{uv}(t)\,\, dN_{vs}(t)} \, d t \\
  & + \sum_{e_i^l \in \Acal} \log\rbr{
    \mu_{u_i} + \alpha_{u_i}\sum_{v \in \Fcal_{u_i}(t_i)}  \,\left. \bigl(  \kappa_{\omega_2}(t)\star dN_{vs}(t) \bigr) \right\vert_{t = t_i} }  \\
    & -\sum_{u,s\in[m]} \mu_u  \int_0^{T}  (1-A_{us}(t))  \, d t  + \alpha_u \int_0^{T} (1-A_{us}(t)) \bigl( \sum_{v \in \Fcal_{u}(t)} \kappa_{\omega_2}(t)\star dN_{vs}(t) \bigr) \, d t
    \end{split}
\end{equation}
If we stack all parameters in a vector $\xb = (\cbr{\mu_u}, \cbr{\alpha_u}, \cbr{\eta_u},\cbr{\beta_s})$, one can easily notice that the log-likelihood $\Lfra$ can be written as $\sum_{j} \log(\ab_j^{\top} \xb) - \sum_{k} \bb_k^{\top} \xb$,
which is clearly a concave function with respect to $\xb$~\cite{BoyVan04}, and thus $-\Lfra$ is convex. Moreover, the constraints are linear inequalities and thus the domain is a convex set. This completes the proof for convexity of the optimization problem.
\end{proof}

%
%

It's notable that the optimization problem decomposes in $m$ independent problems, one per node $u$, and can be readily parallelized.  
%


\begin{algorithm}[t]
\caption{MM-type parameter learning for \coevolve}
\label{alg:em-type} 
\begin{algorithmic}
\State \textbf{Input:} Set of retweet events $\Ecal=\{e_i^r\}$ and link creation events $\Acal = \{e_i^l\}$ observed in time window $[ 0, T)$
\State \textbf{Output:} Learned parameters $\cbr{\mu_u}, \cbr{\alpha_u}, \cbr{\eta_u},\cbr{\beta_s}$
\State {\bf Initialization:}
\For{$u \gets 1$ to $m$}
\State Initialize $\mu_u$ and $\alpha_u$ randomly
\EndFor
\For{$u \gets 1$ to $m$}
\State $\eta_u = \frac{ \sum_{e_i^r \in \Ecal} \II[u=u_i = s_i]}{T}$
\EndFor
\For{$s \gets 1$ to $m$}
\State $\beta_s =   \frac{ \sum_{e_i^r \in \Ecal} \II[ s=s_i \neq u_i]}
{ \sum_{u \in [m]} \II[u \neq s] \, \sum\nolimits_{v \in \Fcal_u(t)} \int_0^T   \kappa_{\omega_1}(t) \star  \rbr{A_{uv}(t)\,\, dN_{vs}(t)} \, d t }$
\EndFor
\While{ not converged}
\For{$i \gets 1$ to $n_l$}
\State $ \nu_{i1} = \frac{\mu_{u_i}}{\mu_{u_i} +  \alpha_{u_i} \sum_{v \in \Fcal_{u_i}(t_i)} \left. \bigl( \kappa_{\omega_2}(t)\star dN_{vs}(t) \bigr) \right\vert_{t = t_i}} $
\State $ \nu_{i2} = \frac{\alpha_{u_i}  \sum_{v \in \Fcal_{u_i}(t_i)} \left. \bigl( \kappa_{\omega_2}(t)\star dN_{vs}(t) \bigr) \right\vert_{t = t_i} }{\mu_{u_i} + \alpha_{u_i} \sum_{v \in \Fcal_{u_i}(t_i)} \left. \bigl( \kappa_{\omega_2}(t)\star dN_{vs}(t) \bigr) \right\vert_{t = t_i} }$
\EndFor
\For{$u \gets 1$ to $m$}
\State $ \mu_u = \frac{\sum_{e_i^l \in \Acal} \II[ u = u_i] \, \nu_{i1}} { \sum_{s\in[m]} \int_0^{T}  (1-A_{us}(t)) \, d t}$
\State $\alpha_u = \frac{\sum_{e_i^l \in \Acal} \II[u = u_i] \nu_{i2}} { \sum_{s\in[m]} \int_0^{T} (1-A_{us}(t)) (\kappa_{\omega_2}(t)\star dN_{us}(t)) \, d t} $
\EndFor
\EndWhile
\end{algorithmic}
\end{algorithm}

\subsection{Efficient Minorization-Maximization Algorithm}


Since the optimization problem is jointly convex with respect to all the parameters, one can simply take any convex optimization method to learn the parameters. However, these methods usually require hyper parameters like step size or initialization, which may 
significantly influence the convergence. Instead, the structure of our problem allows us to develop an efficient algorithm inspired by previous work~\cite{ZhoZhaSon13, ZhoZhaSon13b}, which leverages Minorization Maximization (MM)~\cite{HunLan04} and
is parameter free and insensitive to initialization.
%


Our algorithm utilizes Jensen'{}s inequality to provide a lower bound for the second log-sum term in the log-likelihood given by Equation~\eqref{eq:data_loglikelihood}. More specifically, consider a set of arbitrary auxiliary variable $\nu_{ij}$, where $1 \le i \le n_l$, $j=1,2$ and $n_l$ is the 
number of link events, \ie, $n_l = | \Acal |$. Further, assume these variables satisfy
\begin{equation}
\label{eq:aux-constraints}
\begin{split}
&  \quad \forall \,\, 1 \le i \le n_l : \quad \nu_{i1}, \nu_{i2} \ge 0, \quad  \nu_{i1} + \nu_{i2} = 1\\
\end{split}
\end{equation}
Then, we can lower bound the logarithm in Equation~\eqref{eq:data_loglikelihood_expan} using Jensen'{}s inequality as 
follows:
\begin{equation}
\begin{split}& \log  \rbr{
    \mu_{u_i} + \alpha_{u_i}  \sum_{v \in \Fcal_{u_i}(t_i)} \left. \bigl( \kappa_{\omega_2}(t)\star dN_{vs}(t) \bigr) \right\vert_{t = t_i} } \\ 
     & = \log  \rbr{
    \nu_{i1} \frac{\mu_{u_i}}{\nu_{i1}} + \nu_{i2} \frac{\alpha_{u_i}}{\nu_{i2}}  \sum_{v \in \Fcal_{u_i}(t_i)} \left. \bigl( \kappa_{\omega_2}(t)\star dN_{vs}(t) \bigr) \right\vert_{t = t_i} } \\ 
     & \ge  \nu_{i1} \log \rbr{ \frac{\mu_{u_i}}{\nu_{i1}}} + \nu_{i2} \log \rbr{ \frac{\alpha_{u_i}}{\nu_{i2}}  \sum_{v \in \Fcal_{u_i}(t_i)} \left. \bigl( \kappa_{\omega_2}(t)\star dN_{vs}(t) \bigr) \right\vert_{t = t_i} }  \\
     & \ge  \nu_{i1} \log(\mu_{u_i}) + \nu_{i2} \log(\alpha_{u_i}) + \nu_{i2} \log \rbr{  \sum_{v \in \Fcal_{u_i}(t_i)} \left. \bigl( \kappa_{\omega_2}(t)\star dN_{vs}(t) \bigr) \right\vert_{t = t_i}} \\
    &  \quad - \nu_{i1} \log(\nu_{i1}) - \nu_{i2} \log(\nu_{i2}).
    \end{split}
\end{equation}
%
%
Now, we can lower bound the log-likelihood given by Equation~\eqref{eq:data_loglikelihood_expan} as:
\begin{equation}
\label{eq:data_loglikelihood_expan2} 
\begin{split}
\Lfra  \ge
\Lfra' =  
&  \sum_{e_i^r \in \Ecal} \II[u_i = s_i]\, \log \rbr{\eta_{u_i}}
+ \sum_{e_i^r \in \Ecal} \II[u_i \neq s_i]\,    \log(\beta_{s_i})  \\
& + \sum_{e_i^r \in \Ecal} \II[u_i \neq s_i]\,  \log \bigl( \sum\nolimits_{v \in \Fcal_{u_i(t_i)}}  \left. \Bigl( \kappa_{\omega_1}(t) \star  \rbr{A_{u_i v}(t)\,\, dN_{vs_i}(t)} \Bigr) \right\vert_{t = t_i} \bigr) \\
  & -  \sum_{u,s \in [m]} \eta_{u} T +  \beta_{s} \sum\nolimits_{v \in \Fcal_u(t)} \int_0^T   \kappa_{\omega_1}(t) \star  \rbr{A_{uv}(t)\,\, dN_{vs}(t)} \, d t \\
  & + \sum_{e_i^l \in \Acal} 
\nu_{i1} \log(\mu_{u_i}) + \nu_{i2} \log(\alpha_{u_i}) + \nu_{i2} \log \bigl(  \sum_{v \in \Fcal_{u_i}(t_i)} \left. \bigl( \kappa_{\omega_2}(t)\star dN_{vs}(t) \bigr) \right\vert_{t = t_i} \bigr) \\
    &  -  \sum_{e_i^l \in \Acal}  \nu_{i1} \log(\nu_{i1}) + \nu_{i2} \log(\nu_{i2})  \\
    & -\sum_{u,s\in[m]} \mu_u  \int_0^{T}  (1-A_{us}(t))  \, d t  + \alpha_u \int_0^{T} (1-A_{us}(t)) (\kappa_{\omega_2}(t)\star dN_{us}(t)) \, d t
    \end{split}
\end{equation}

By taking the gradient of the lower-bound with respect to the parameters, we can find the closed form updates to optimize the lower-bound:
\begin{align}
\label{eq:param-update-1}
& \eta_u = \frac{ \sum_{e_i^r \in \Ecal} \II[u = u_i = s_i]\,}{T} \\
\label{eq:param-update-2}
& \beta_s =   \frac{ \sum_{e_i^r \in \Ecal} \II[s=s_i \neq u_i] }
{ \sum_{u \in [m]} \II[u \neq s] \, \sum\nolimits_{v \in \Fcal_u(t)} \int_0^T   \kappa_{\omega_1}(t) \star  \rbr{A_{uv}(t)\,\, dN_{vs}(t)} \, d t } \\
\label{eq:param-update-3}
& \mu_u = \frac{\sum_{e_i^l \in \Acal} \II[u = u_i] \, \nu_{i1}} { \sum_{s\in[m]} \int_0^{T}  (1-A_{us}(t)) \, d t} \\
\label{eq:param-update-4}
& \alpha_u = \frac{\sum_{e_i^l \in \Acal} \II[u = u_i] \, \nu_{i2}} { \sum_{s\in[m]} \int_0^{T} (1-A_{us}(t)) (\kappa_{\omega_2}(t)\star dN_{us}(t)) \, d t}.
\end{align}

Finally, although the lower bound is valid for every choice of $\nu_{ij}$ satisfying Equation~\eqref{eq:aux-constraints}, by maximizing the lower bound with respect to the auxiliary variables 
we can make sure that the lower bound is tight:
\begin{equation}
	\label{eq:optimization}
	\begin{array}{ll}
		\mbox{maximize}_{\cbr{\nu_{ij}}} &  \Lcal'(\cbr{\mu_u}, \cbr{\alpha_u}, \cbr{\eta_u},\cbr{\beta_s}, \cbr{\nu_{ij}})  \\
		\mbox{subject to} & 
		\nu_{i1} + \nu_{i2} = 1 \hspace{1cm}    \forall \,  i: 1 \le i \le n_l \\
                  & \nu_{i0}, \nu_{i1}  \ge 0 \hspace{1cm}    \forall \,  i: 1 \le i \le n_l.
	\end{array}
\end{equation}
Fortunately, the above constrained optimization problem can be solved easily via Lagrange multipliers, which leads to closed form updates:
\begin{align}
\label{eq:aux-update-1}
& \nu_{i1} = \frac{\mu_{u_i}}{\mu_{u_i} +  \alpha_{u_i} \sum_{v \in \Fcal_{u_i}(t_i)} \left. \bigl( \kappa_{\omega_2}(t)\star dN_{vs}(t) \bigr) \right\vert_{t = t_i}}\\
\label{eq:aux-update-2}
& \nu_{i2} = \frac{\alpha_{u_i}  \sum_{v \in \Fcal_{u_i}(t_i)} \left. \bigl( \kappa_{\omega_2}(t)\star dN_{vs}(t) \bigr) \right\vert_{t = t_i} }{\mu_{u_i} + \alpha_{u_i} \sum_{v \in \Fcal_{u_i}(t_i)} \left. \bigl( \kappa_{\omega_2}(t)\star dN_{vs}(t) \bigr) \right\vert_{t = t_i} }.
\end{align}

Algorithm~\ref{alg:em-type} summarizes the learning procedure. It is guaranteed to converge to a global optimum~\cite{HunLan04, ZhoZhaSon13}

%



\section{Properties of Simulated Co-evolution, Networks and Cascades}
\label{sec:properties}
In this section, we perform an empirical investigation of the properties of the networks and information cascades generated by our model. In particular, we show that our model can generate co-evolutionary retweet and link dynamics and a wide spectrum of static and temporal network patterns and information cascades.

\subsection{Simulation Settings}
Throughout this section, if not said otherwise, we simulate the evolution of a 8,000-node network as well as the pro\-pa\-ga\-tion of information over the network by sampling from our model using 
Algorithm~\ref{alg:simulation}. 
%
We set the exogenous intensities of the link and diffusion events to $\mu_u = \mu = 4 \times 10^{-6}$ and $\eta_u = \eta = 1.5$ respectively, and the triggering kernel parameter to $\omega_1=\omega_2 = 1$. The 
parameter $\mu$ determines the independent growth of the network -- roughly speaking, the expected number of links each user establishes spontaneously before time $T$ is $\mu T$. Whenever we investigate a static property, 
we choose the same sparsity level of $0.001$.  

%

%

\subsection{Retweet and Link Coevolution} Figures~\ref{fig:syntethic_coevolution}(a,b) visualize the retweet and link events, aggregated across different sources, and the corresponding intensities for one node and one realization, picked 
at random. Here, it is already apparent that retweets and link creations are clustered in time and often follow each other. 
Further, Figure~\ref{fig:syntethic_coevolution}(c) shows the cross-covariance of the retweet and link creation intensity, computed across multiple realizations, for the same node, \ie, if $f(t)$ and $g(t)$ are two intensities, the cross-covariance is a function $h(\tau) = \int f(t+\tau)g(t)\,dt$. 
It can be seen that the cross-covariance has its peak around 0, $\ie$, retweets and link creations are highly correlated and co-evolve over time.
For ease of exposition, we illustrated co-evolution using one node, however, we found consistent results across nodes.

\begin{figure}[t!]
        \centering
        \begin{tabular}{c c c} 
          \includegraphics[width=0.24\textwidth]{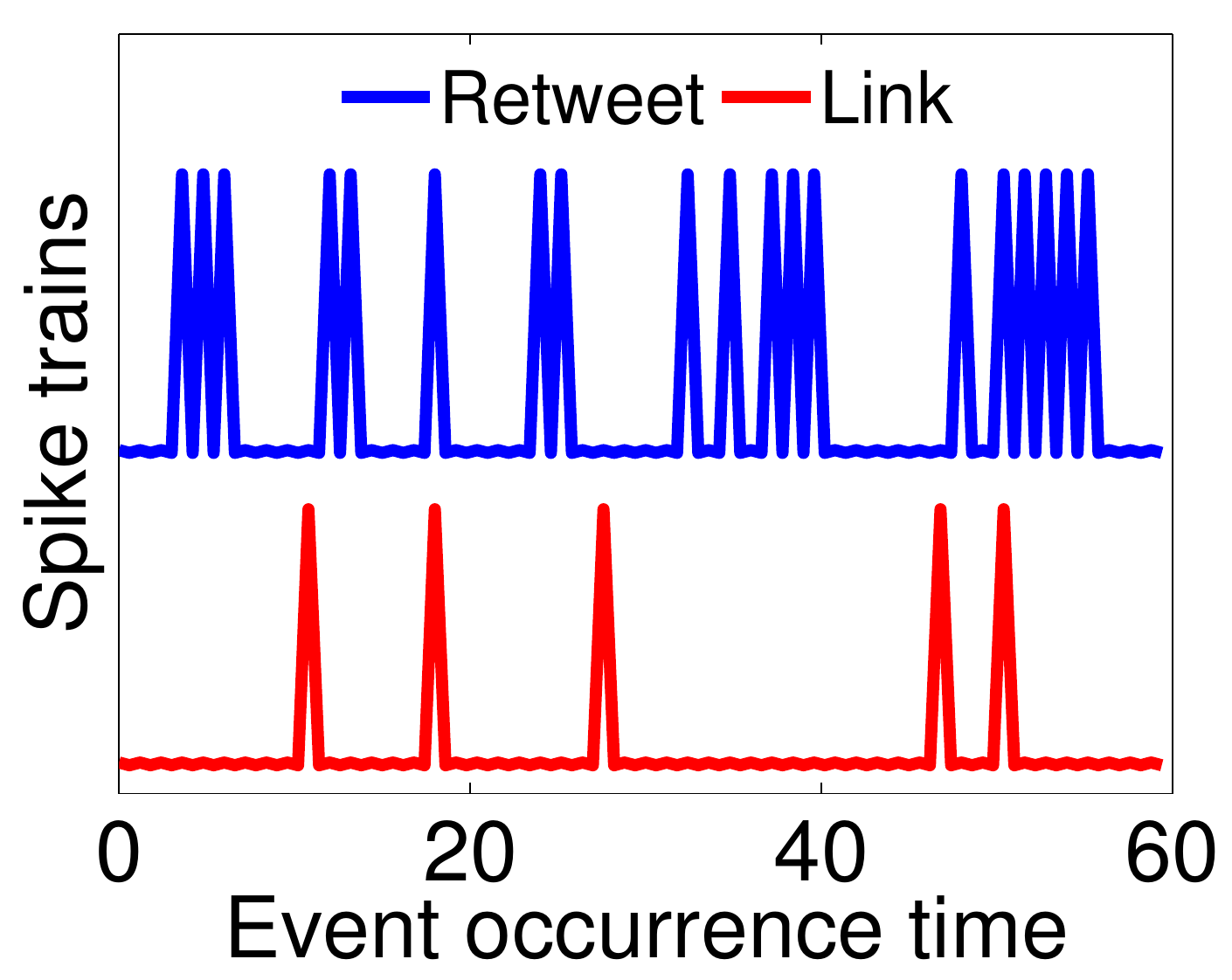} &
          \includegraphics[width=0.25\textwidth]{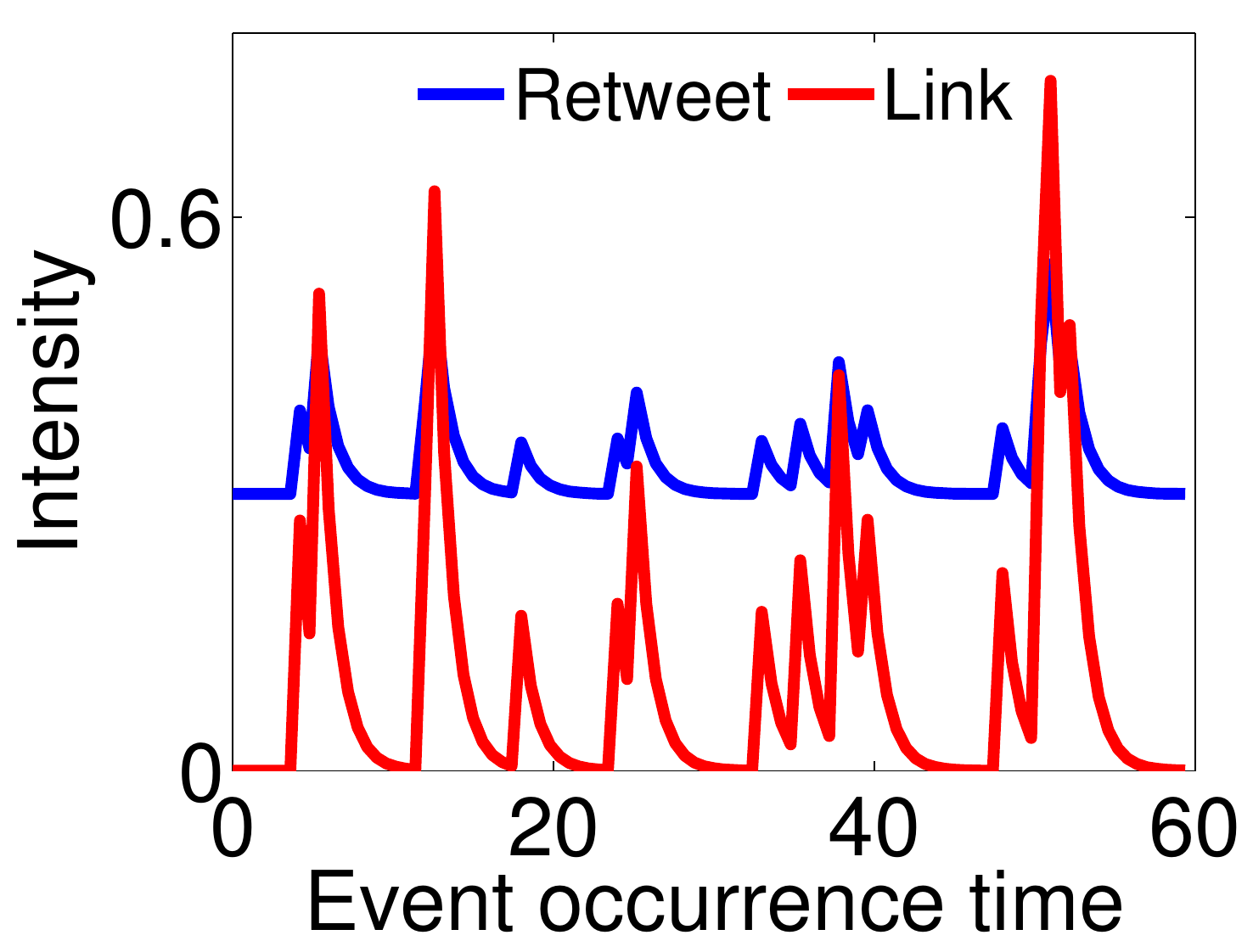} & 
          \includegraphics[width=0.24\textwidth]{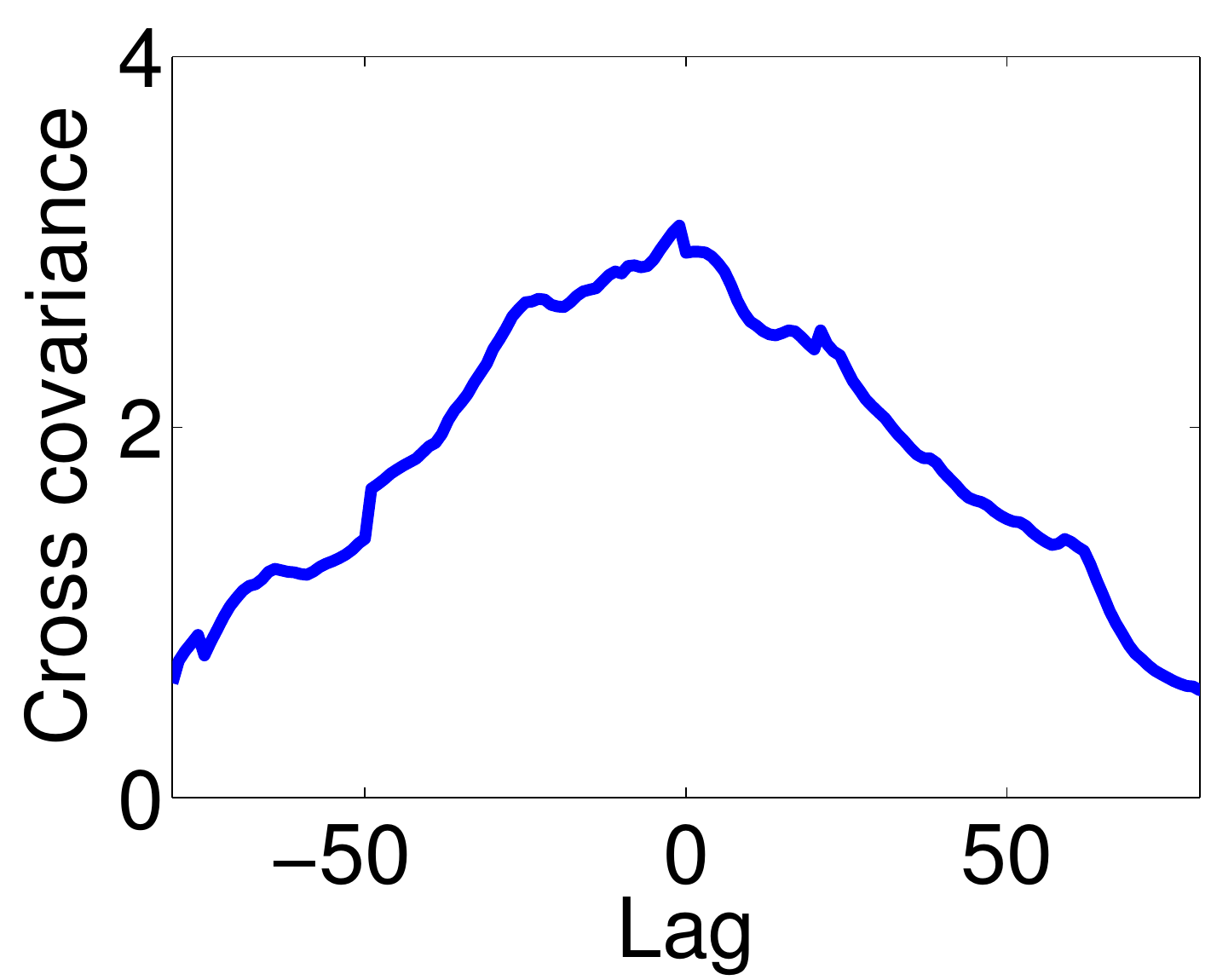} \\
          (a) & (b) & (c)
        \end{tabular}
        \caption{Coevolutionary dynamics for synthetic data. a) Spike trains of link and retweet events. b) Link and retweet intensities. c) Cross covariance of link and retweet intensities. }
        \label{fig:syntethic_coevolution}
\end{figure}

\subsection{Degree Distribution} Empirical studies have shown that the degree distribution of online social networks and microblogging sites follow a power law~\cite{ChaZhaFal2004,KwaLeeParMoo10}, and argued that it
is a consequence of the rich get richer phenomena. 
The degree distribution of a network is a power law if the expected number of nodes $m_d$ with degree $d$ is given by $m_d \propto d^{-\gamma}$, where $\gamma > 0$.  
Intuitively, the higher the values of the parameters $\alpha$ and $\beta$, the closer the resulting degree distribution follows a power-law. This is because the network grows more locally. Interestingly, the lower their values, the closer the distribution to an Erdos-Renyi random graph~\cite{ErdRen60}, because, the edges are added almost uniformly and independently without influence from the local structure. 
Figure~\ref{fig:degree-beta-alpha-varying} confirms this intuition by showing the degree distribution for different values of $\beta$ and $\alpha$. 
\begin{figure}[t]
        \centering
        \begin{tabular}{c c c c} 
         \includegraphics[width=0.20\textwidth]{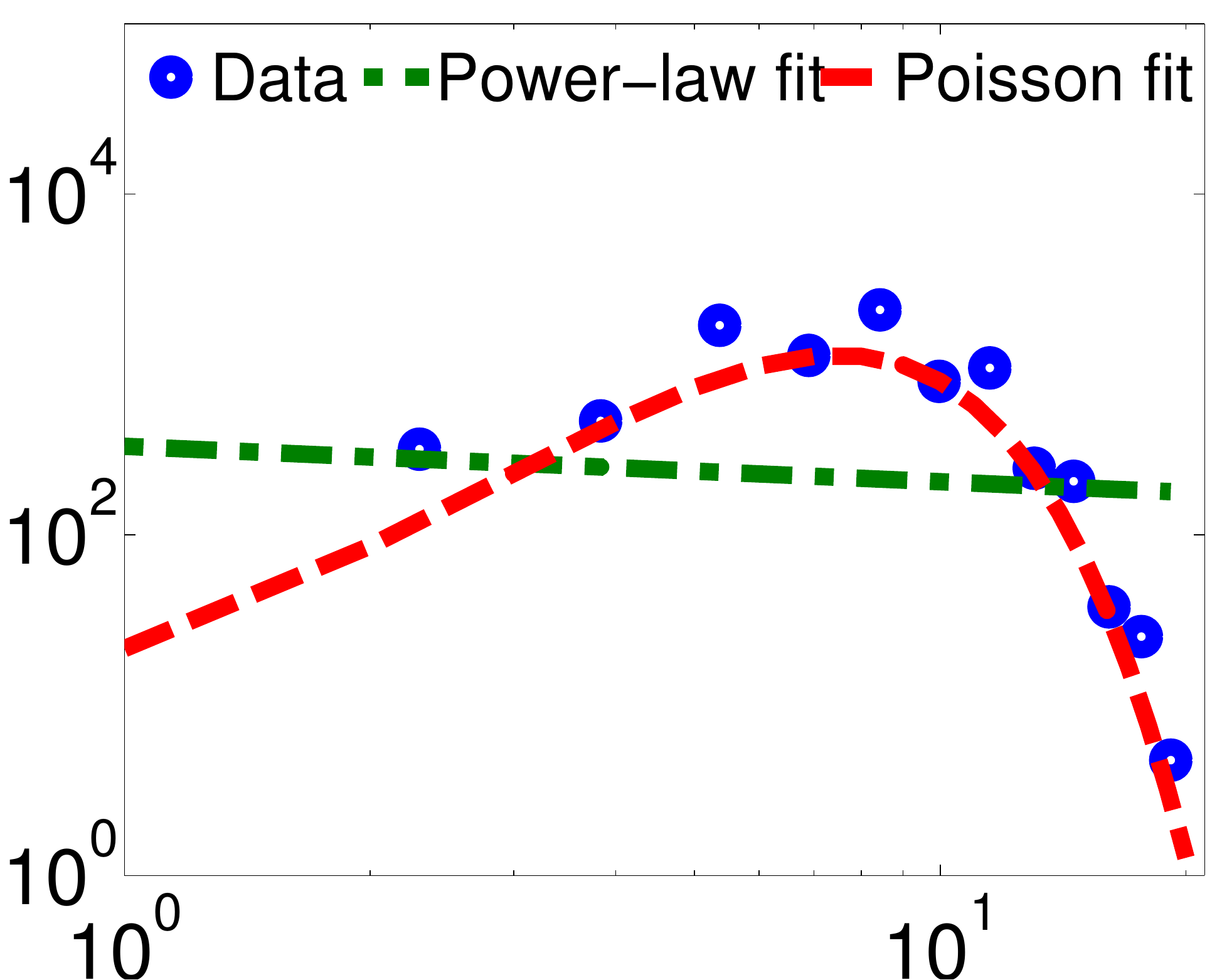} &
          \includegraphics[width=0.20\textwidth]{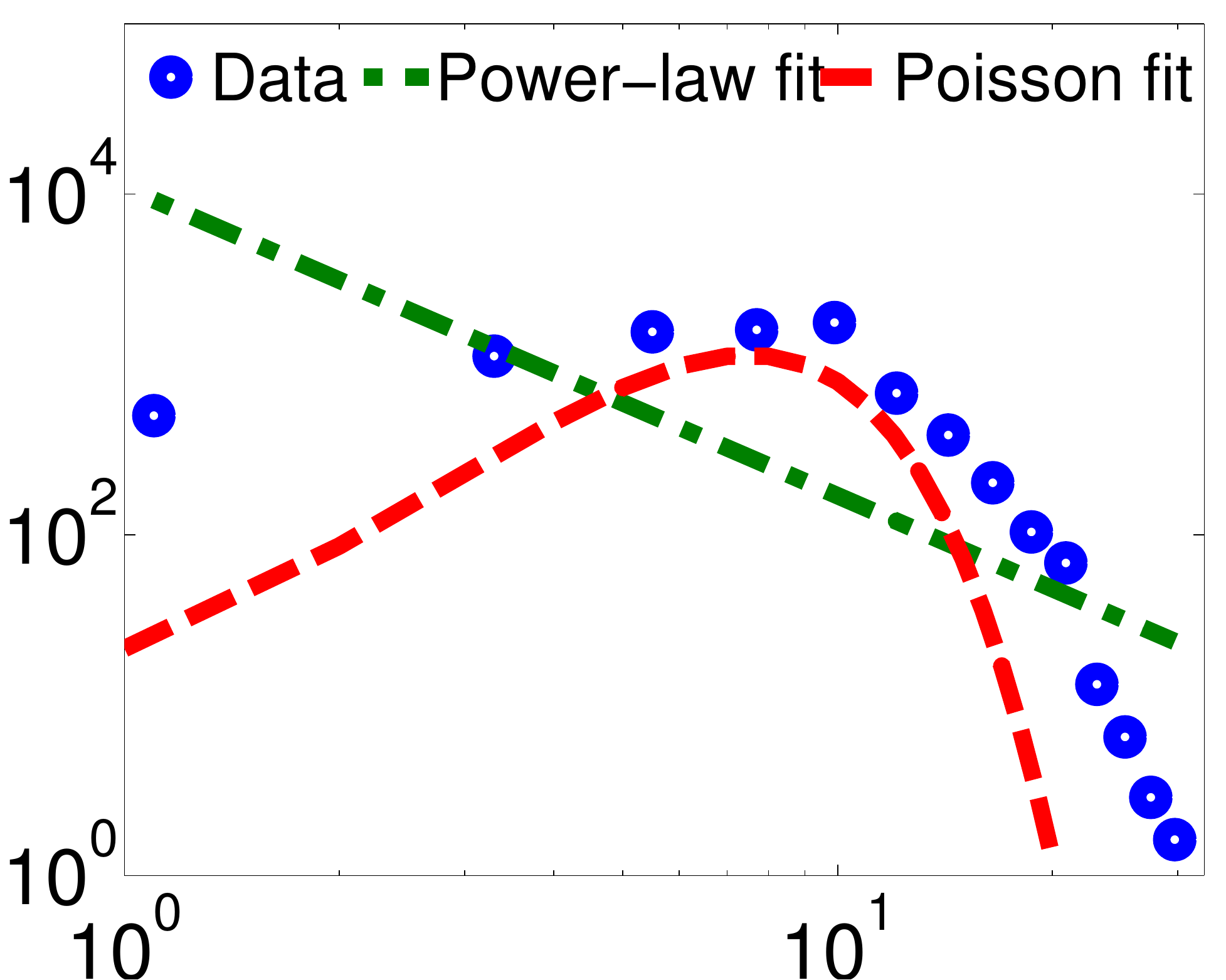} &  
          \includegraphics[width=0.20\textwidth]{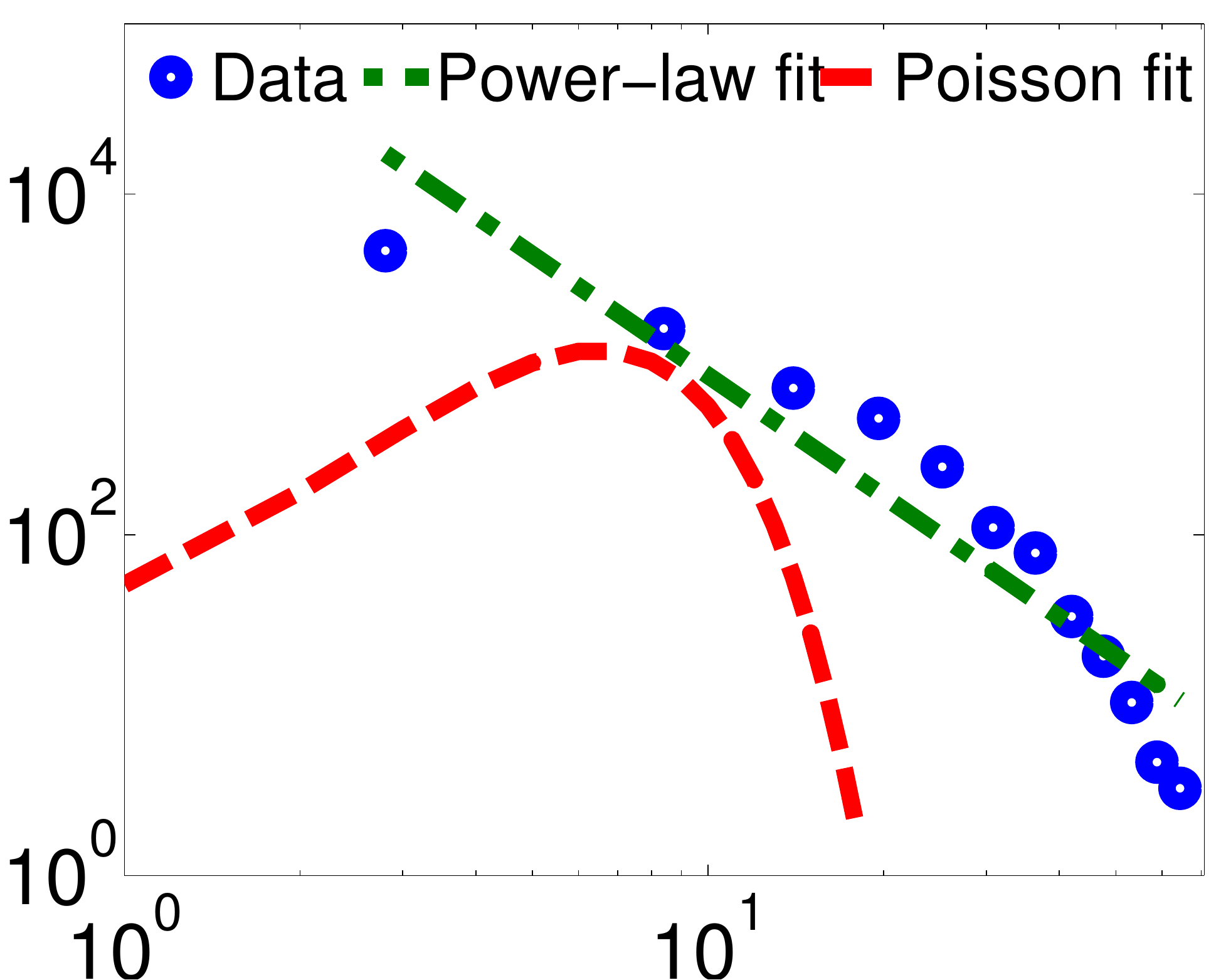} &
          \includegraphics[width=0.20\textwidth]{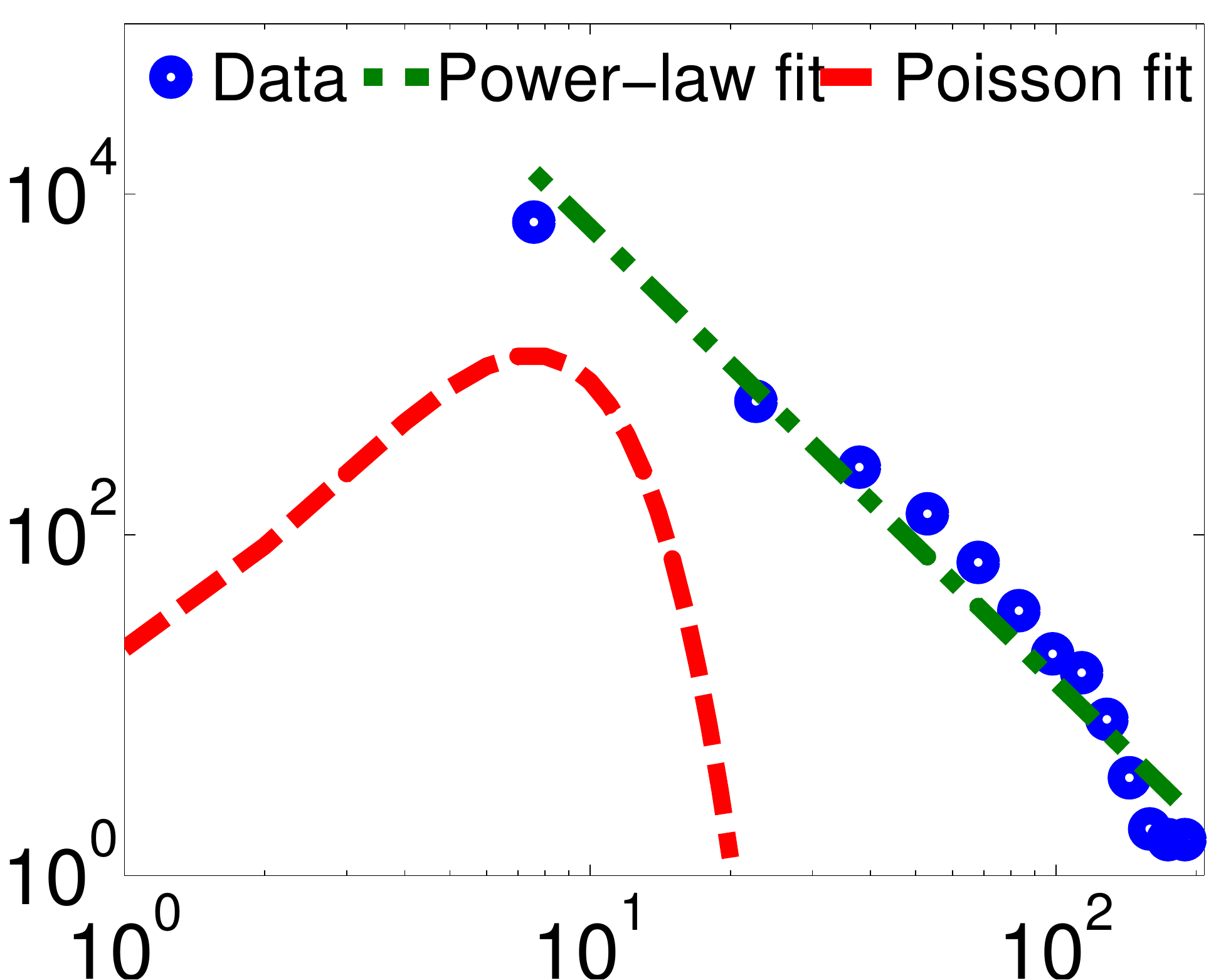} \\
          (a) $\beta = 0$ & (b) $\beta = 0.001$ & (c) $\beta = 0.1$ & (d) $\beta = 0.8$  \\
          \includegraphics[width=0.22\textwidth]{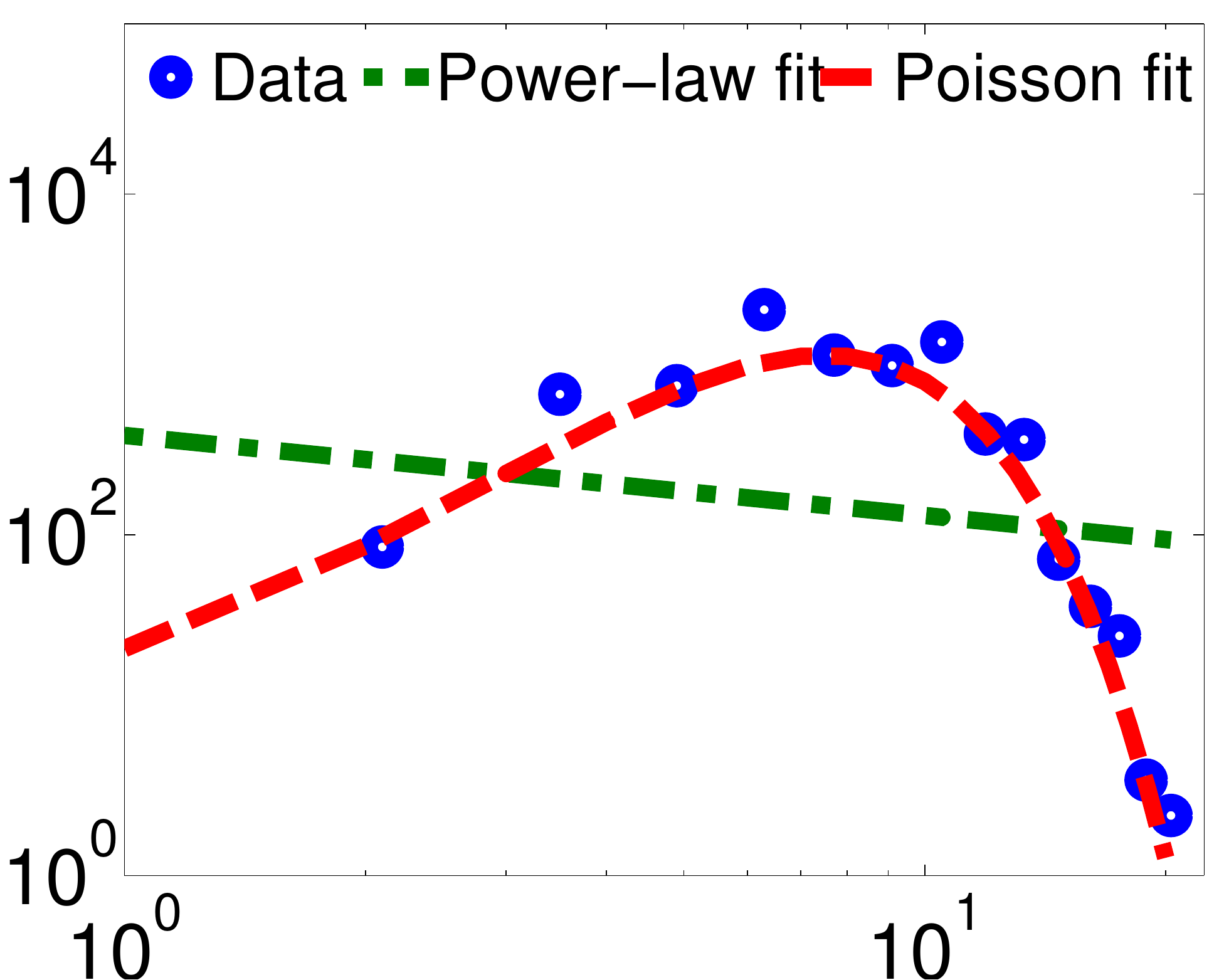} &
          \includegraphics[width=0.22\textwidth]{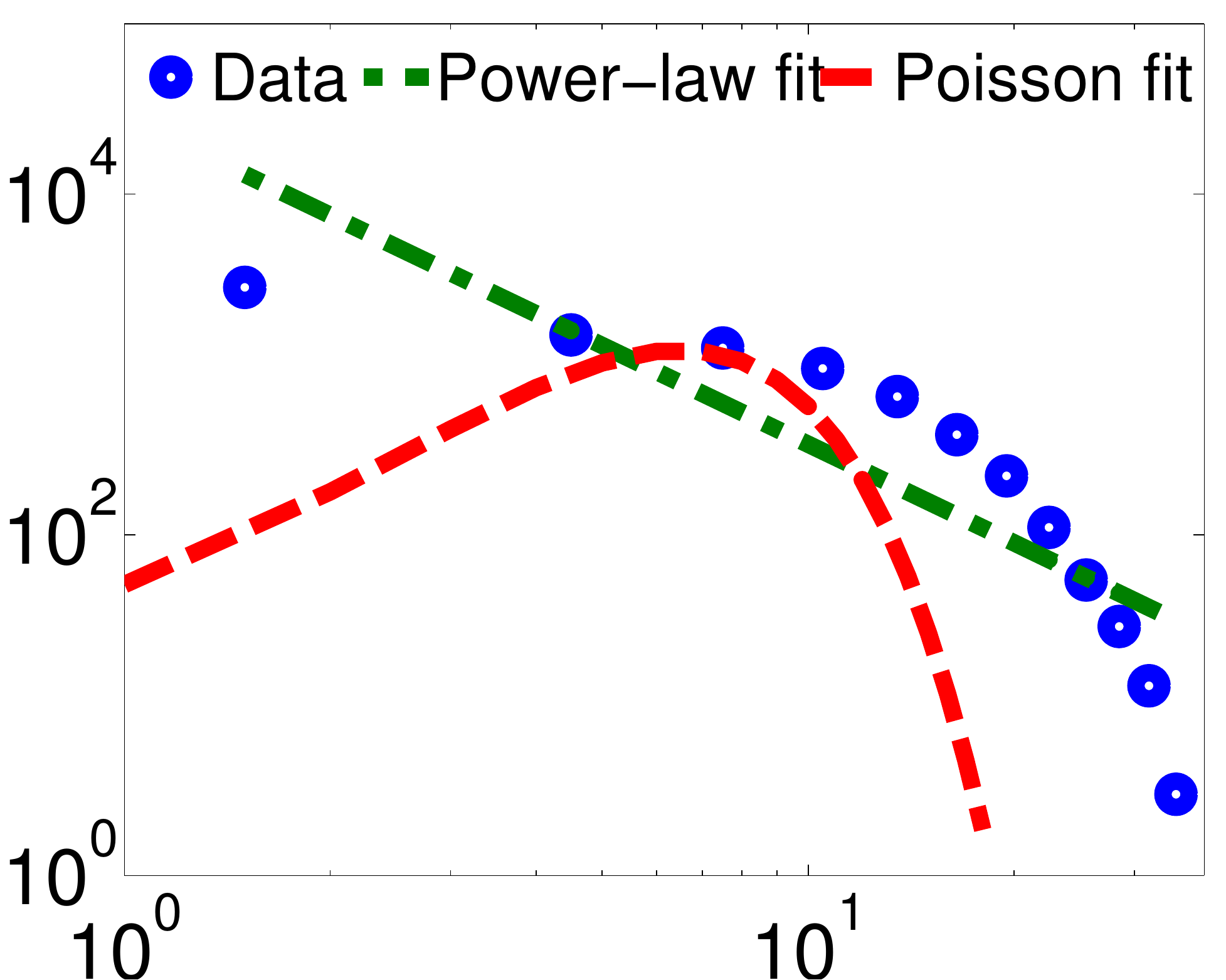} &
          \includegraphics[width=0.22\textwidth]{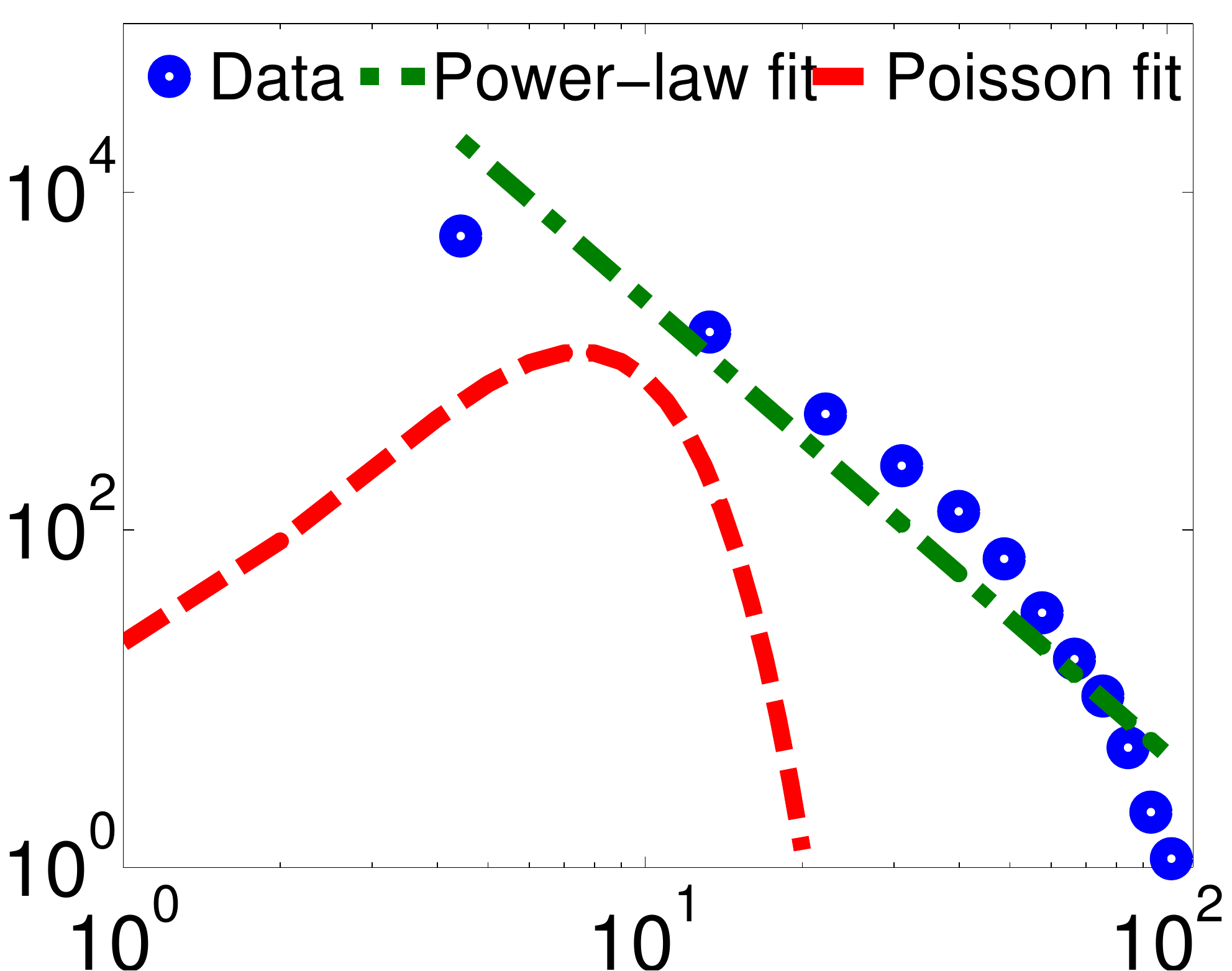}&
          \includegraphics[width=0.22\textwidth]{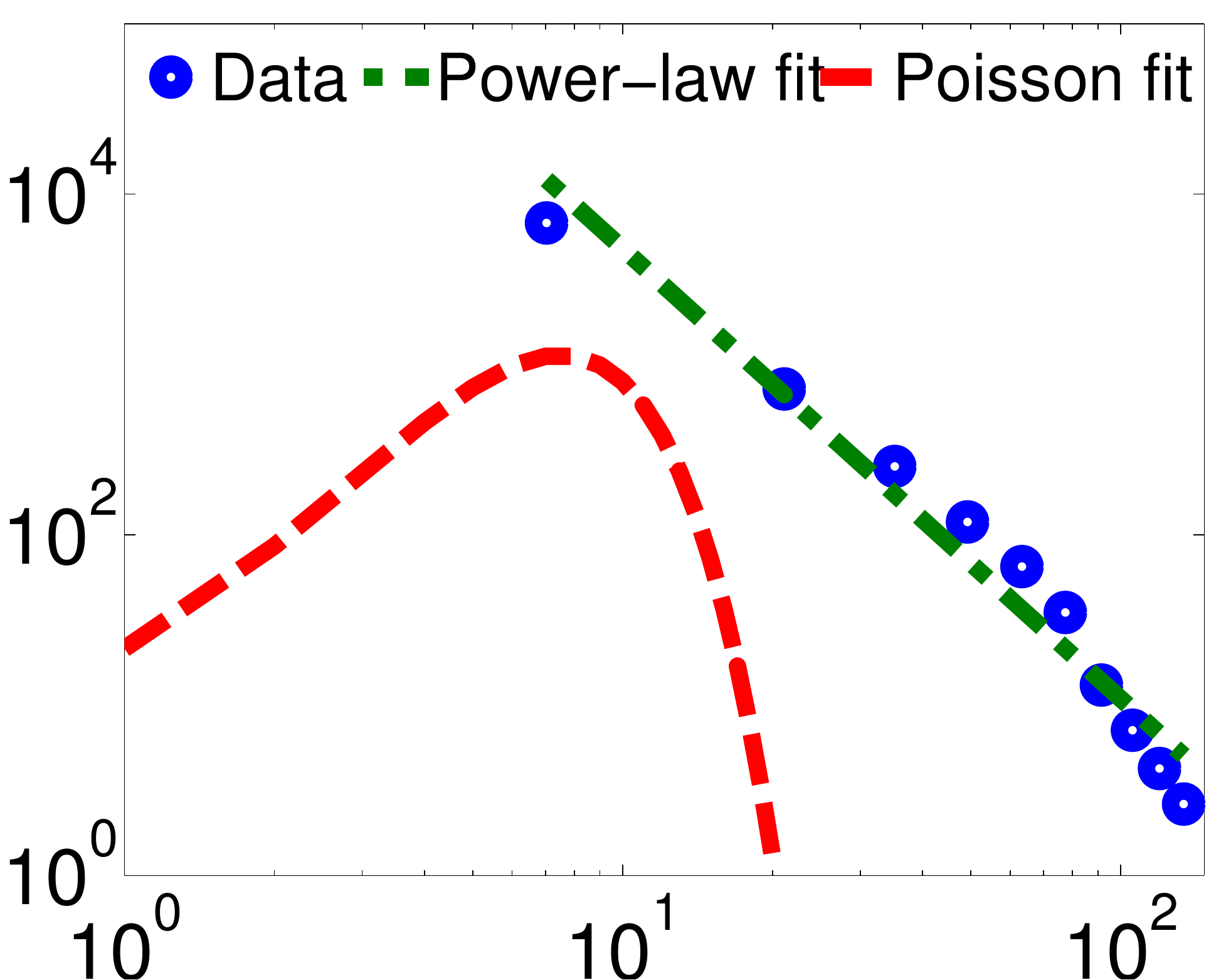}\\     
          (a) $\alpha = 0$ &  (b) $\alpha = 0.05$ & (c) $\alpha = 0.1$ & (d) $\alpha = 0.2$\\
        \end{tabular}
        \caption{Degree distributions when network sparsity level reaches 0.001 for different $\beta$ ($\alpha$) values and fixed $\alpha=0.1$ ($\beta=0.1$).}
        \label{fig:degree-beta-alpha-varying}
\end{figure}

\subsection{Small (shrinking) Diameter}
There is empirical evidence that the diameter of online social networks and microblogging sites exhibit relatively small diameter and shrinks (or flattens) as the network grows~\cite{BacBolRosUgaVig12,ChaZhaFal2004,LesKleFal05}.
Figures~\ref{fig:dens-shrinking-beta-varying-clust-coef}(a-b) show the diameter on the largest connected component (LCC) against the sparsity of the network over time for different values of $\alpha$ and 
$\beta$.
Although at the beginning, there is a short increase in the diameter due to the merge of small connected components, the diameter decreases as the network evolves. 
Moreover, larger values of $\alpha$ or $\beta$ lead to higher levels of local growth in the network and, as a consequence, slower shrinkage.
%
%
Here, nodes \emph{arrive} to the network when they follow (or are followed by) a node in the largest connected component.

\subsection{Clustering Coefficient}
Triadic closure~\cite{Gra73,LesBacKumTom08,RomKle10} has been often presented as a plausible link creation me\-cha\-nism. However, different social networks and microblogging sites present diffe\-rent
levels of triadic closure~\cite{UgaBacKle13}.
Importantly, our method is able to generate networks with diffe\-rent le\-vels of triadic closure, as shown by Figure~\ref{fig:dens-shrinking-beta-varying-clust-coef}(c-d), where we plot the clustering coefficient~\cite{WatStr98}, 
which is proportional to the frequency of triadic closure, for different values of $\alpha$ and $\beta$.
%
%
\begin{figure}[t]
  \centering
  \begin{tabular}{c c c c}
      \includegraphics[width=0.21\textwidth]{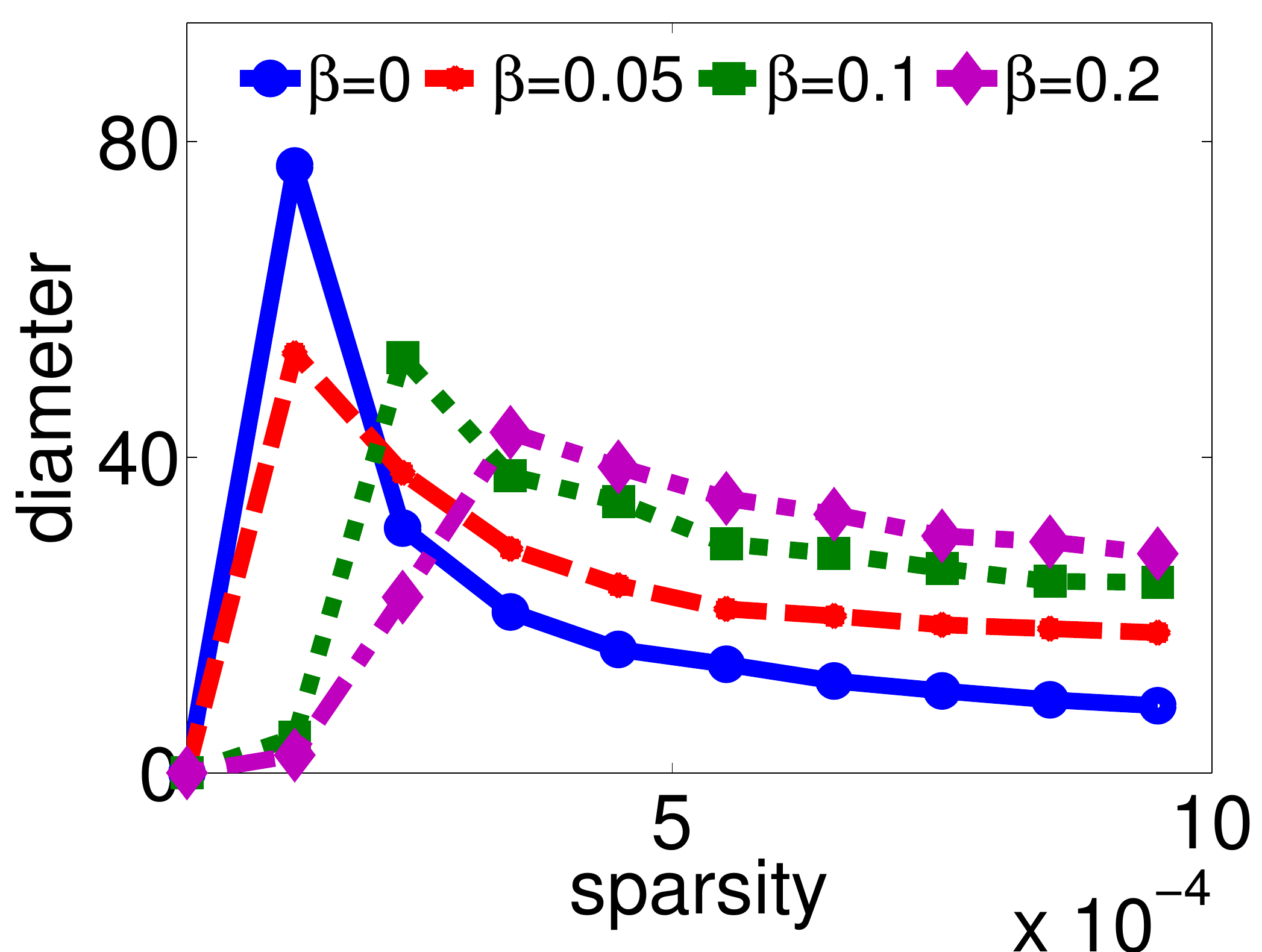}&
      \includegraphics[width=0.21\textwidth]{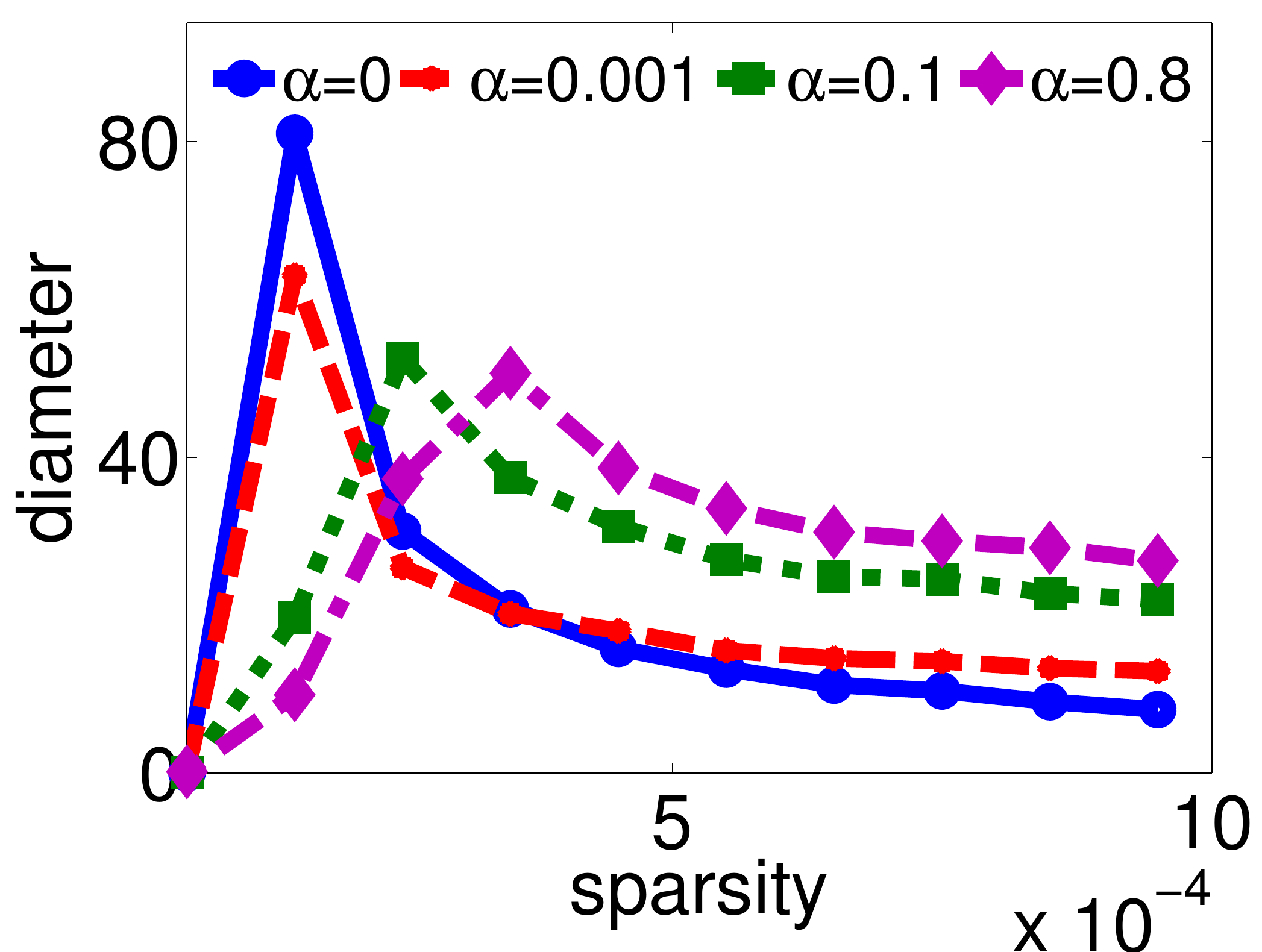} & 
      \includegraphics[width=0.21\textwidth]{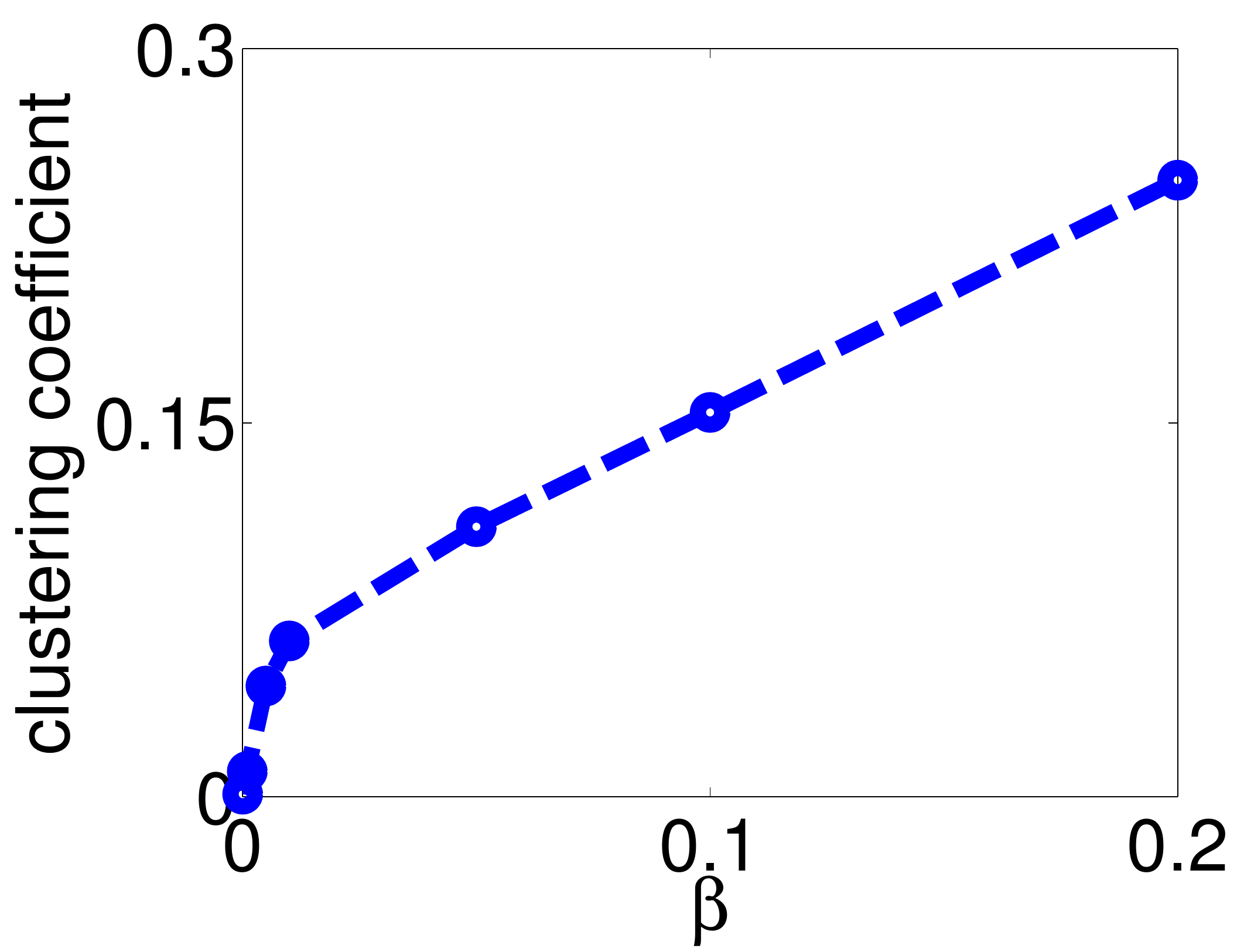} &
      \includegraphics[width=0.21\textwidth]{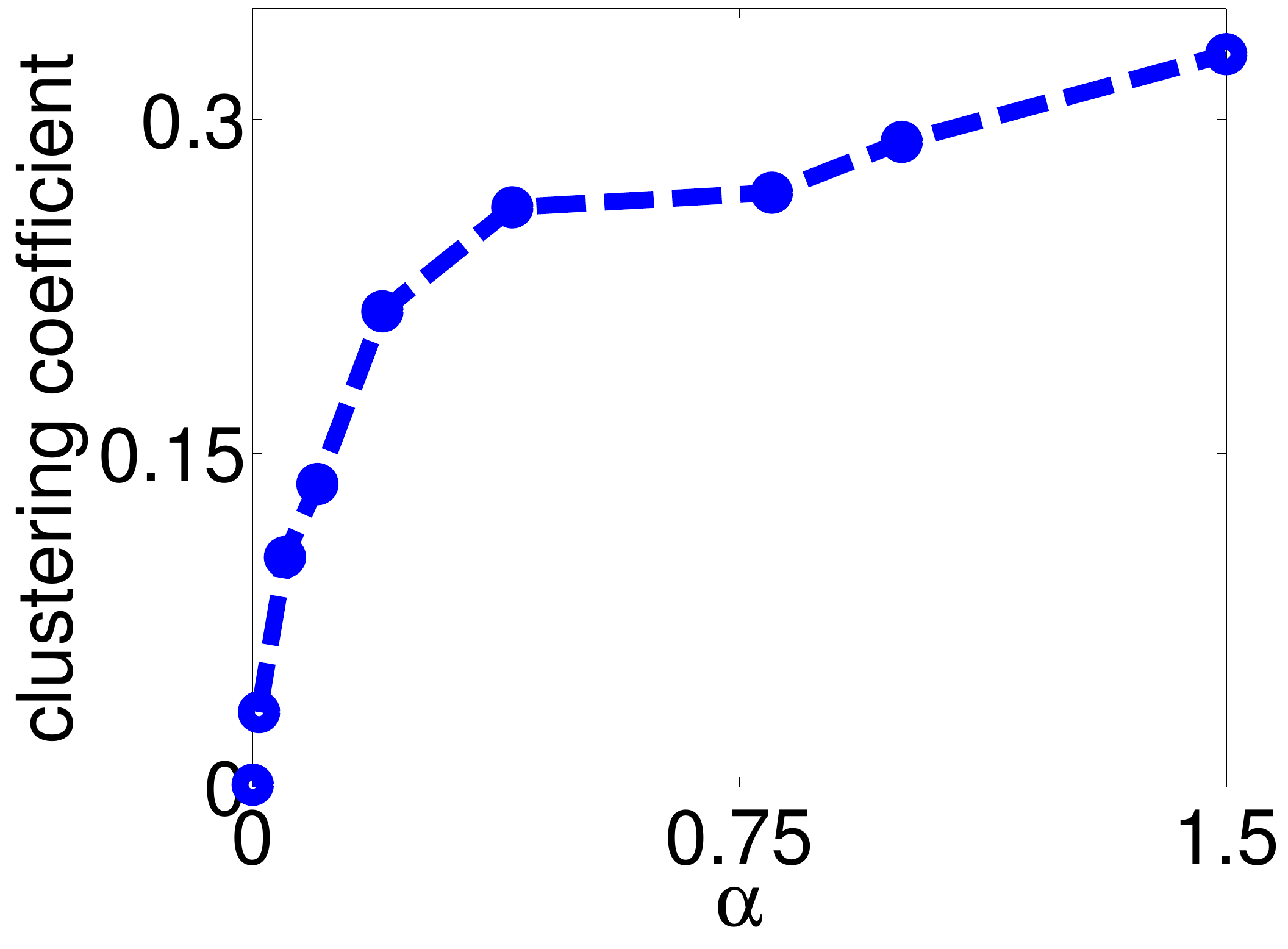}\\
      (a) Diameter, $\alpha=0.1$ & (b) Diameter, $\beta=0.1$ & (c) CC, $\alpha=0.1$ & (d) CC, $\beta=0.1$ \\
  \end{tabular}
  \caption{Diameter and clustering coefficient for network sparsity 0.001. Panels (a) and (b) show the diameter against sparsity over time for fixed $\alpha=0.1$, and for fixed $\beta = 0.1$ respectively. Panels (c) and (d) show the clustering coefficient (CC) against $\beta$ and $\alpha$, respectively.} 
  \label{fig:dens-shrinking-beta-varying-clust-coef}
\end{figure}

%
%
%
%

\subsection{Network Visualization}
Figure~\ref{fig:vis-beta} visualizes several snapshots of the largest connected component (LCC) of two 300-node networks for two particular realizations of our model, under two 
different values of $\beta$. In both cases, we used $\mu= 2 \times 10^{-4}$, $\alpha=1$, and $\eta=1.5$.
The top two rows correspond to $\beta=0$ and represent one end of the spectrum, $\ie$, Erdos-Renyi random network. Here,
the network evolves uniformly.
The bottom two rows correspond to $\beta=0.8$ and represent the other end, $\ie$, scale-free networks. Here, the network
evolves locally, and clusters emerge naturally as a consequence of the local growth.
They are depicted using a combination of forced directed and Fruchterman Reingold layout with Gephi\footnote{http://gephi.github.io/}.
Moreover, the figure also shows the retweet events (from others as source) for two nodes, $A$ and $B$, on the bottom row. These two nodes arrive almost at the same time and establish links to two other 
nodes. However, node $A$'{}s followees are more central, therefore, $A$ is being exposed to more retweets. Thus, node $A$ performs more retweets than $B$ does. It again shows how information diffusion is affected by network structure.
%
%
%
%
%
%
%
Overall, this figure clearly illustrates that by careful choice of parameters we can generate networks with a very different structure. 
\begin{figure} [!!h]
        \centering
       \begin{tabular}{c c c}   
         \includegraphics[width=0.202\textwidth]{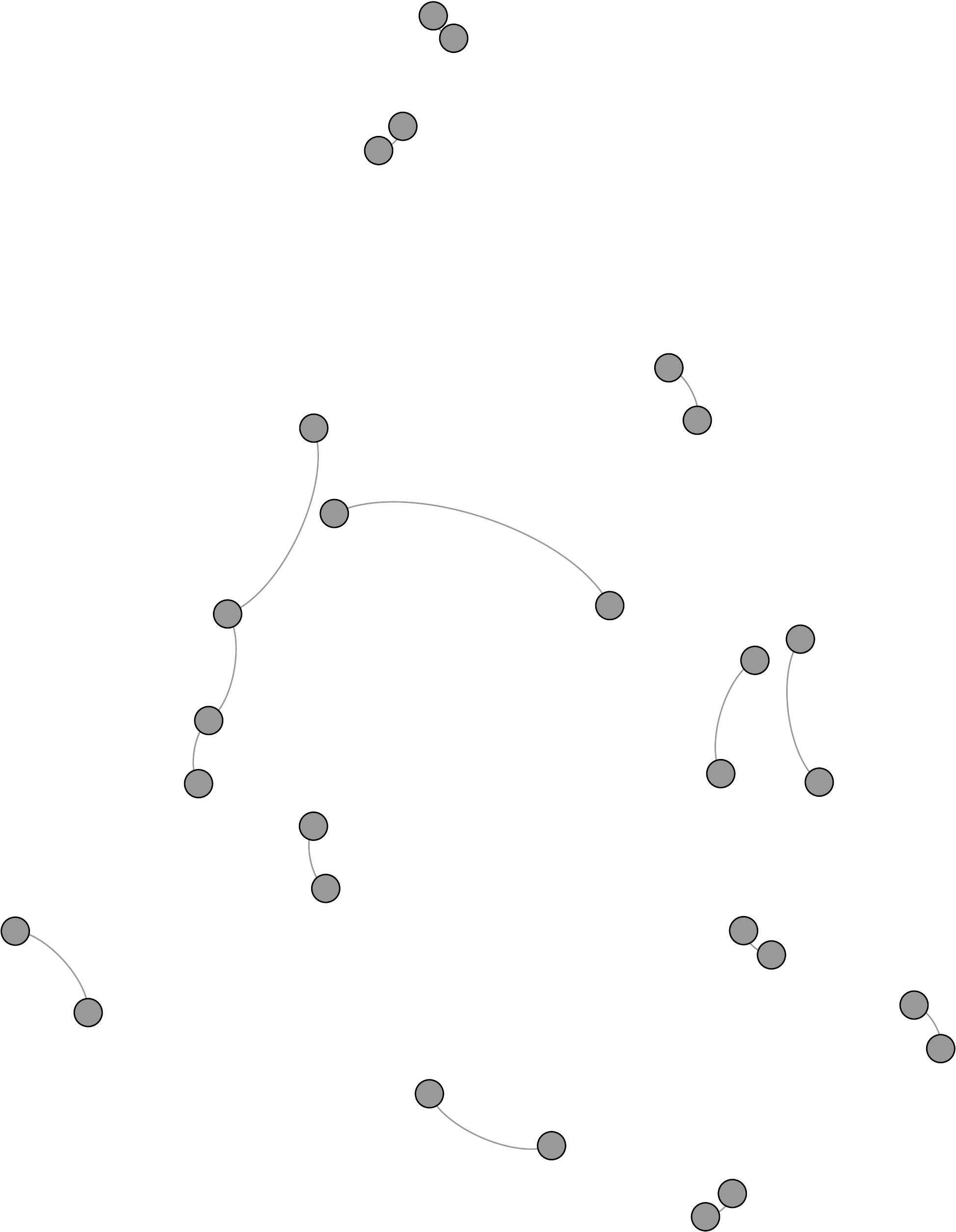} &
         \includegraphics[width=0.202\textwidth]{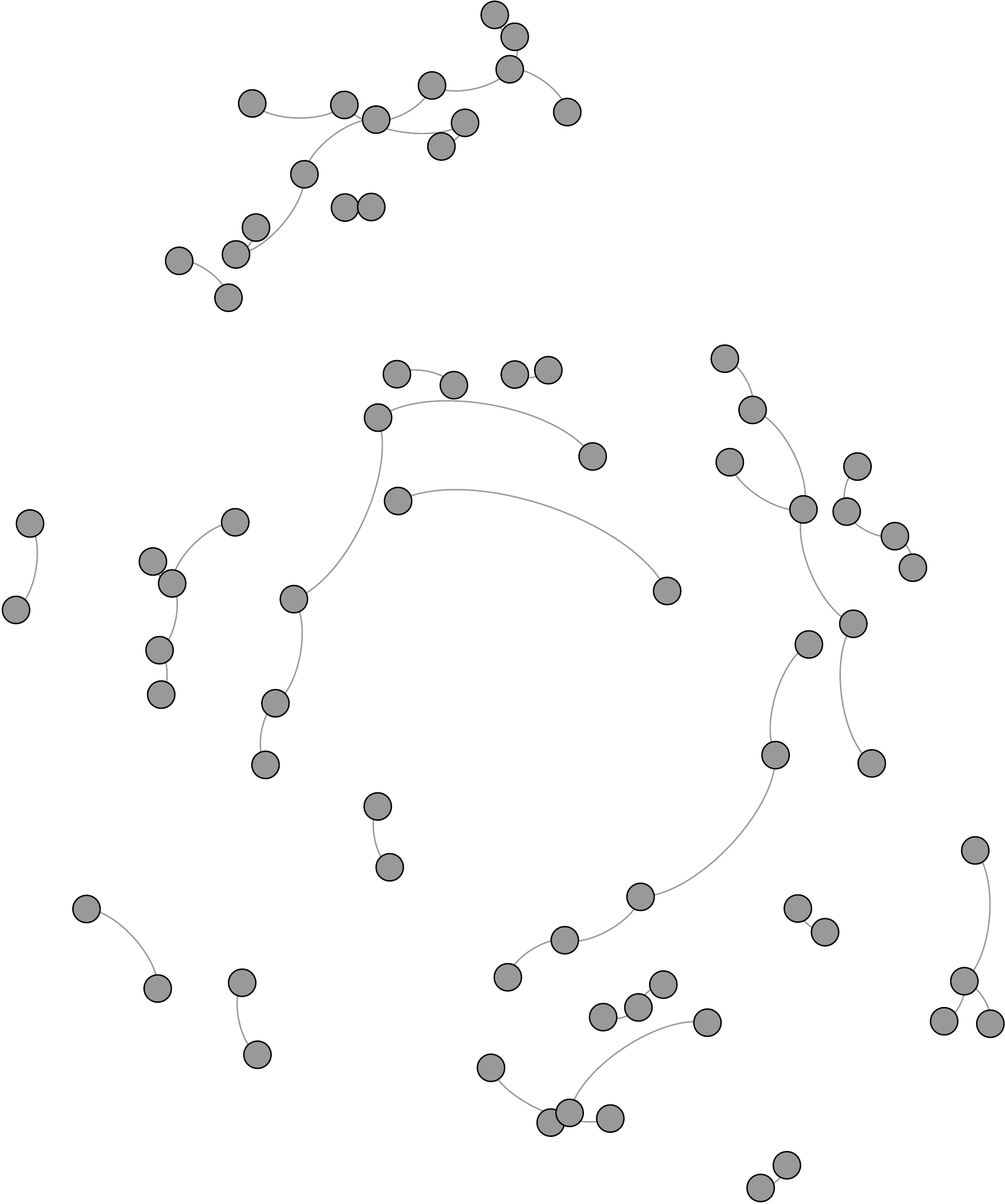} &
         \includegraphics[width=0.202\textwidth]{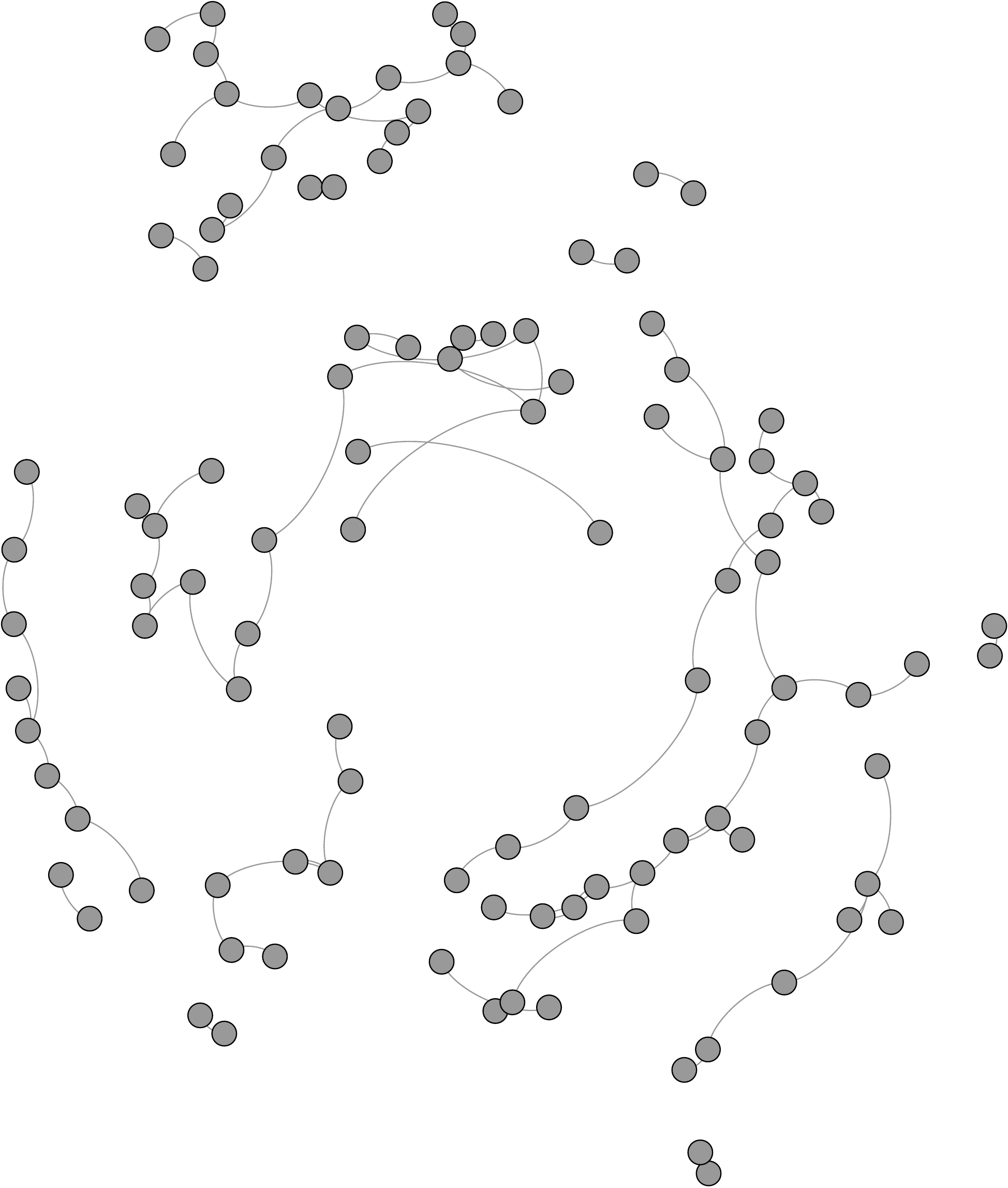}   \\
         t = 5 & t=20 & t=35 \vspace{1mm} \\
         \includegraphics[width=0.202\textwidth]{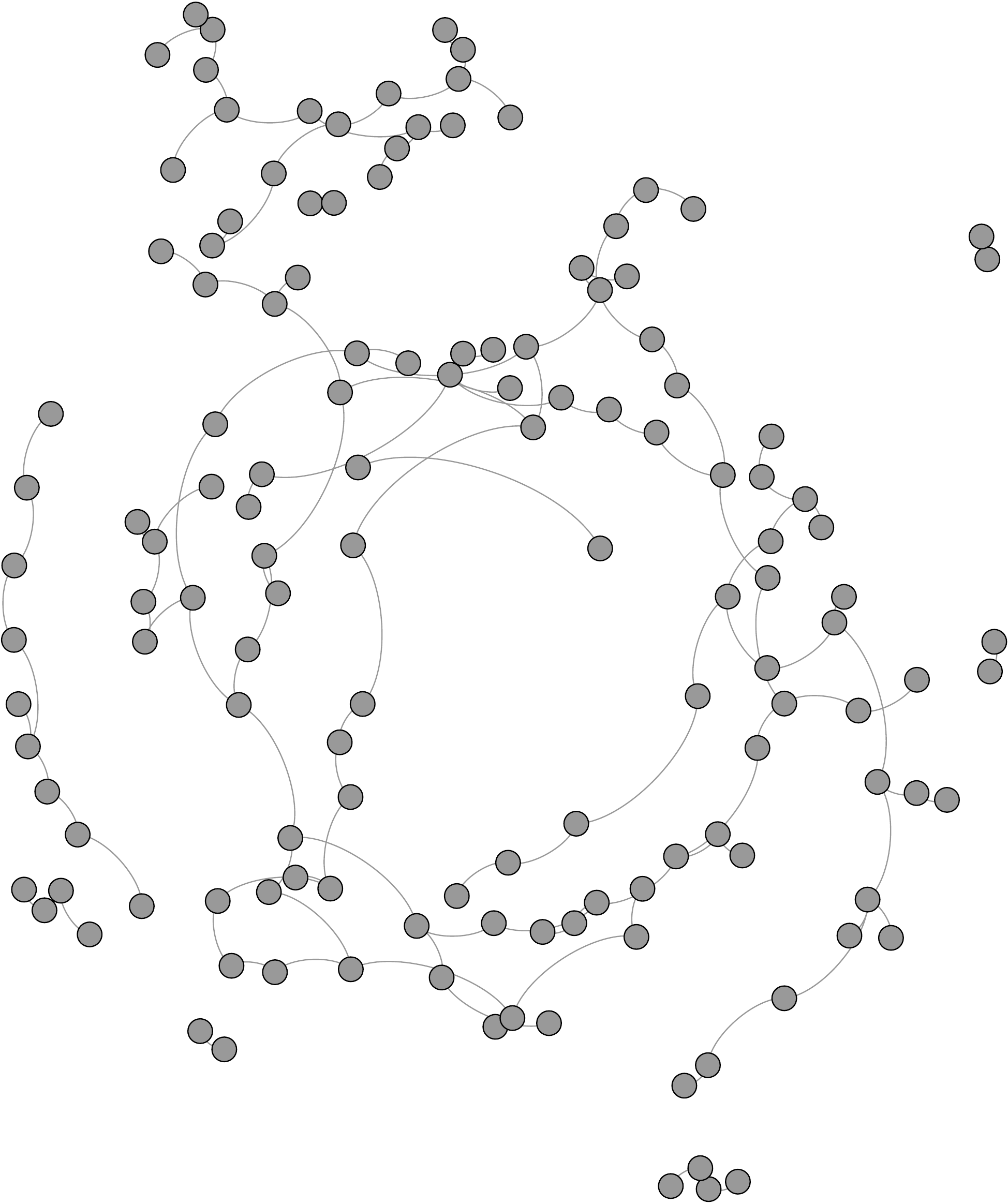} &
         \includegraphics[width=0.202\textwidth]{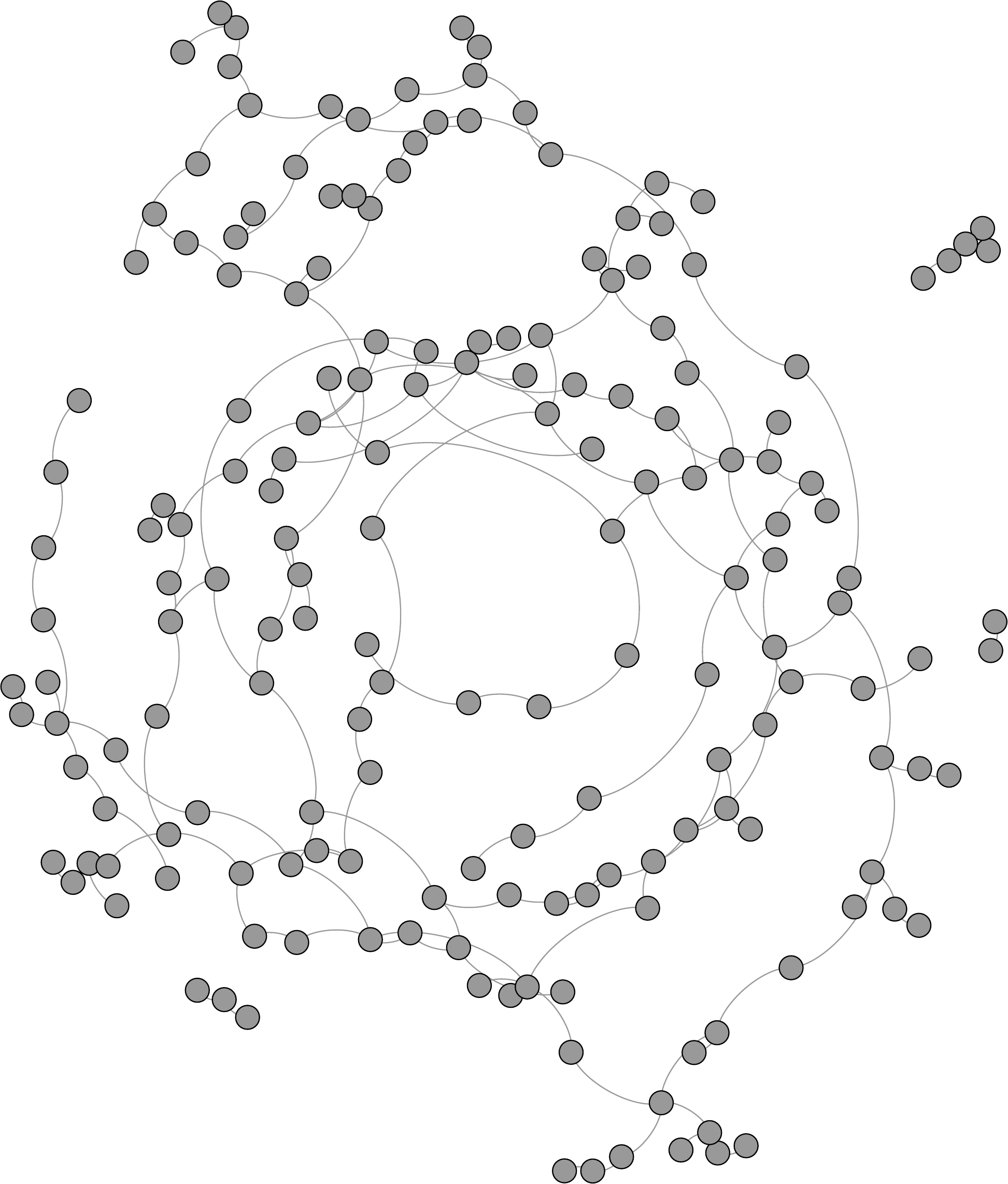} &
          \includegraphics[width=0.202\textwidth]{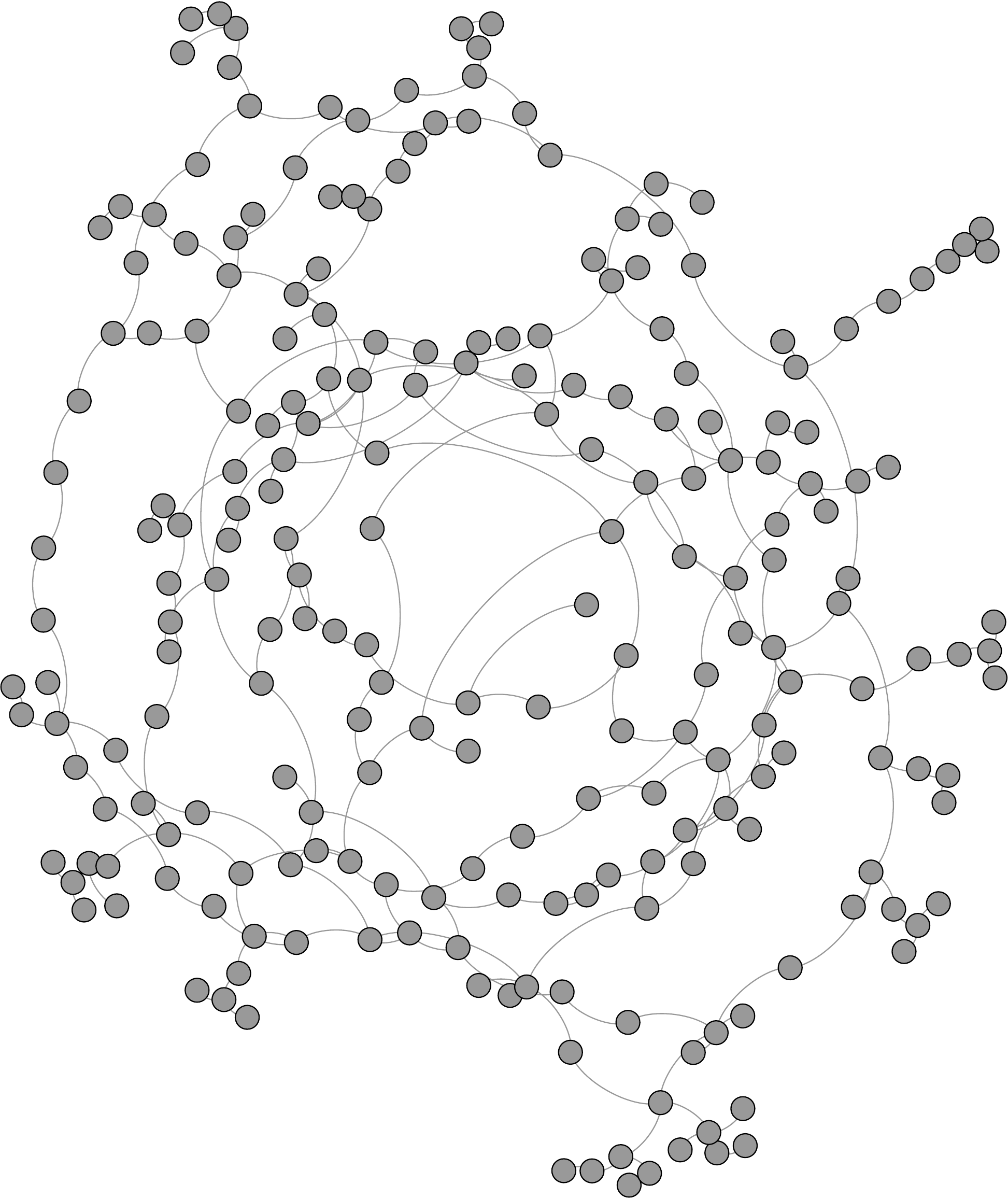} \\
         t=50 & t=65 & t=80 \vspace{1mm} \\
                 \includegraphics[width=0.202\textwidth]{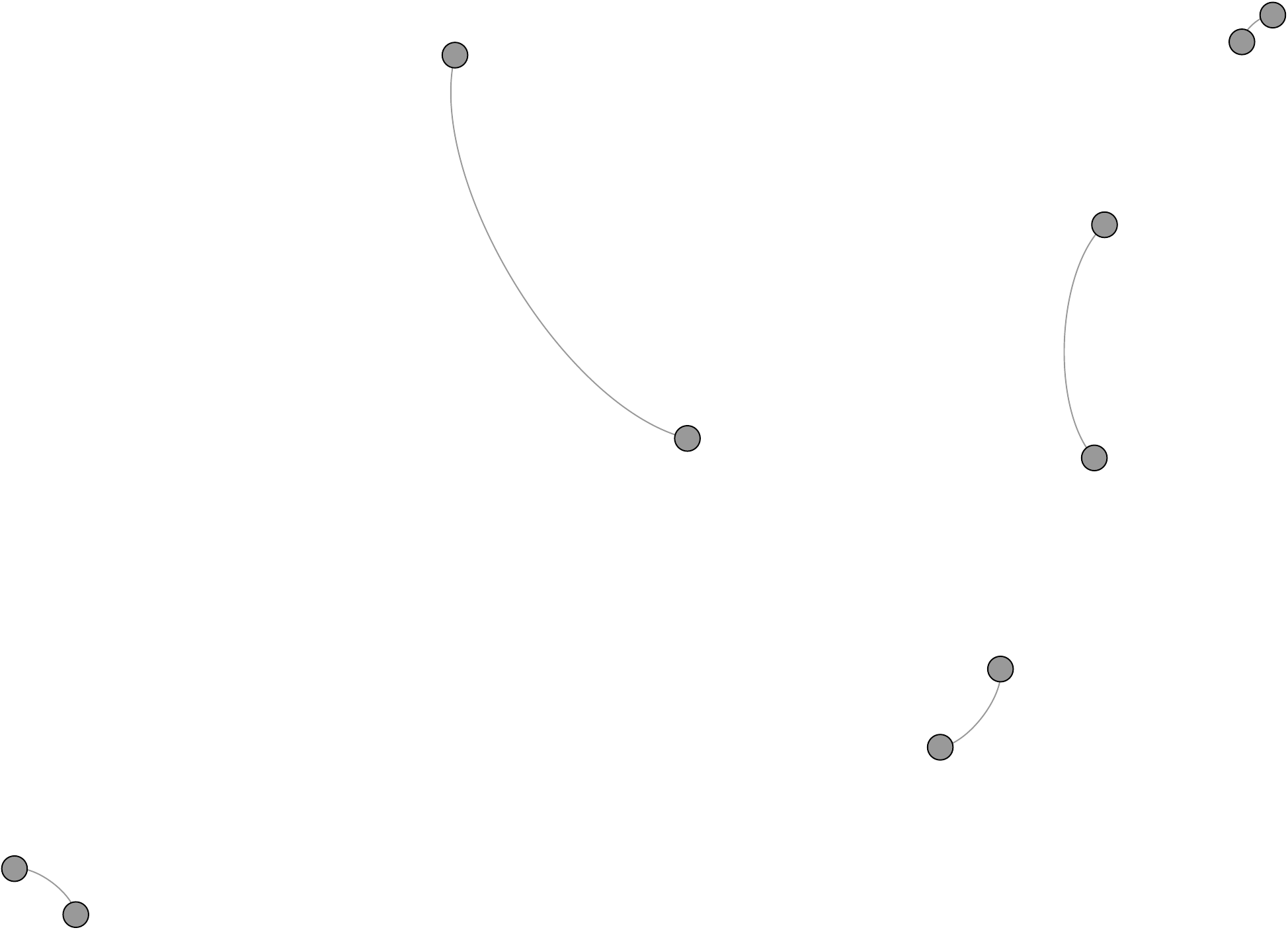} &
         \includegraphics[width=0.202\textwidth]{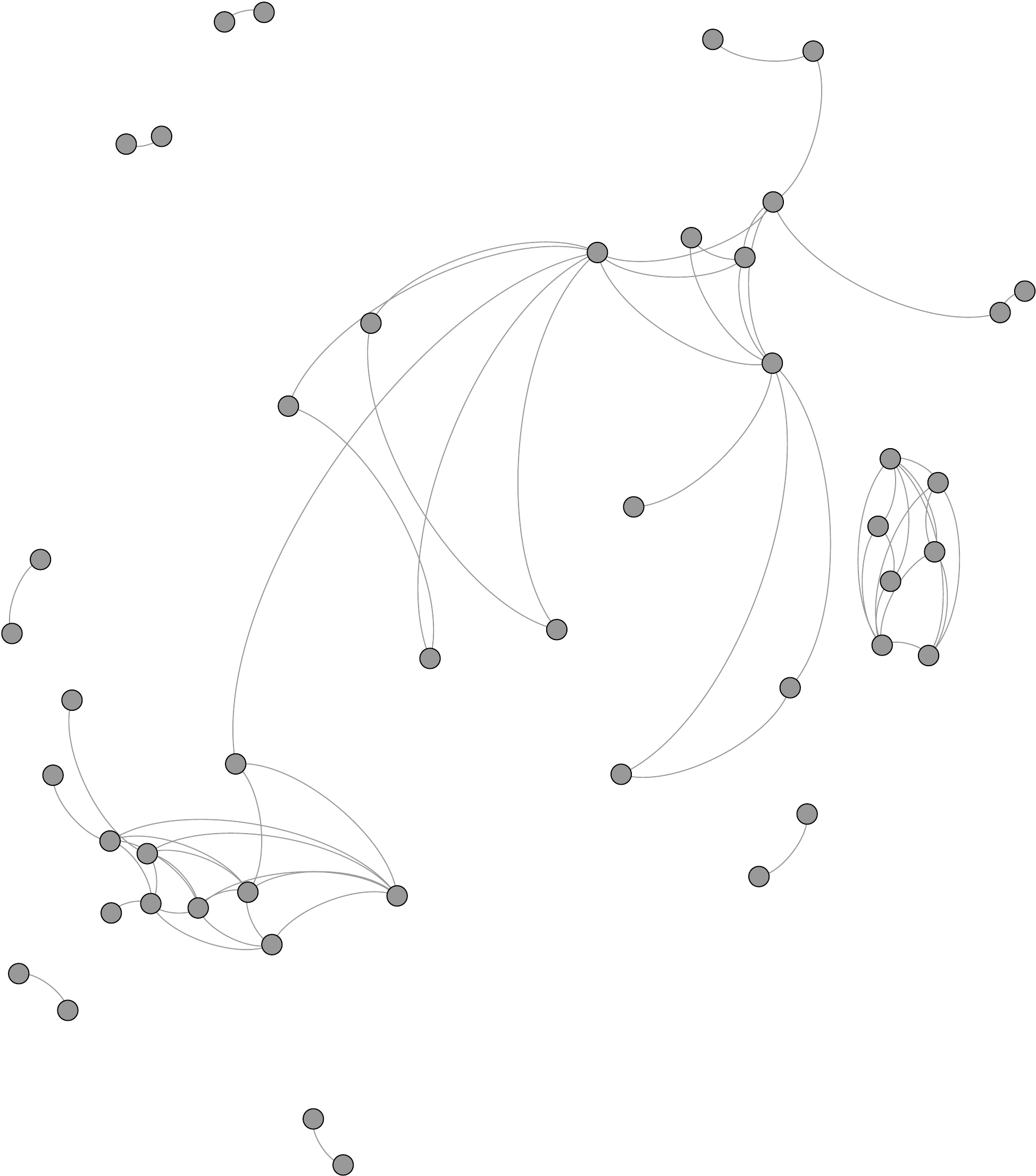} &
         \includegraphics[width=0.202\textwidth]{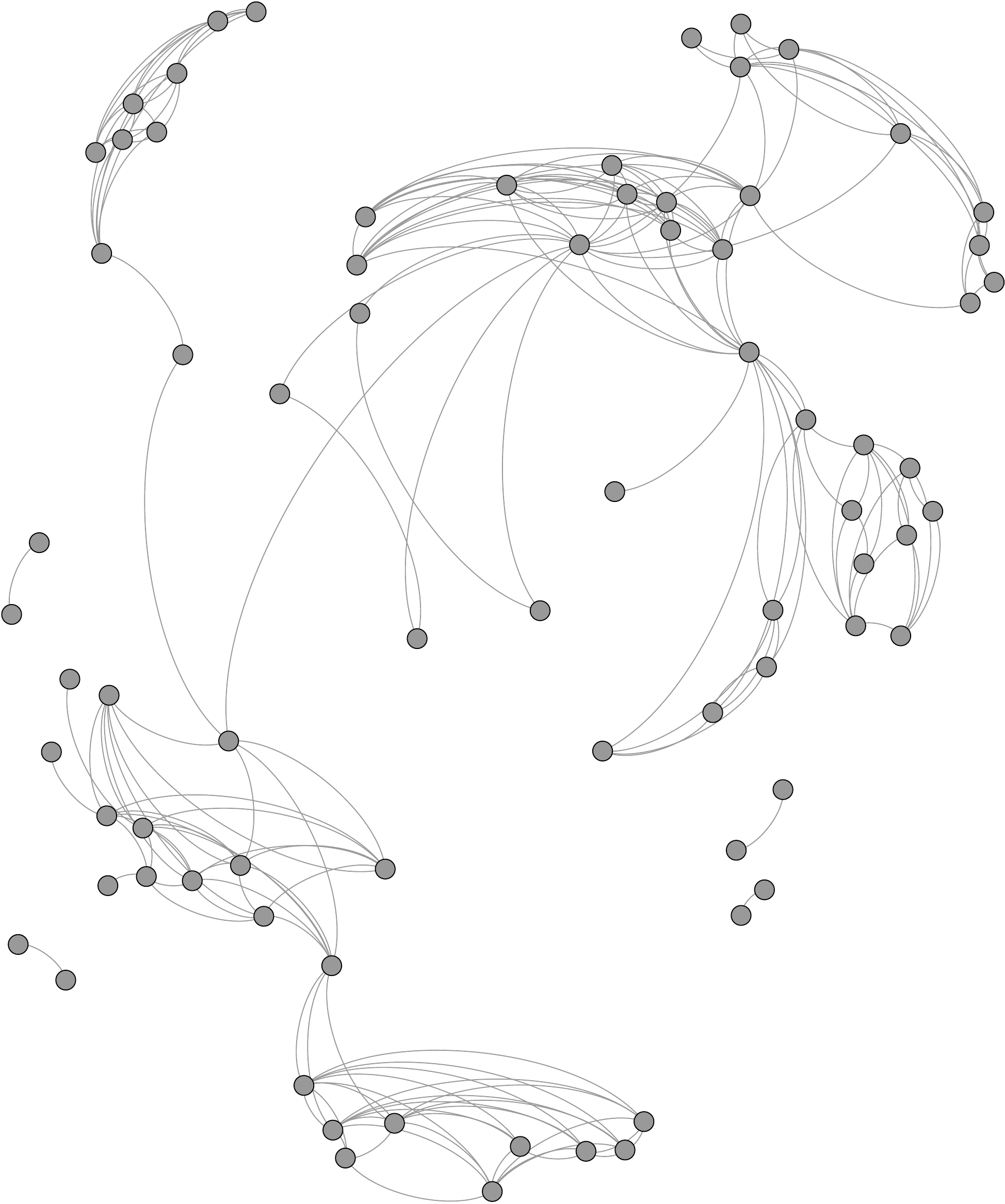}   \\
         t = 5 & t=20 & t=35 \vspace{1mm} \\
         \includegraphics[width=0.202\textwidth]{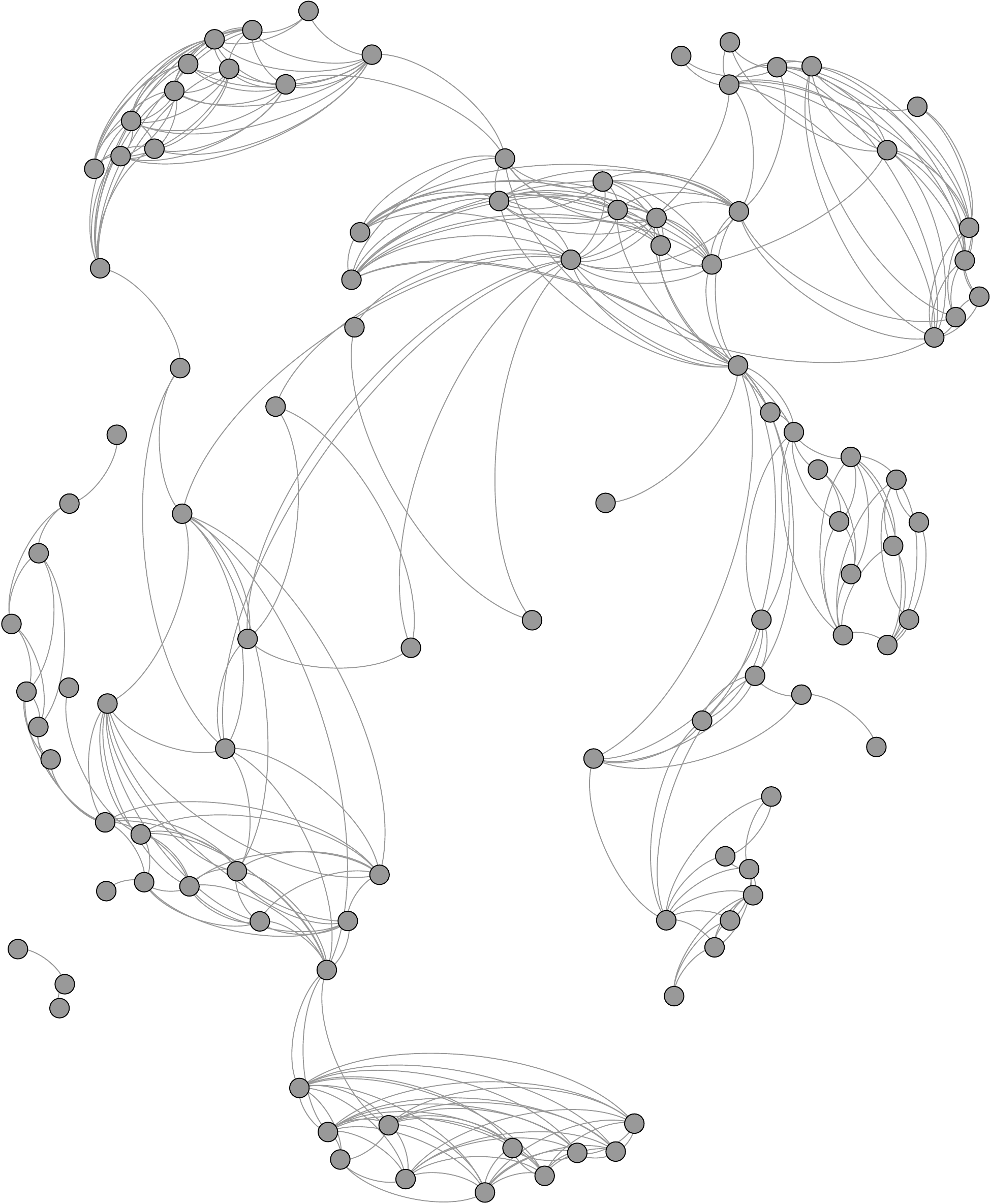} &
         \includegraphics[width=0.202\textwidth]{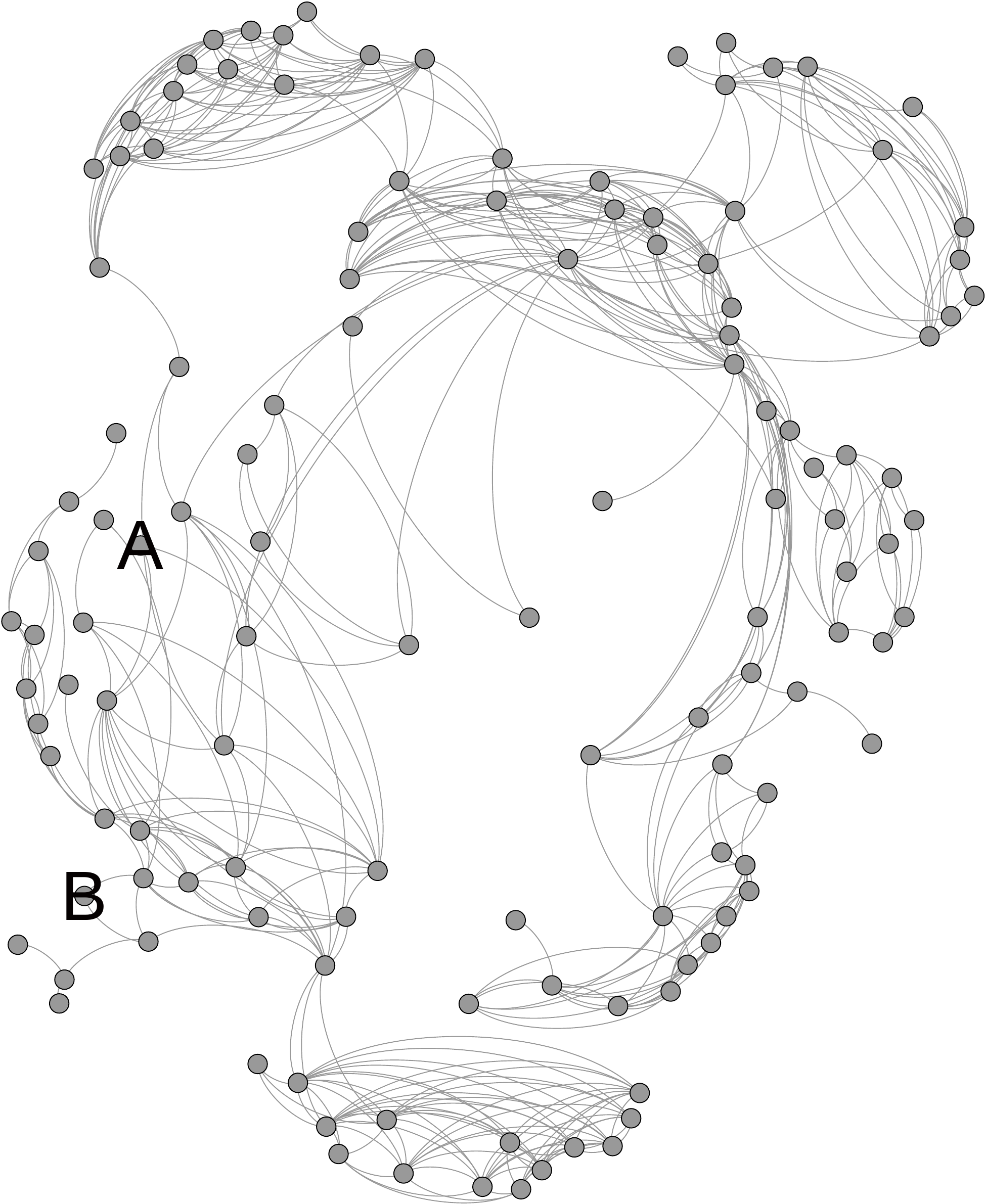} &
          \includegraphics[width=0.202\textwidth]{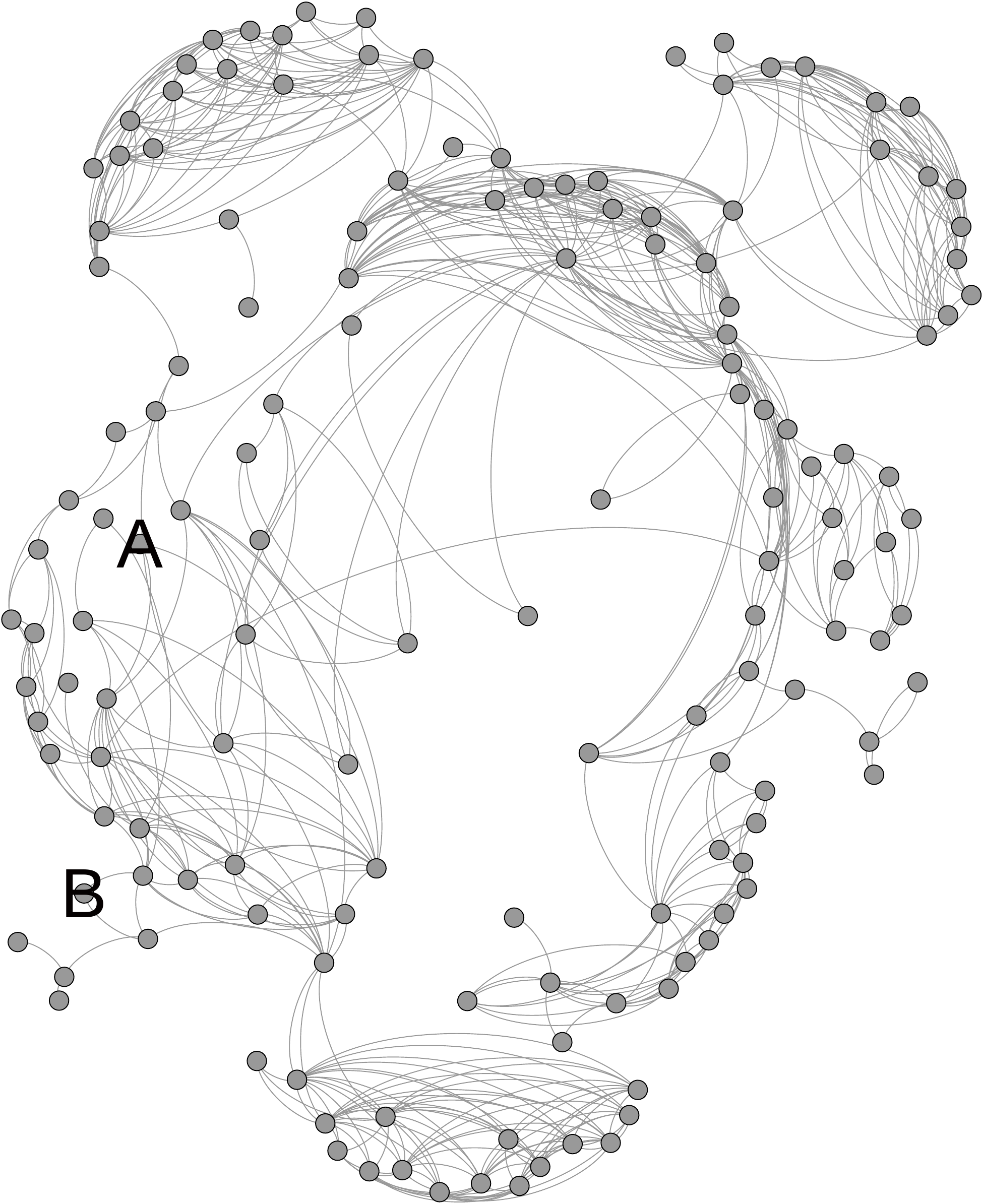} \\
         t=50 & t=65 & t=80 \vspace{1mm} \\
         &
         \includegraphics[width=0.18\textwidth]{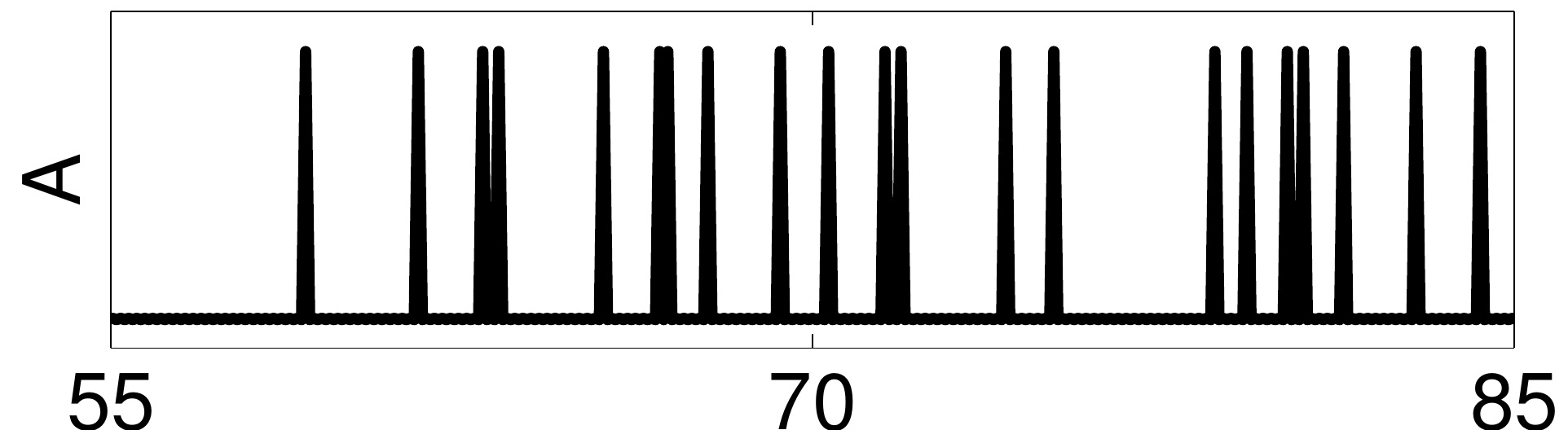}  & \includegraphics[width=0.18\textwidth]{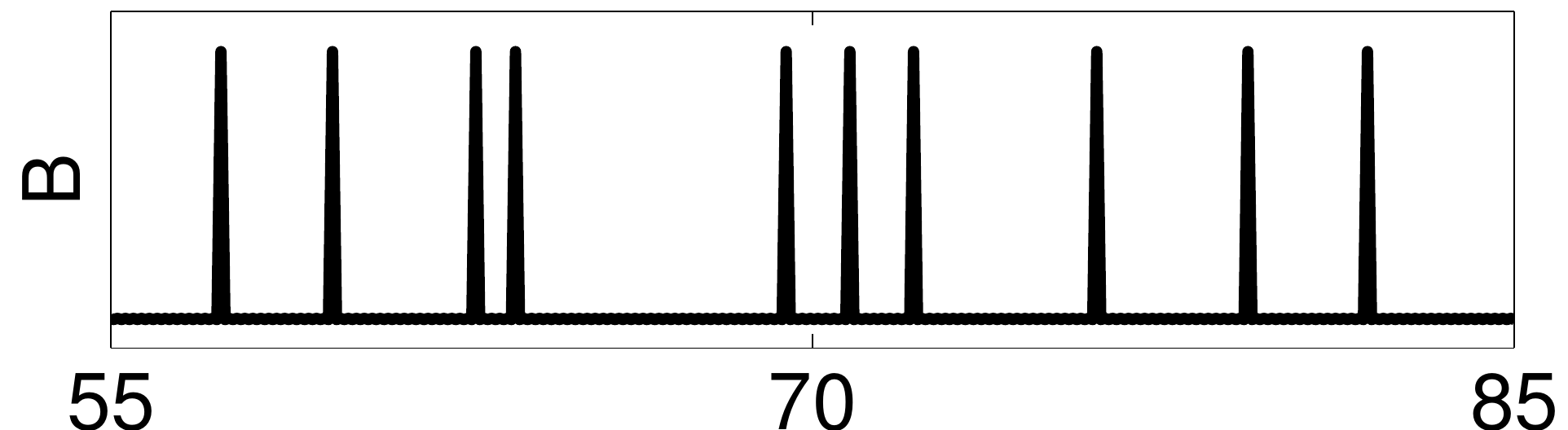}
          \end{tabular}
        \caption{Evolution of two networks: one with $\beta = 0$ (1st and 2nd rows) and another one with $\beta=0.8$ (3rd and 4th rows), and spike trains of nodes A and B (5th row).
}
        \label{fig:vis-beta}
\end{figure}

\begin{figure} [!!h]
        \vspace{-2mm}
        \centering
                         \includegraphics[width=0.76\textwidth]{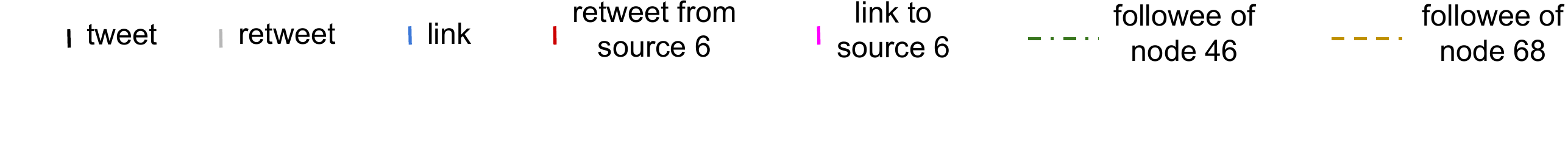} \\
                         \vspace{-8mm}
                 \includegraphics[width=.93\textwidth]{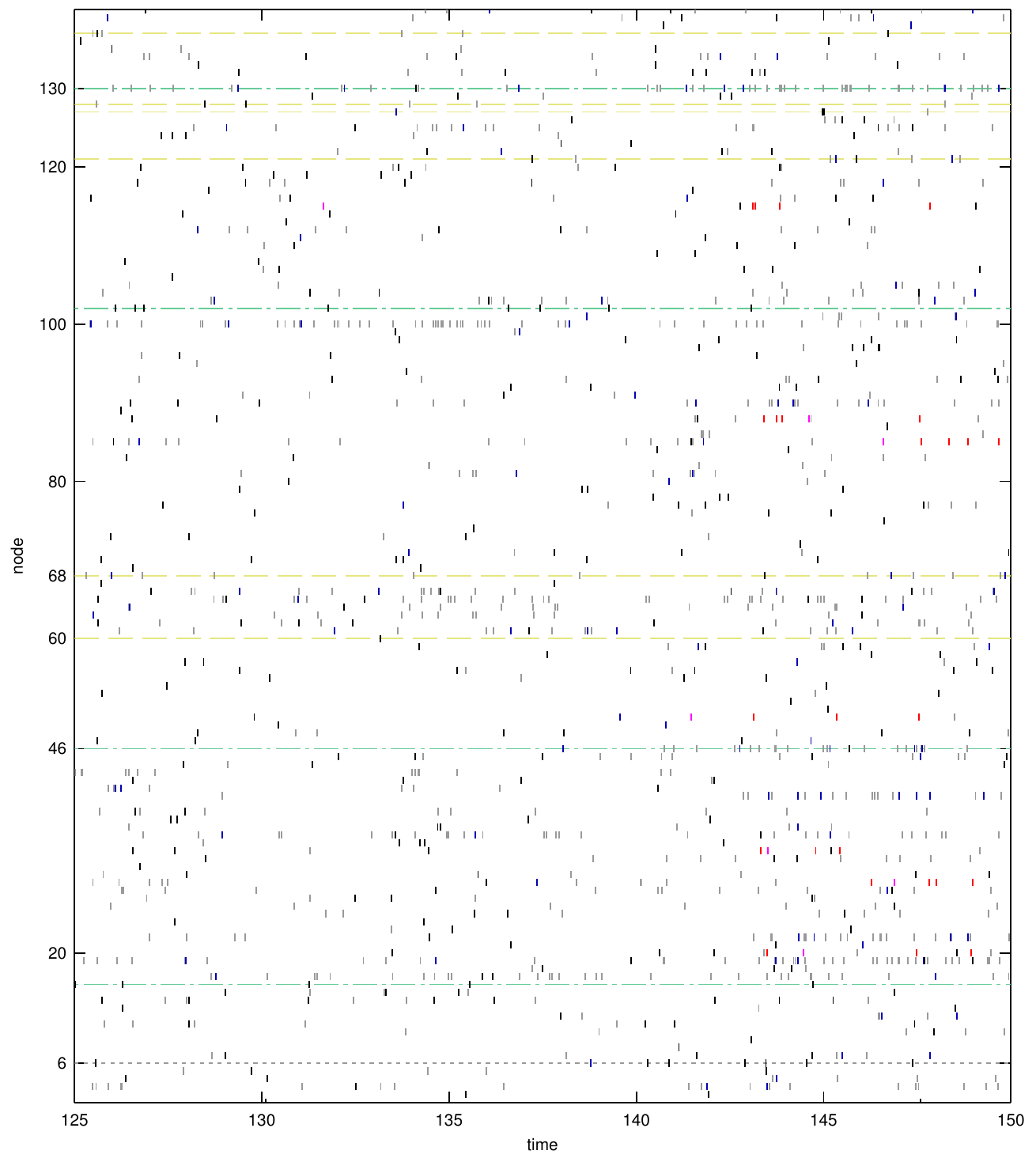}
        \caption{Coevolutionary dynamics of events for the network shown in Figure~\ref{fig:nets-for-spikes}. \\
Information Diffusion $\longrightarrow$ Network Evolution: When node 6 joins the network a few nodes follow her and retweet her posts. Her tweets being propagated (shown in red) turning her to a valuable source of information. Therefore, those retweets are followed by links created to her (shown in magenta). \\
        Network Evolution $\longrightarrow$ Information Diffusion: Nodes 46 and 68 both have almost the same number of followees. However, as soon as node 46 connects to node 130 (which is a central node and retweets very much) her activity dramatically increases compared to node 68.}
        \label{fig:spikes-vis}
\end{figure}

\begin{figure} [!!h]
        \centering
        \begin{tabular}{c c c}
         \includegraphics[width=0.30\textwidth]{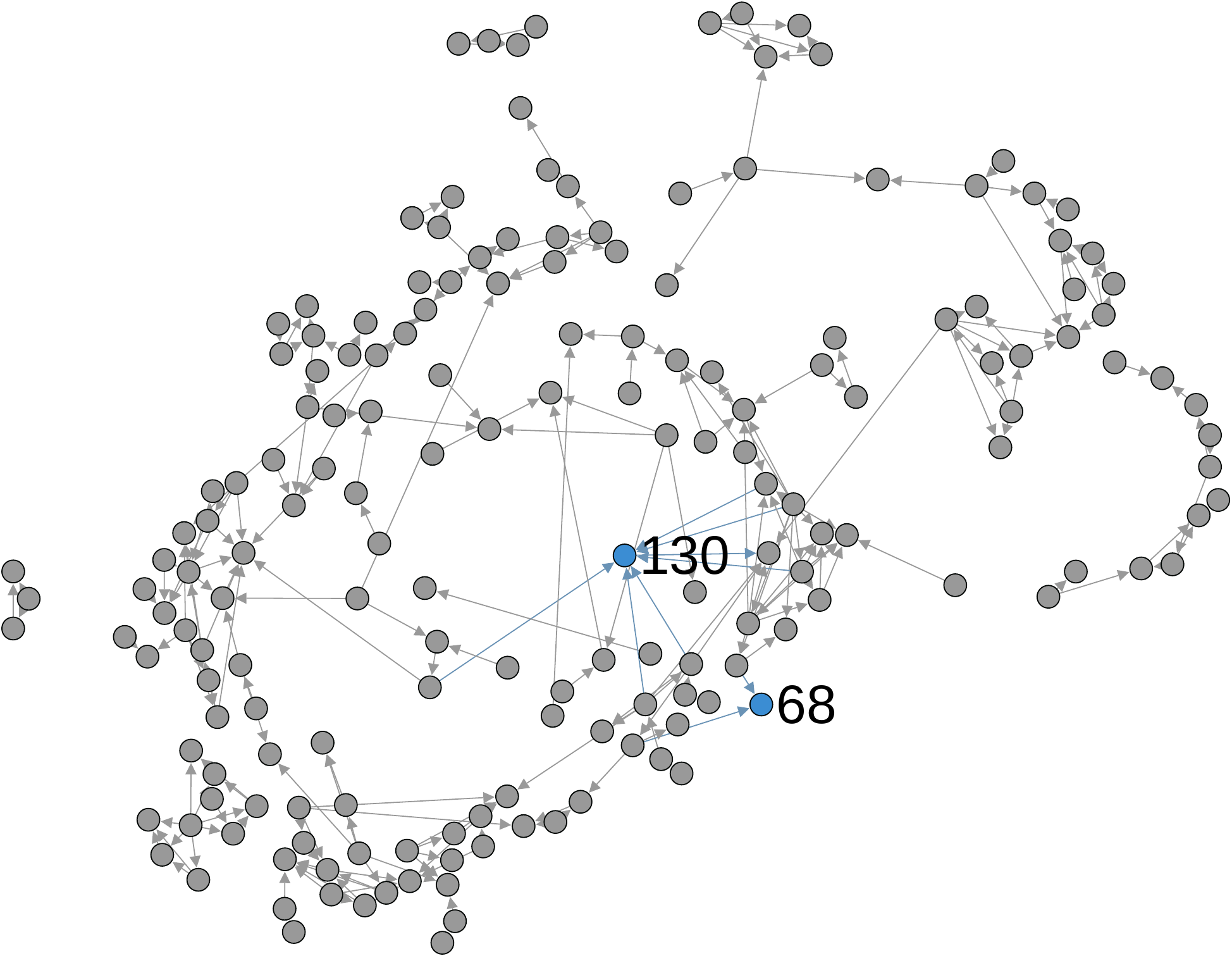} &
         \includegraphics[width=0.30\textwidth]{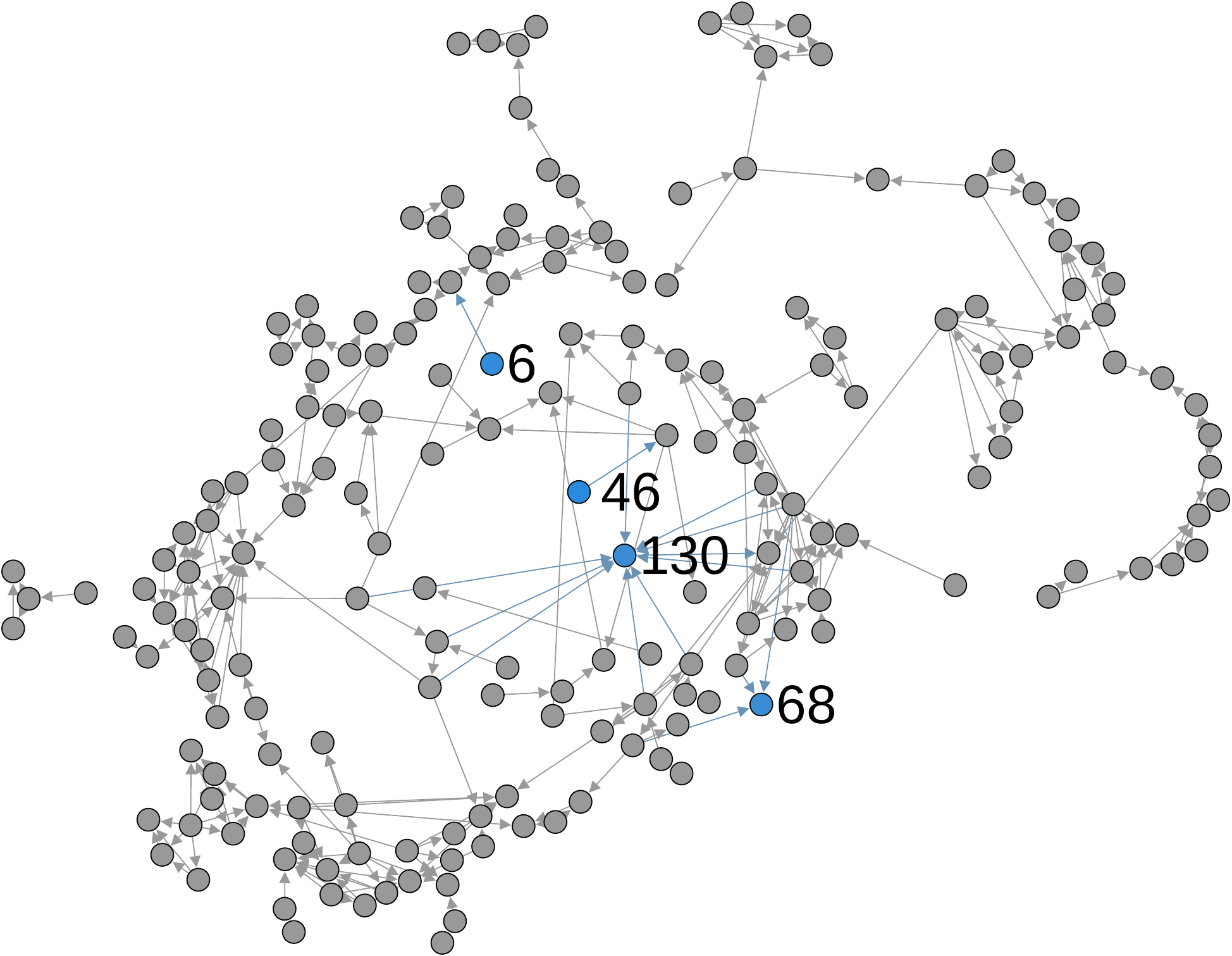} &
          \includegraphics[width=0.30\textwidth]{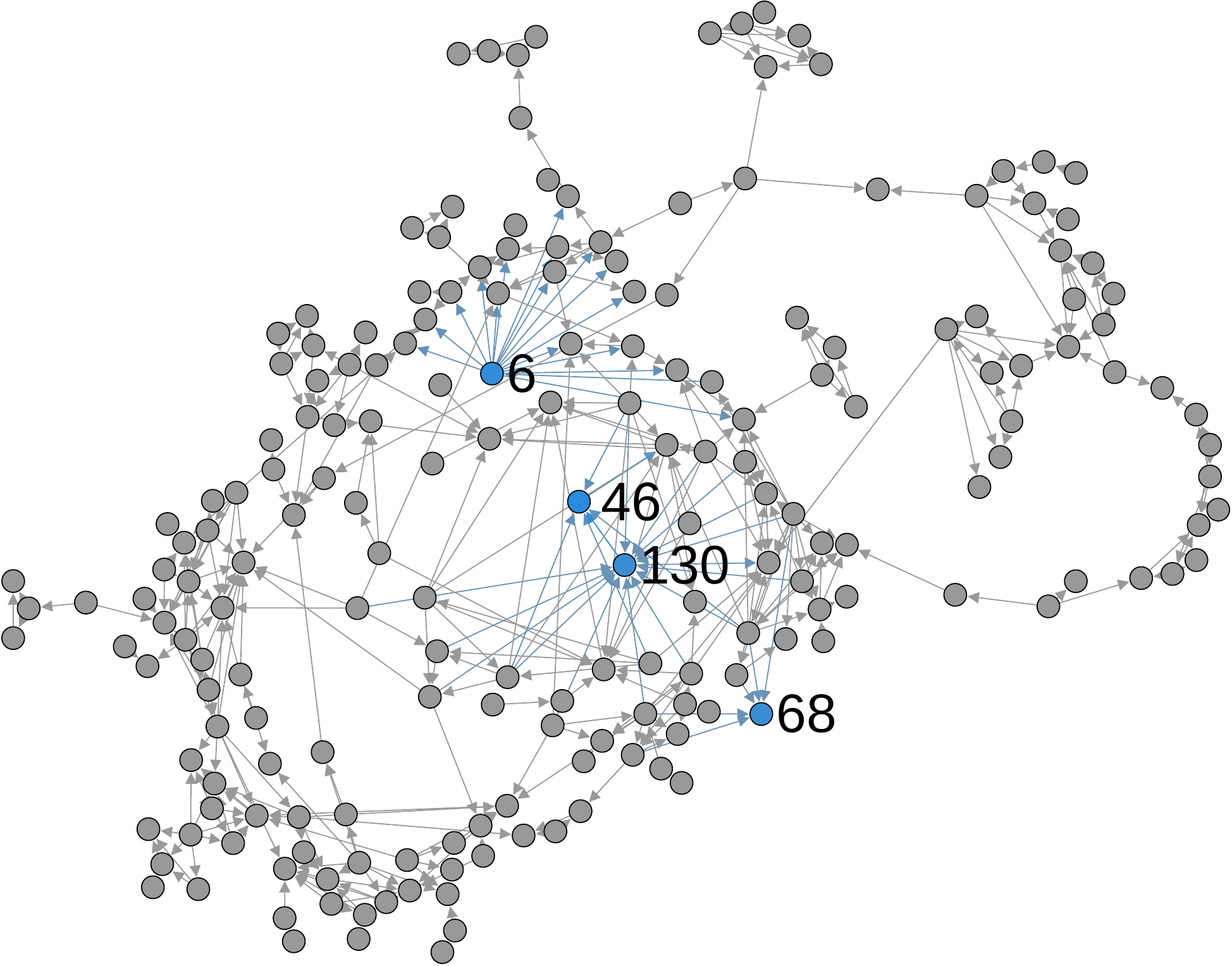} \\
          t=125 & t=137 & t=150
          \end{tabular}
          \caption{Network structure in which events from Figure~\ref{fig:spikes-vis} take place, at different times. 
          }  
                  \label{fig:nets-for-spikes}
\end{figure}

Figure~\ref{fig:spikes-vis} illustrates the spike trains (tweet, retweet, and link events) for the first 140 nodes of a network simulated with a similar set of parameters as above and Figure \ref{fig:nets-for-spikes} shows three snapshots of the network at different times.
First, consider node 6 in the network. After she joins the network, a few nodes begin to follow him. Then, when she starts to tweet, her tweets are retweeted many times by others (red spikes) in the figure and these retweets subsequently boost the number of nodes
that link to her (Magenta spikes). This clearly illustrates the scenario in which information diffusion triggers changes on the network structure.
Second, consider nodes 46 and 68 and compare their associated events over time. After some time, node 46 becomes much more active than node 68. To understand why, note that soon after time 137, node 46 followed node 130, which is a very central node ($\ie$ 
following a lot of people), while node 68 did not.
This clearly illustrates the scenario in which network evolution triggers changes on the dynamics of information diffusion.

\vspace{-2mm}
\subsection{Cascade Patterns}
Our model can produce the most commonly occurring cascades structures as well as heavy-tailed cascade size and depth distributions, as observed in historical Twitter data reported in ~\cite{GoeWatGol12}.
Figure~\ref{fig:cascades-structures-size-depth-alpha-varying} summarizes the results, which provide empirical evidence that the higher the $\alpha$ ($\beta$) value, the shallower and wider the 
cascades.

\section{Experiments on Model Estimation and Prediction on Synthetic Data} \label{sec:synthetic}
In this section, we first show that our model estimation method can accurately recover the true model parameters from historical link and diffusion events data and
then demonstrate that our model can accurately predict the network evolution and information diffusion over time, significantly outperforming two state of the art 
me\-thods~\cite{WenRatPerGonCasBonSchMenFla13,AntDov13,MyeLes14} at predicting new links, and a baseline Hawkes process that does not consider network 
evolution at predicting new events.

\subsection{Experimental Setup}
Throughout this section, we experi\-ment with our model considering $m$$=$$400$ nodes.
We set the model parameters for each node in the network by drawing samples from $\mu$$\sim$$U(0, 0.0004)$,  $\alpha$$\sim$$U(0, 0.1)$, $\eta$$\sim$$U(0, 1.5)$ and $\beta$$\sim$$U(0, 0.1)$. 
We then sample up to 60,000 link and information diffusion events from our model using Algorithm~\ref{alg:simulation} and average over 8 different simulation runs.
%

%
\begin{figure} [t]
        \centering
        \begin{tabular}{c c c}
        \includegraphics[width=0.21\textwidth]{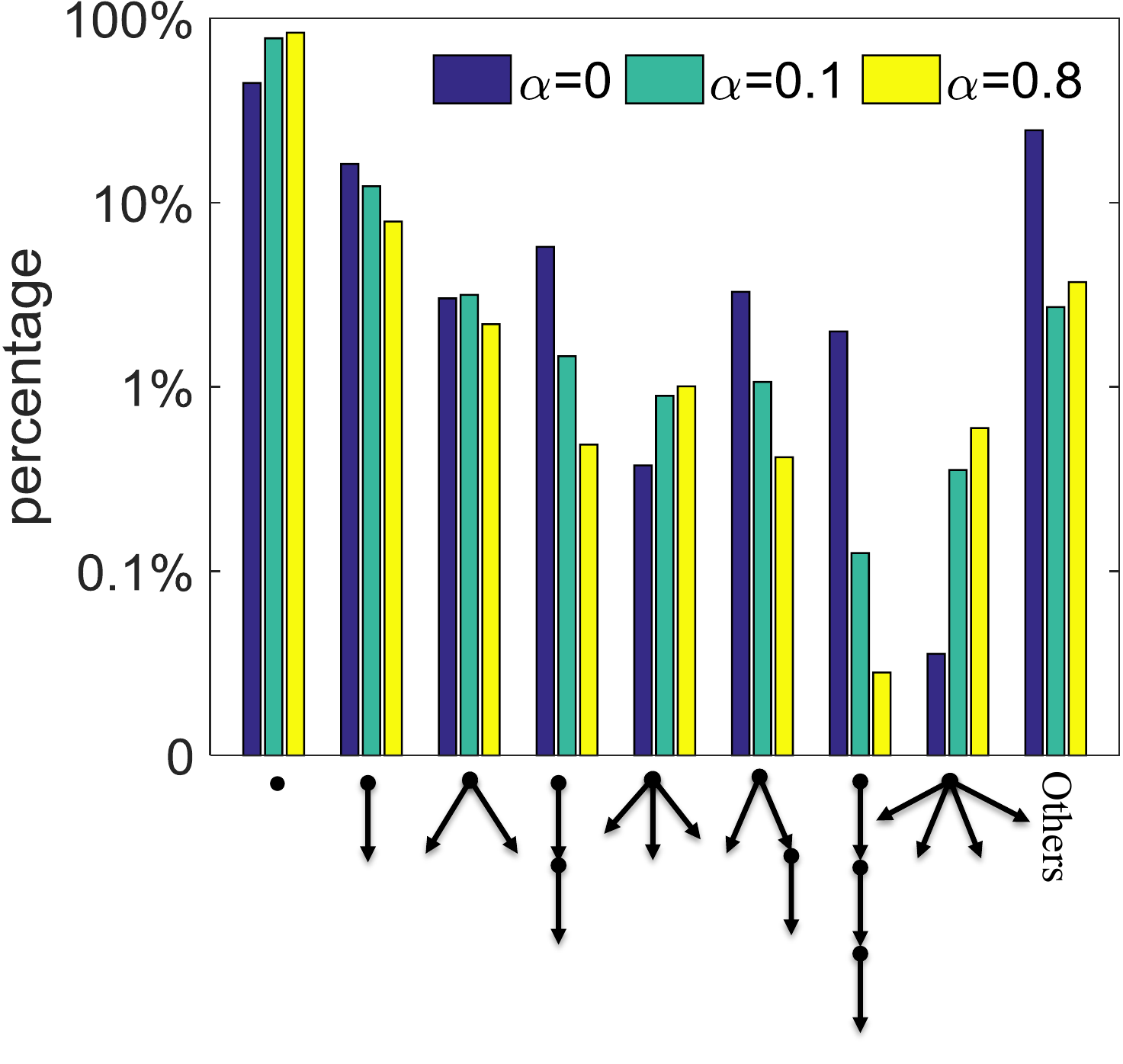} & \includegraphics[width=0.26\textwidth]{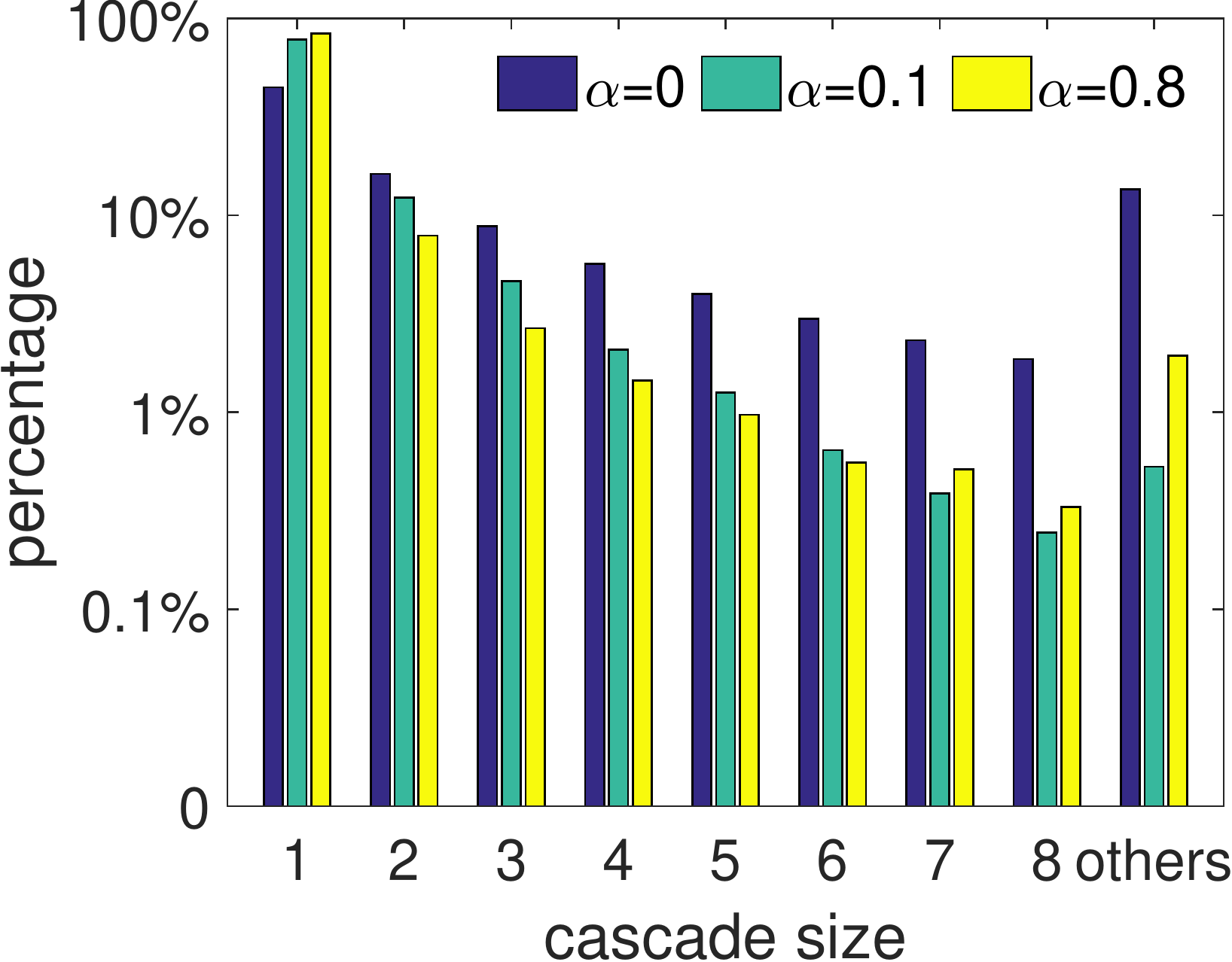} &
        \includegraphics[width=0.26\textwidth]{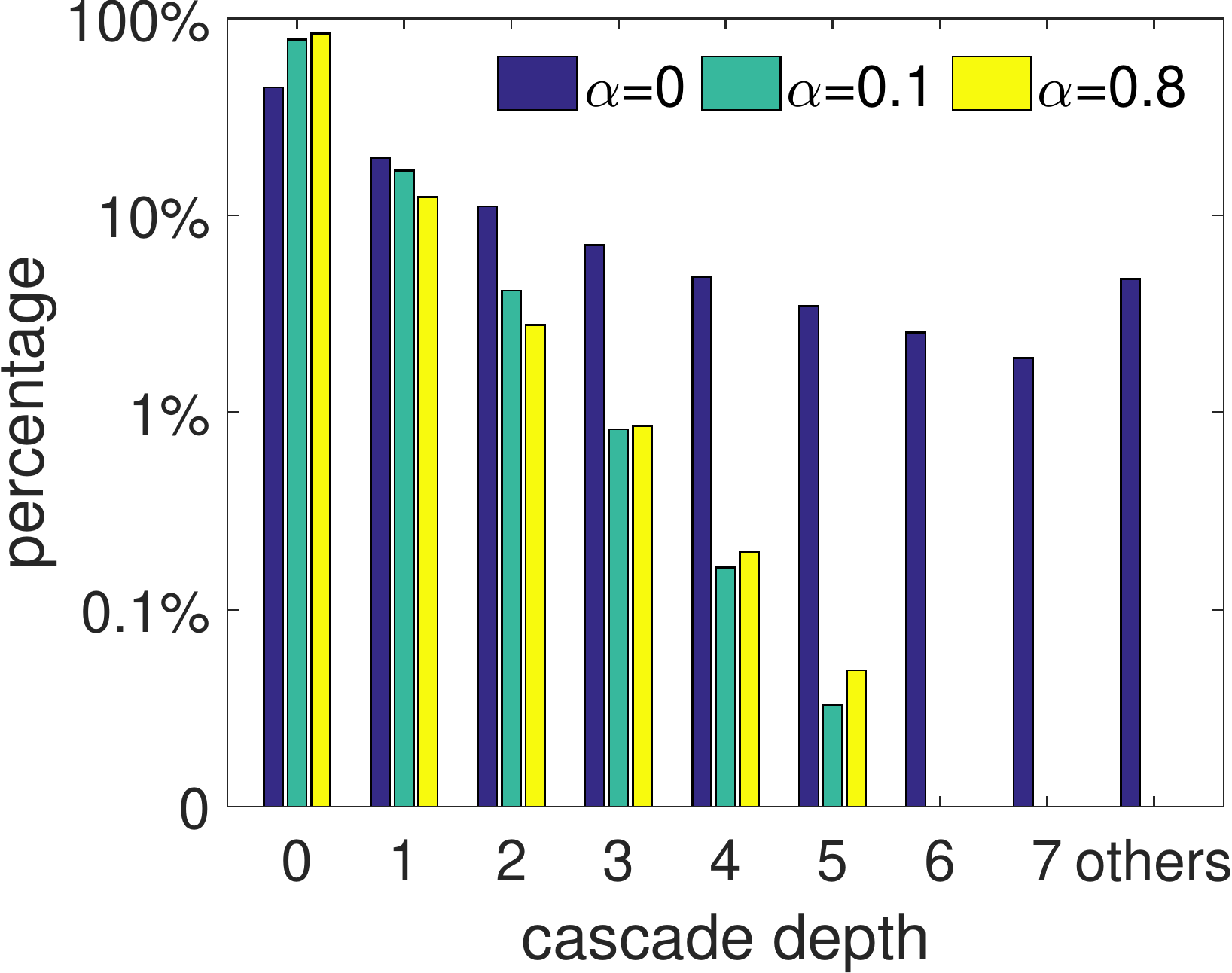} \\
         (a) & (b) & (c) \\
        \includegraphics[width=0.21\textwidth]{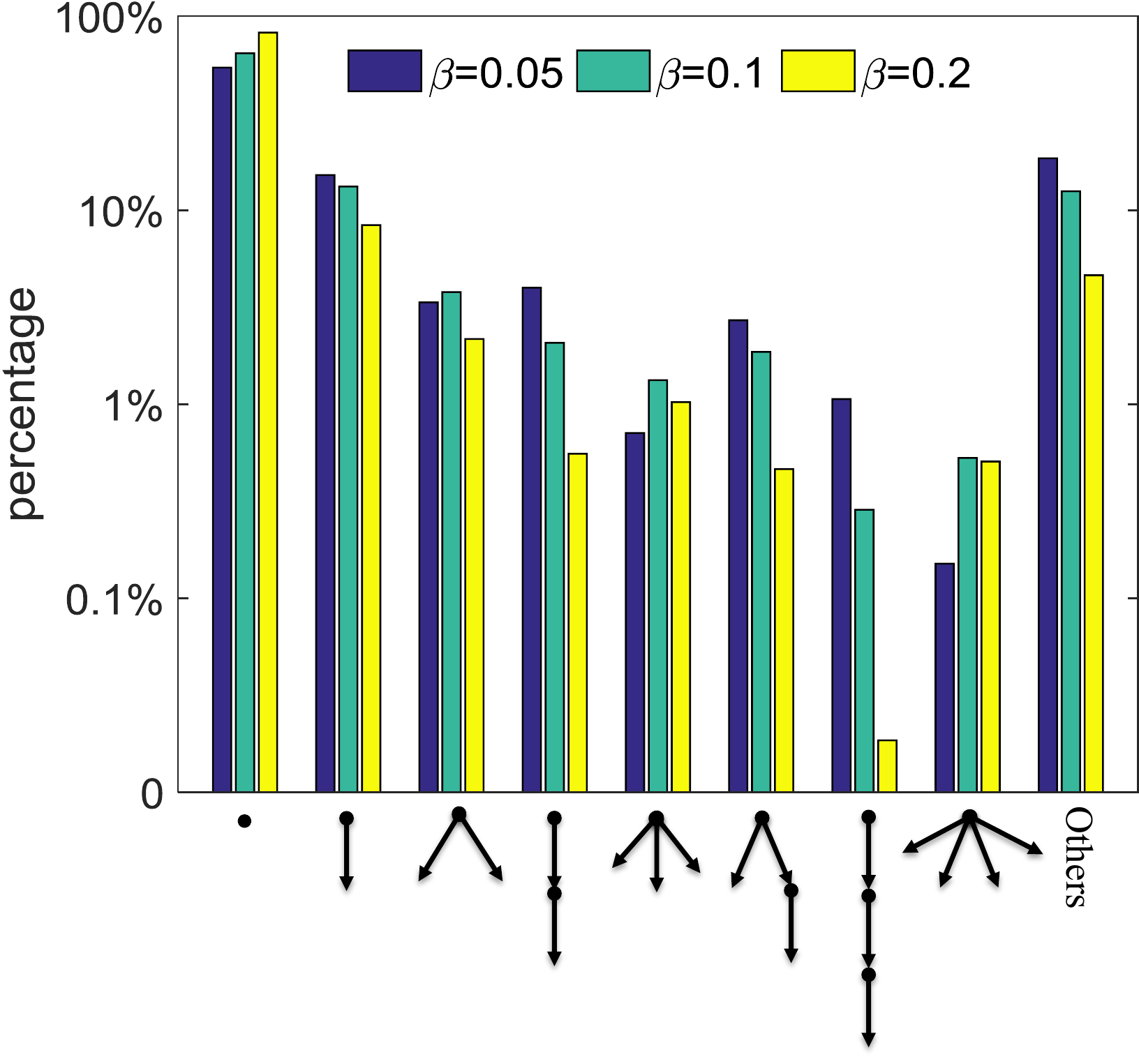} &
         \includegraphics[width=0.26\textwidth]{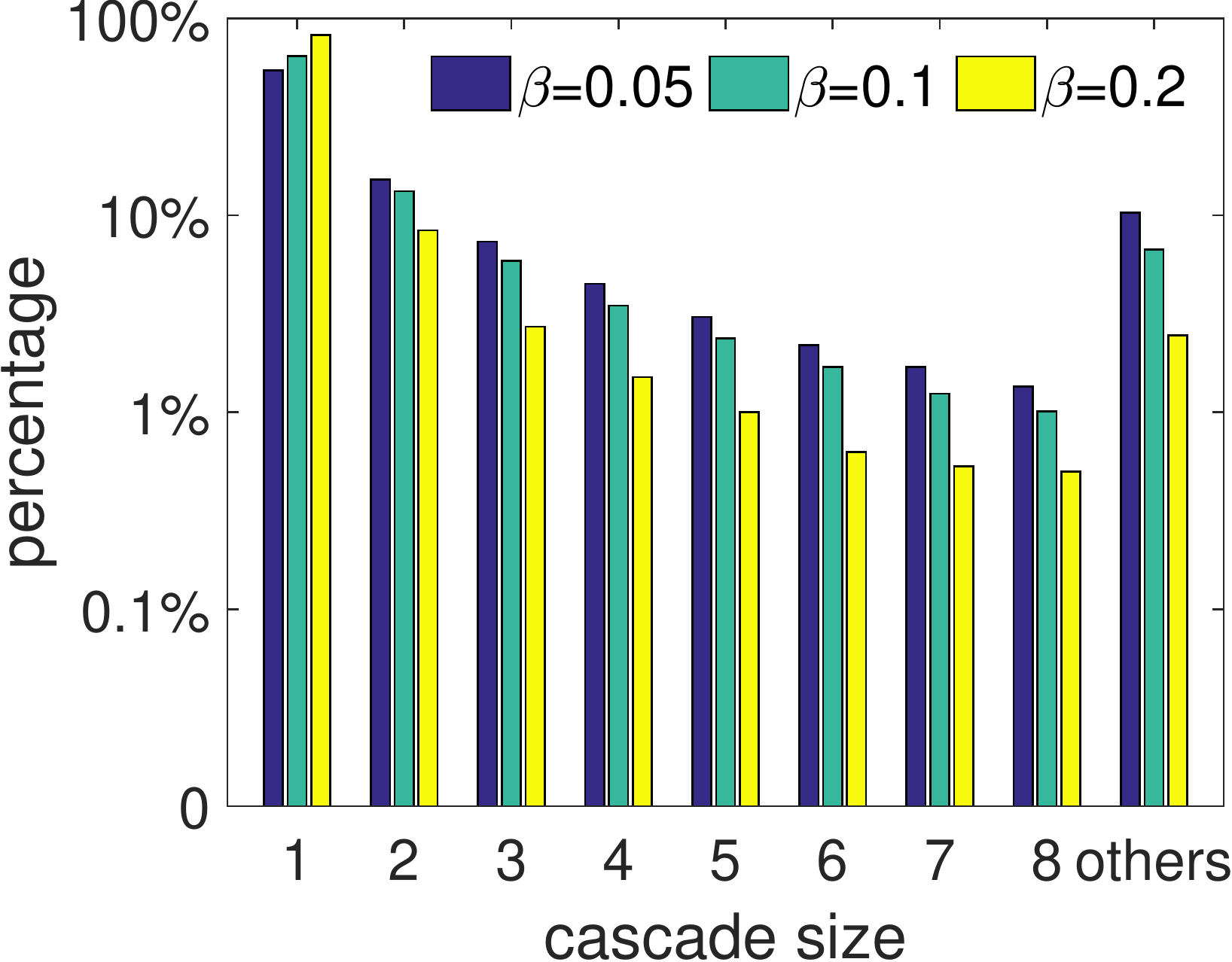} &
        \includegraphics[width=0.26\textwidth]{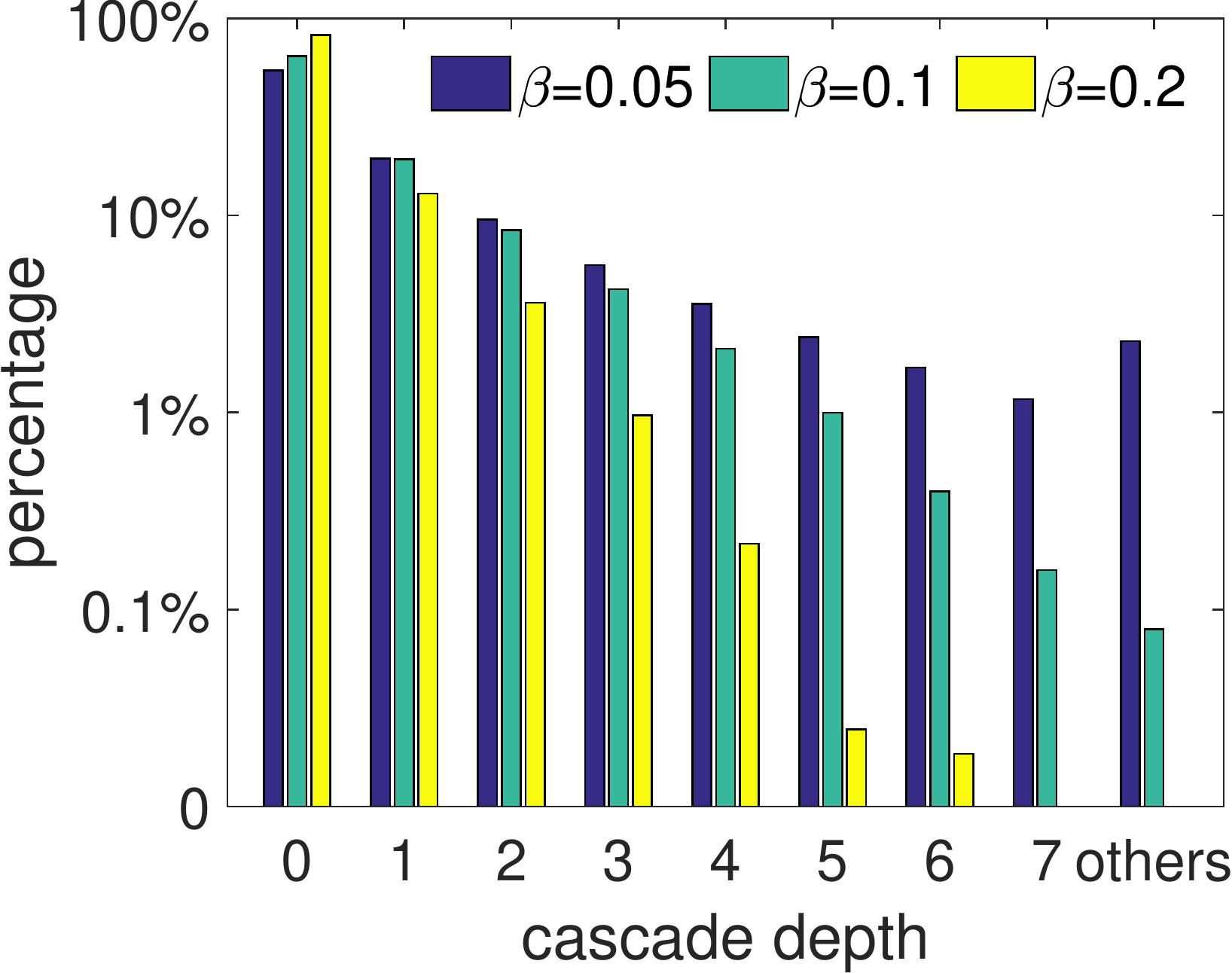} \\
         (d) & (e) & (f) \\
        \end{tabular}
        \caption{Distribution of cascade structure, size and depth for different $\alpha$ ($\beta$) va\-lues and fixed $\beta = 0.2$ ($\alpha = 0.8$).}
        \label{fig:cascades-structures-size-depth-alpha-varying}
\end{figure}

\subsection{Model Estimation}
We evaluate the accuracy of our model estimation procedure via two measures: (i) the relative mean absolute error (\ie, $\EE[|x-\hat{x}|/x]$, MAE) between the estimated parameters ($x$) and the true parameters ($\hat{x}$), 
(ii) the Kendall'{}s rank correlation coefficient between each estimated parameter and its true value, and (iii) test log-likelihood.
Figure~\ref{fig:params-inference} shows that as we feed more events into the estimation procedure, the estimation becomes more accurate. 
%
%

\subsection{Link Prediction}
We use our model to predict the identity of the source for each test link event, given the historical events before the time of the prediction, and compare its performance with two state of the art methods, which 
we denote as TRF~\cite{AntDov13} and WENG~\cite{WenRatPerGonCasBonSchMenFla13}. 
TRF measures the probability of creating a link from a source at a given time by simply computing the proportion of new links created from the source over all total 
created links up to the given time.
WENG considers several link creation strategies and makes a prediction by combining these strategies.

Here, we evaluate the performance by computing the pro\-ba\-bility of all potential links using our model, TRF and WENG and then compute (i) the average rank of all true (test) events (AvgRank) and, (ii) the success 
probability that the true (test) events rank among the top-1 potential events at each test time (Top-1).
%
Figure~\ref{fig:prediction-synthetic} summarizes the results, where we trained our model with an increasing number of events.
Our model outperforms both TRF and WENG for a significant margin.
%
%

\subsection{Activity Prediction} 
We use our model to predict the identity of the node that generates each test diffusion event, given the historical events before the time of the prediction, and compare its performance with a baseline 
consisting of a Hawkes process without network evolution. For the Hawkes baseline, we take a snapshot of the network right before the prediction time, and use all historical retweeting events 
to fit the model. 
Here, we evaluate the performance via the same two measures as in the link prediction task and summarize the results in Figure~\ref{fig:prediction-synthetic} against an increasing number of training 
events.
The results show that, by modeling the network evolution, our model performs significantly better than the baseline.
\begin{figure} [t]
        \centering
        \begin{tabular}{c c c}
        \includegraphics[width=0.25\textwidth]{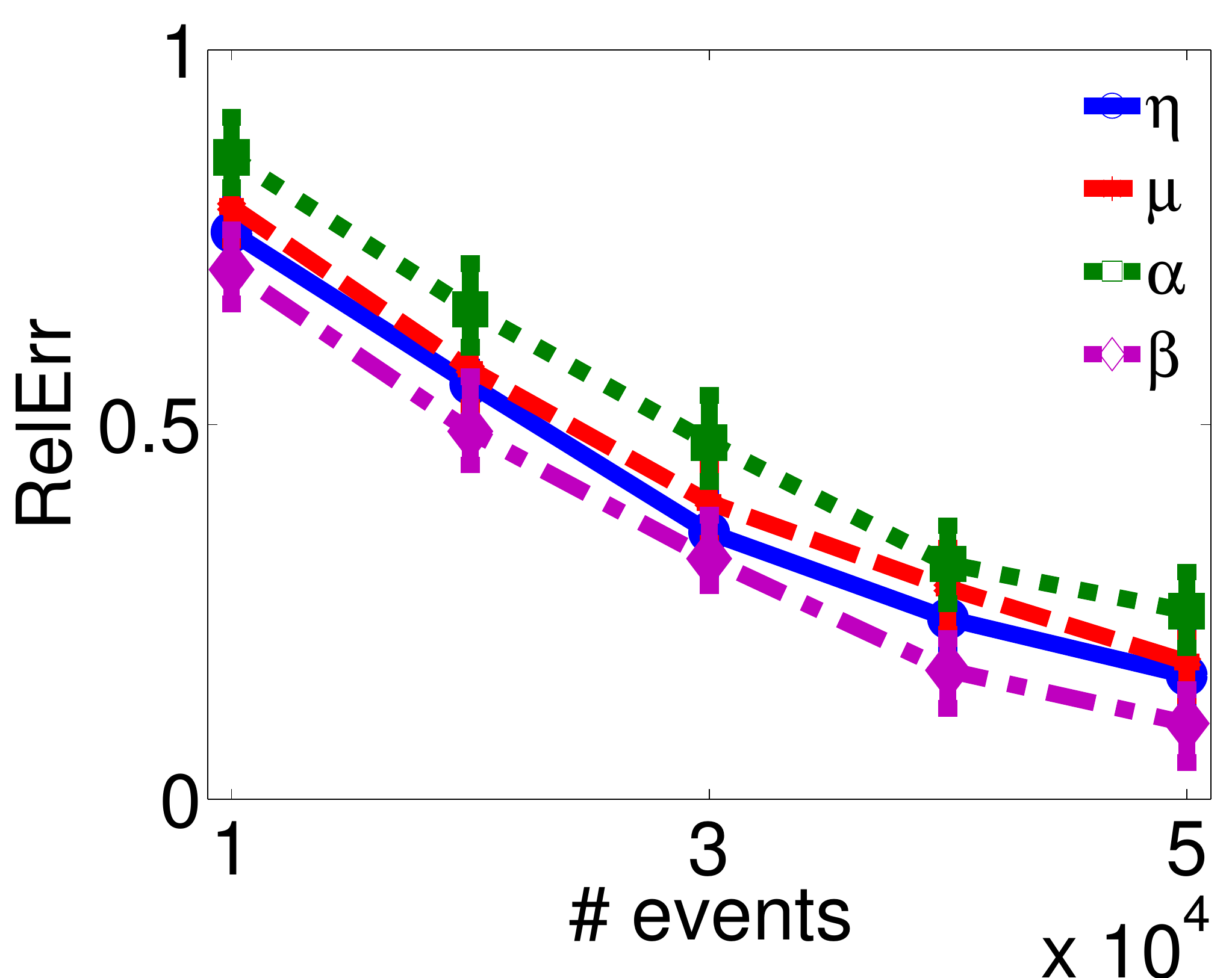} &
        \includegraphics[width=0.25\textwidth]{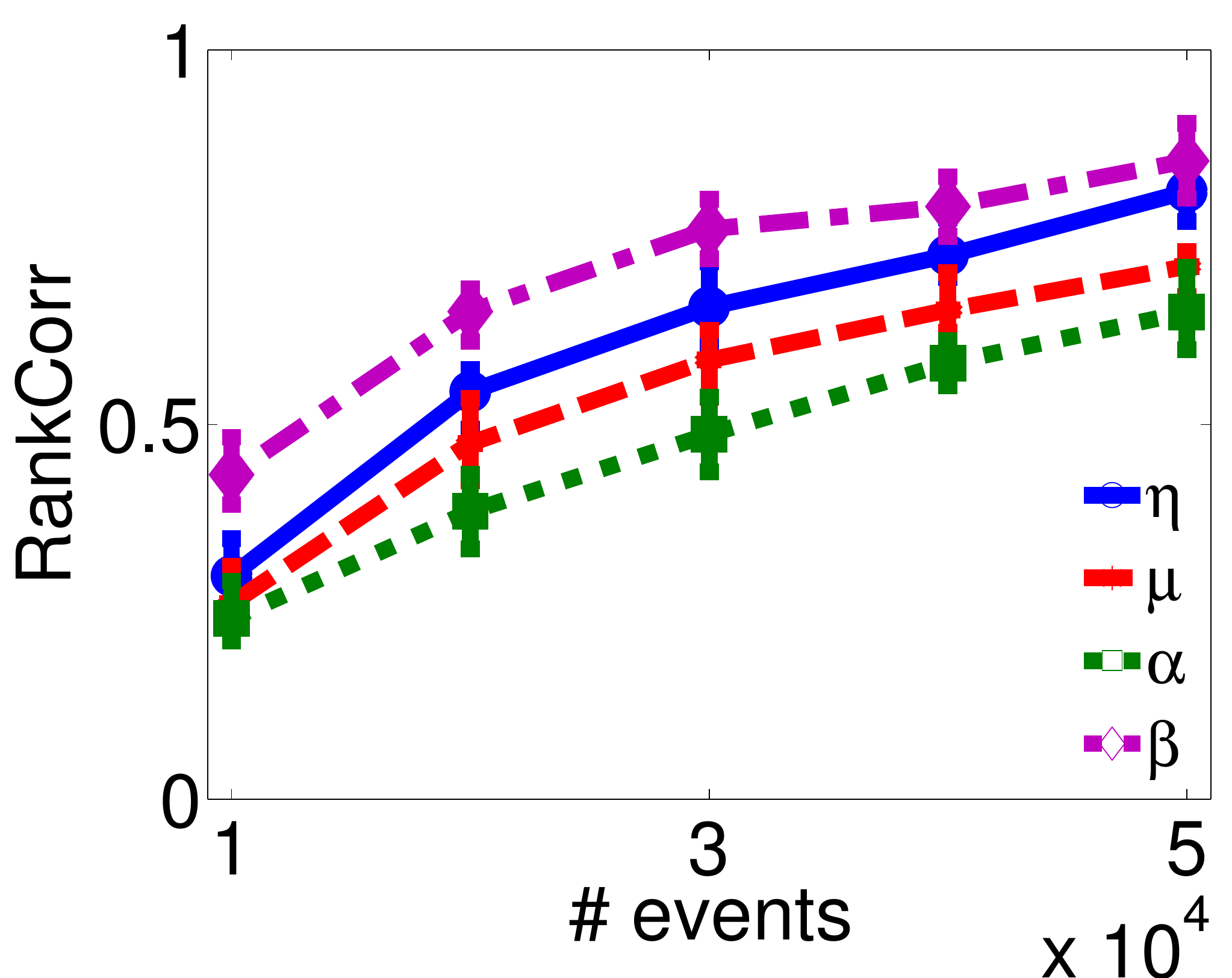} &
        \includegraphics[width=0.25\textwidth]{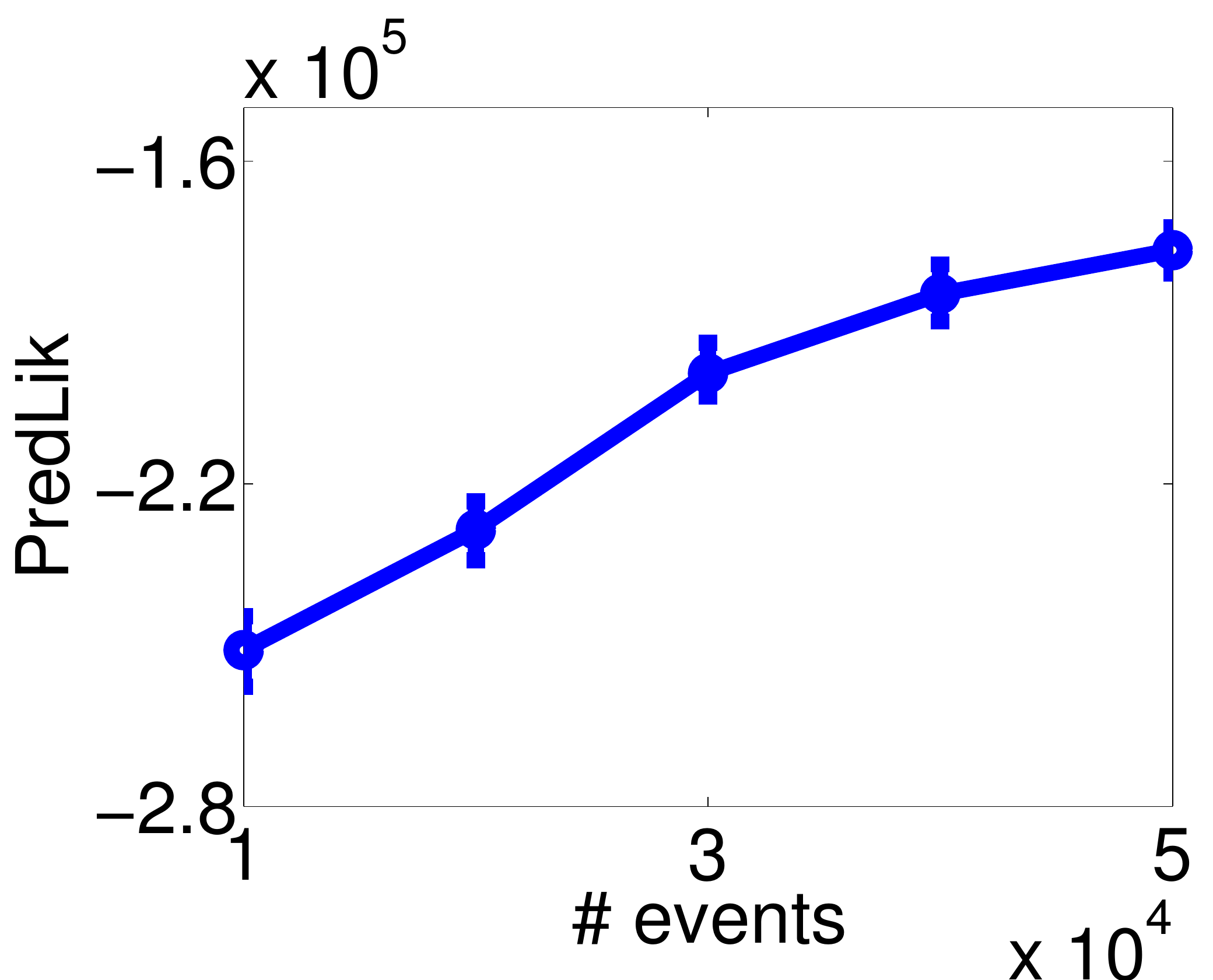} \\
        (a) Relative MAE & (b) Rank correlation & (c) Test log-likelihood \\
        \end{tabular}
        \caption{Performance of model estimation for a 400-node synthetic network.}
        \label{fig:params-inference}
\end{figure}
\begin{figure*}[t]
        \centering
        \begin{tabular}{c c c c}
         \includegraphics[width=0.20\textwidth]{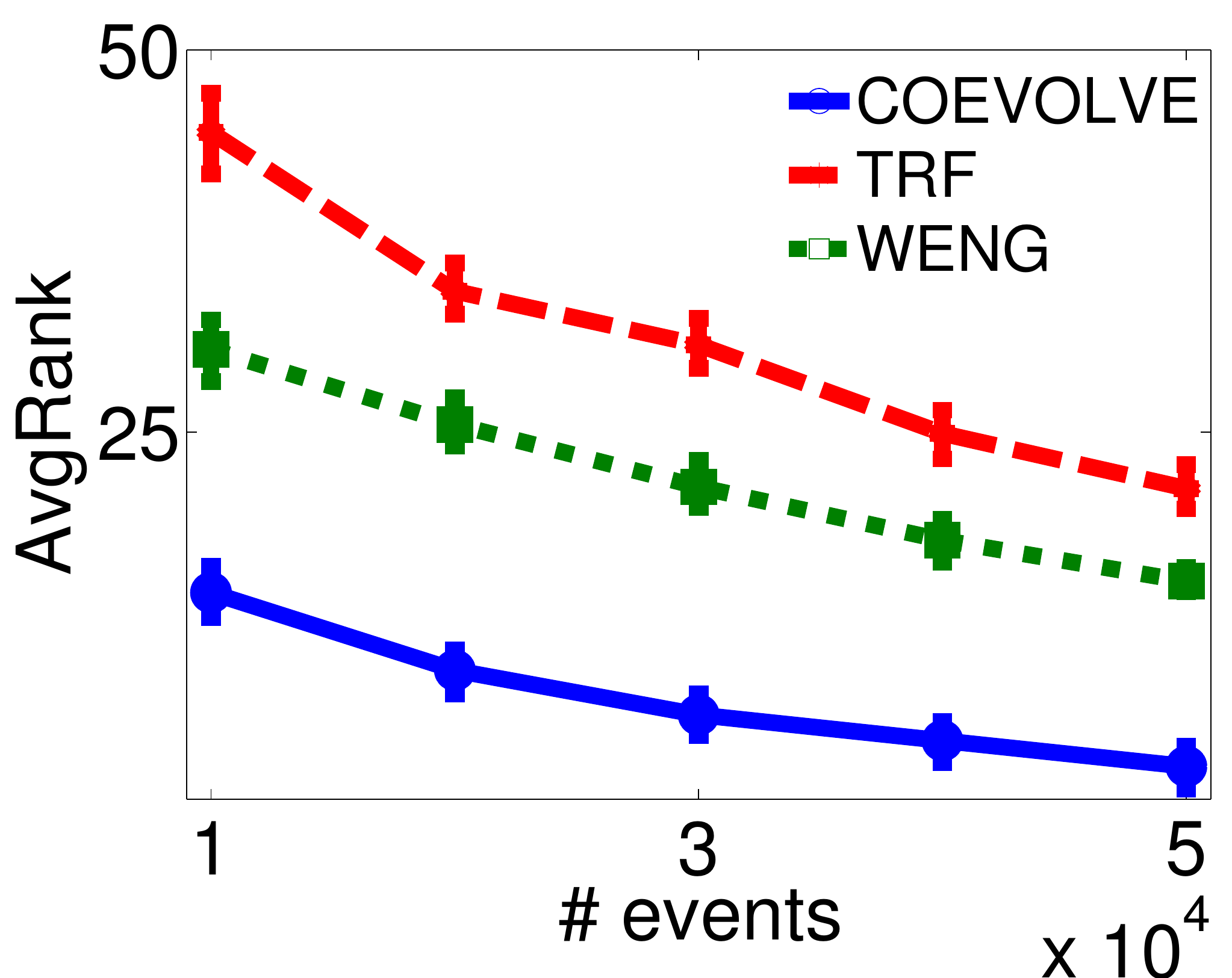} \hspace{2mm} &
          \includegraphics[width=0.20\textwidth]{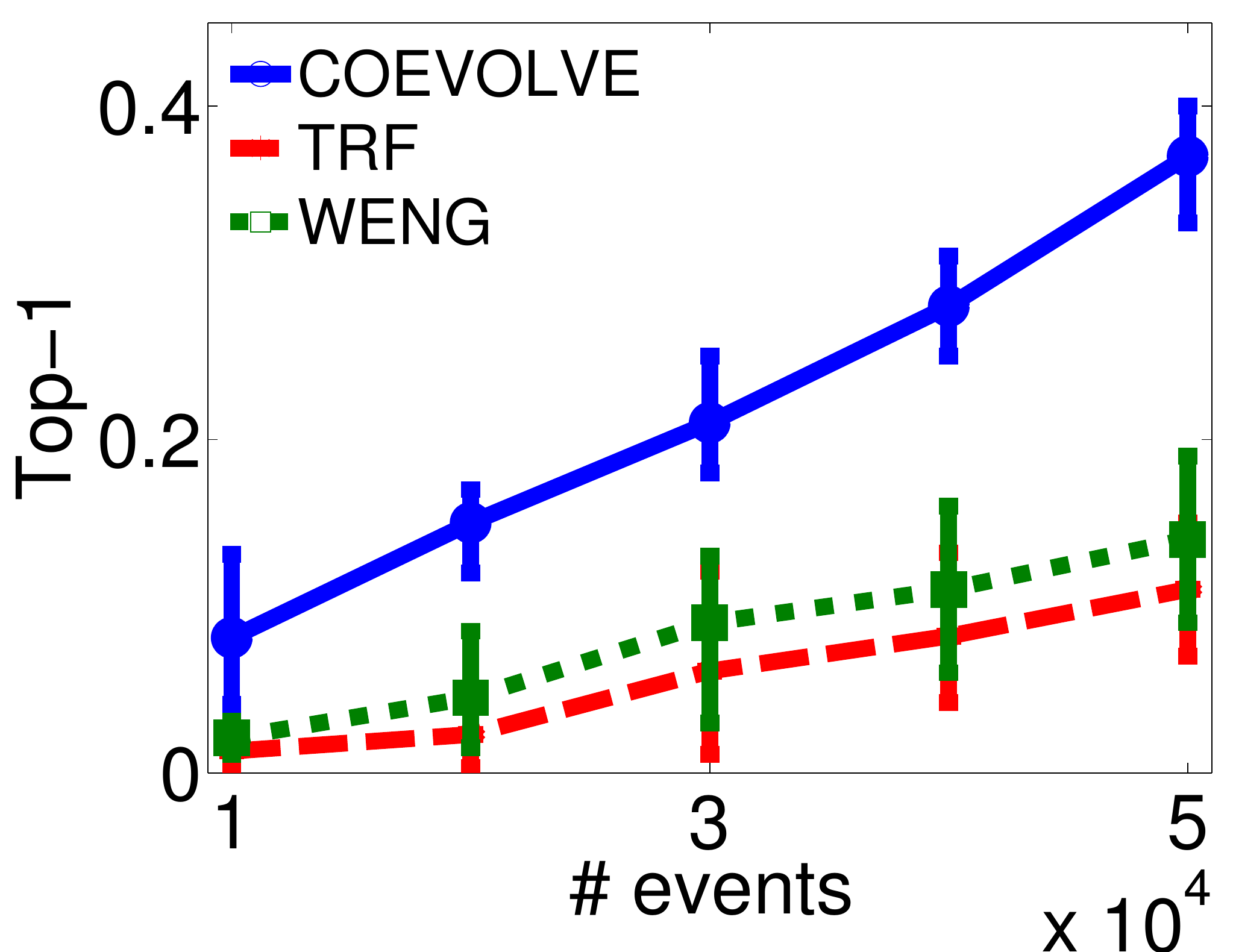}  \hspace{2mm} &
           \includegraphics[width=0.20\textwidth]{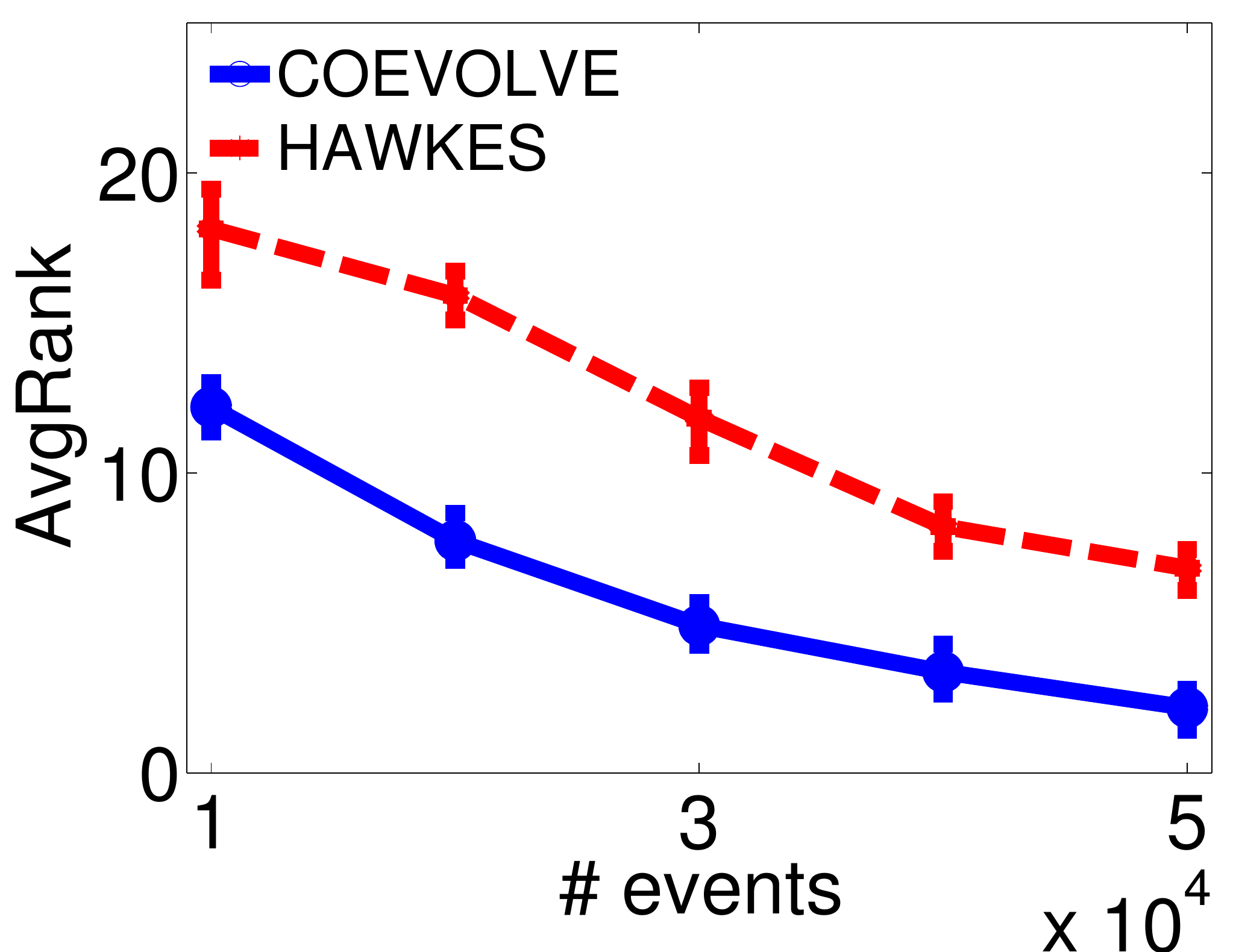} \hspace{2mm}  &
          \includegraphics[width=0.20\textwidth]{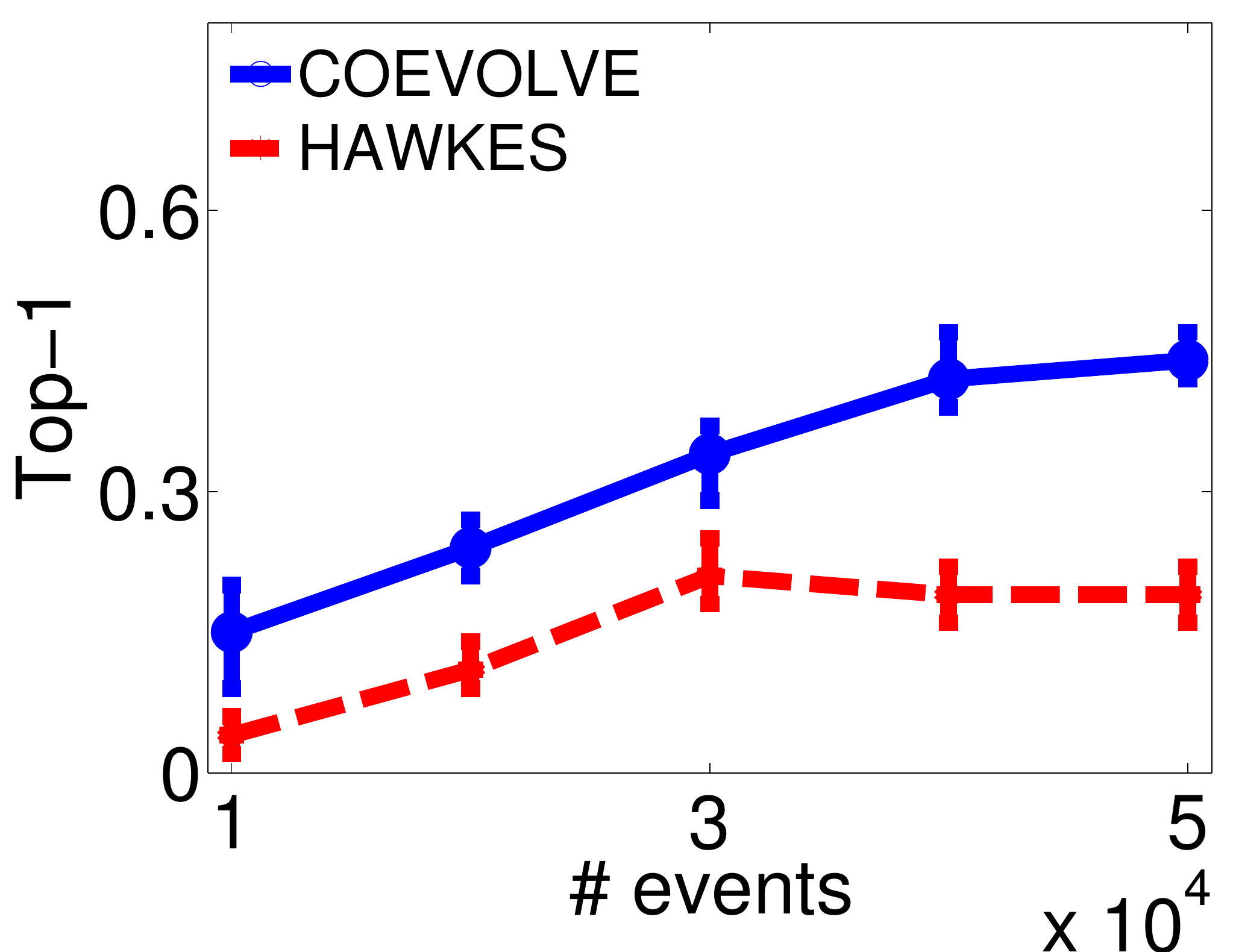} \\
          (a) Links: AR &  (b) Links: Top-1 & (c) Activity: AR & Activity: Top-1 \\
        \end{tabular}
        \caption{Prediction performance for a 400-node synthetic network by means of average rank (AR) and success probability that the true (test) events rank among the top-1 events (Top-1).}
        \label{fig:prediction-synthetic}
\end{figure*}

\section{Experiments on Coevolution and Prediction on Real Data}
\label{sec:real}
In this section, we validate our model using a large Twitter dataset containing nearly 550,000 tweet, retweet and link events 
from more than  280,000 users~\cite{AntDov13}.
We will show that our model can capture the co-evolutionary dynamics and, by doing so, it predicts retweet and link creation events more accurately than several alternatives. 
%

\subsection{Dataset Description \& Experimental Setup}
We use a dataset that contains both link events as well as tweets/retweets from millions of Twitter users~\cite{AntDov13}.
In particular, the dataset contains data from three sets of users in 20 days; nearly 8 million tweet, retweet, and link events by more than 6.5 million users.
The first set of users (8,779 users) are source nodes $s$, for whom all their tweet times were collected. The second set of users (77,200 users) are the follo\-wers of the first set of users,
for whom all their retweet times (and source identities) were collected. The third set of users (6,546,650 users) are the users that start following at least one user in the first 
set during the recording period, for whom all the link times were collected.

In our experiments, we focus on all events (and users) during a 10-day period (Sep 21 2012 - 30 Sep 2012) and used the information before Sep 21 to construct the initial social network (original links between users).  We model the co-evolution in the second 10-day period using our framework.
More specifically, in the coevolution modeling, we have 5,567 users in the first layer who post 221,201 tweets. In the second layer 101,465 retweets are generated by the whole 77,200 users in that interval. And in the third layer we have 198,518 users who create 219,134 links to 1978 users (out of 5567) in the first layer. 

We split events into a training set (covering 85\% of the retweet and link events) and a test set (covering the remaining 15\%) according to time, \ie, all events 
in the training set occur earlier than those in the test set. We then use our model estimation procedure to fit the parameters from an increasing proportion of events from the training data. 
%


%
%
%
\begin{figure} [t]
        \centering
        \begin{tabular}{c c c c}
        \hspace{-4mm}
         \includegraphics[width=0.23\textwidth]{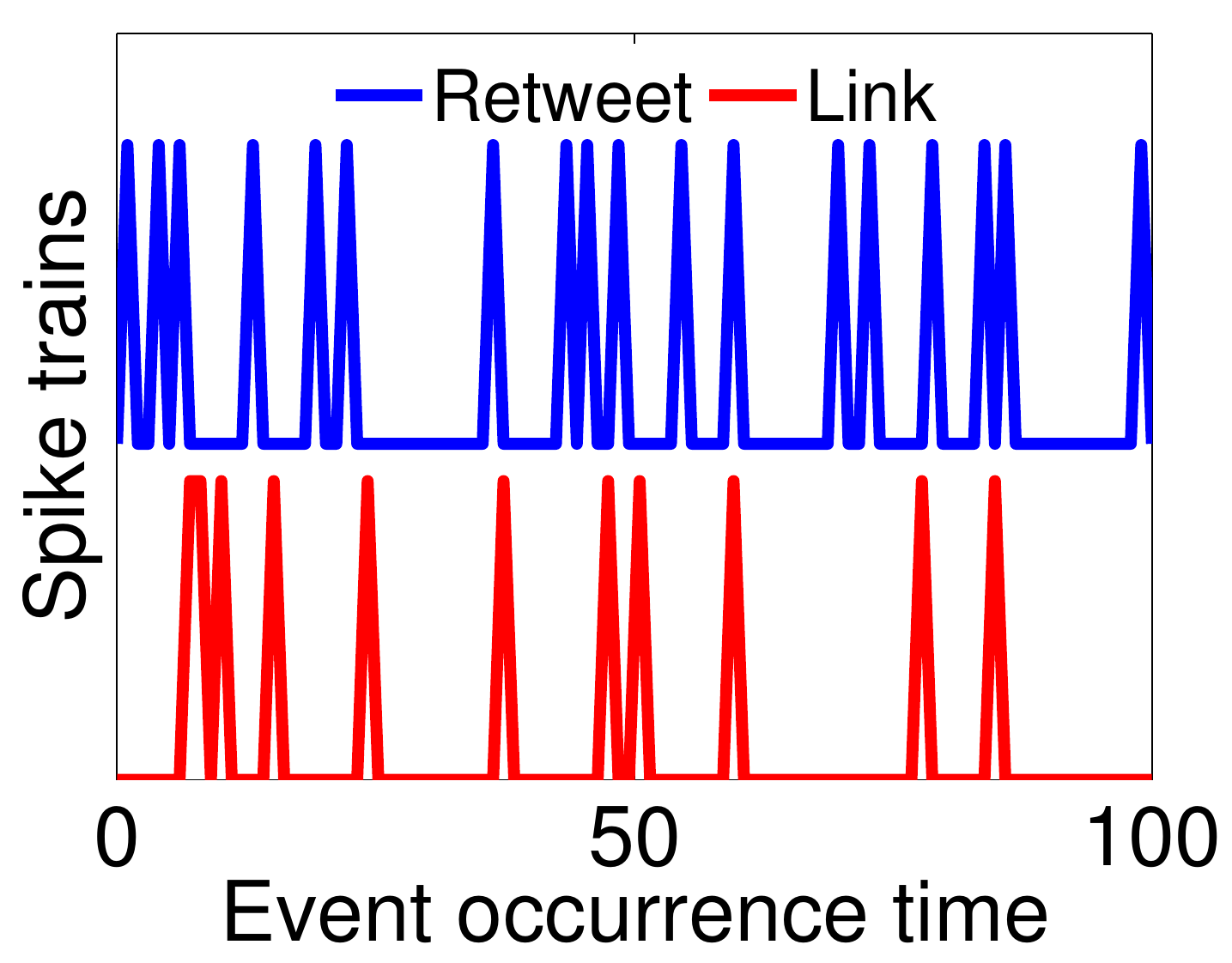} &
         \hspace{-6mm}
          \includegraphics[width=0.24\textwidth]{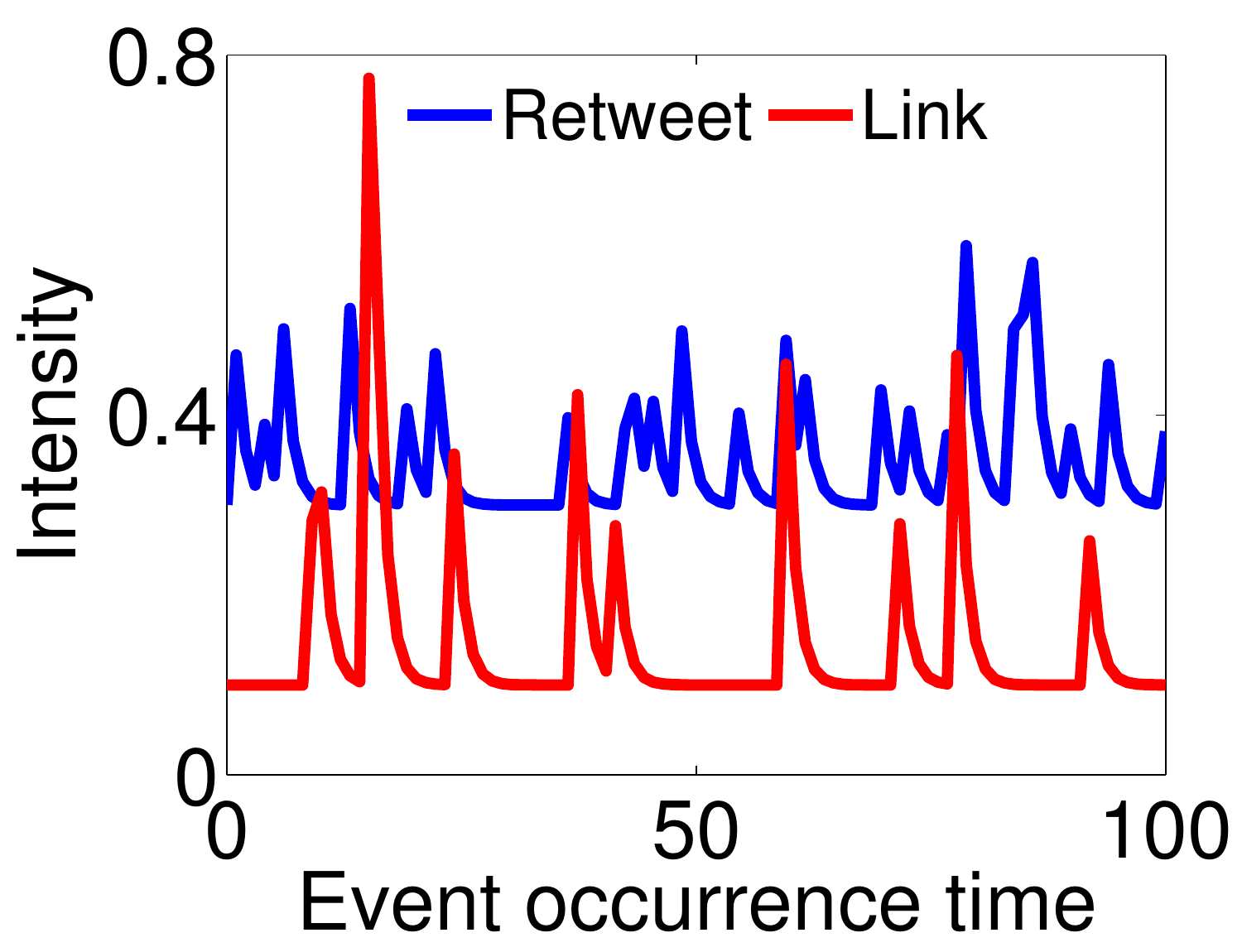} &
          \hspace{2mm}
          \includegraphics[width=0.23\textwidth]{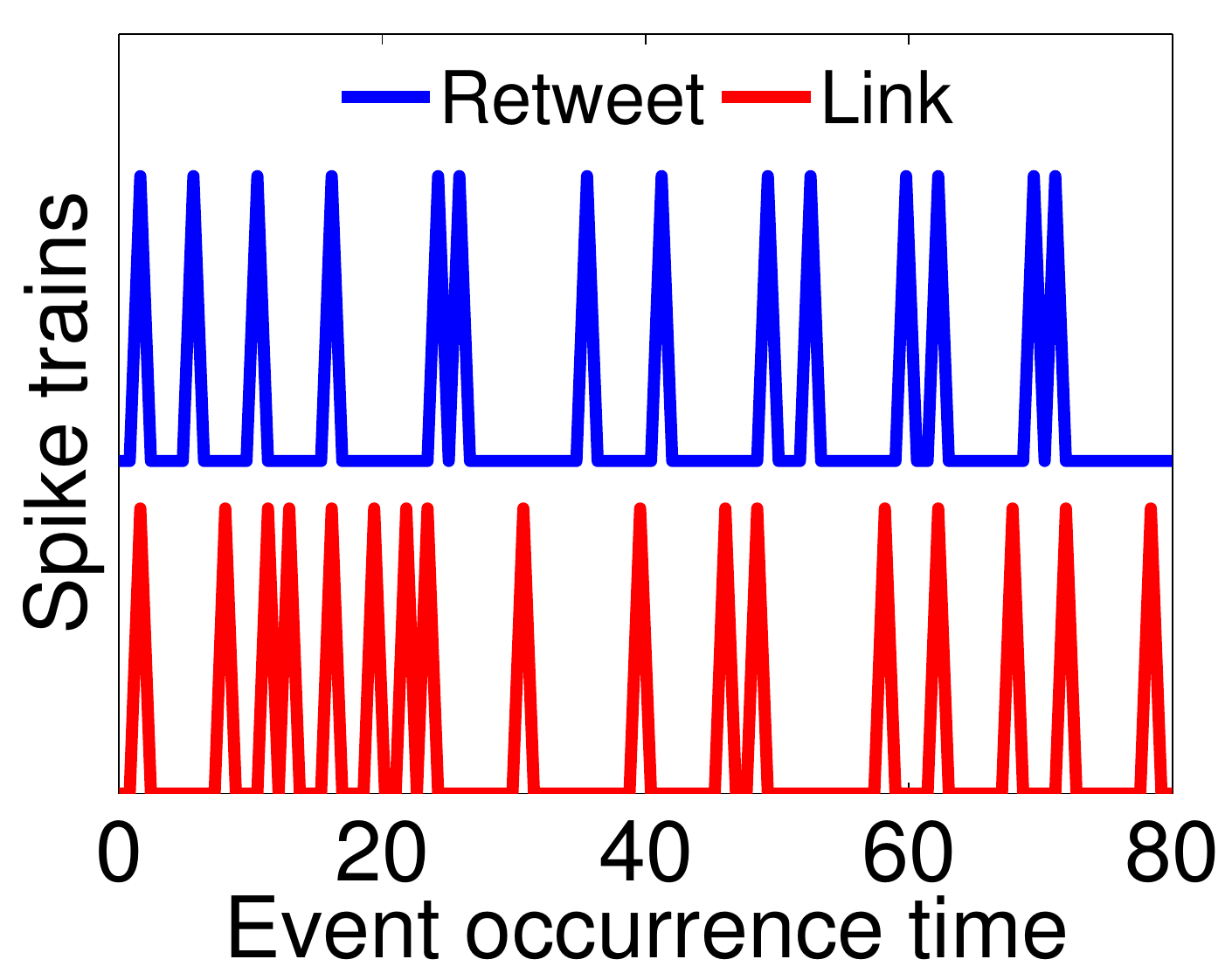} &
                   \hspace{-6mm}
          \includegraphics[width=0.24\textwidth]{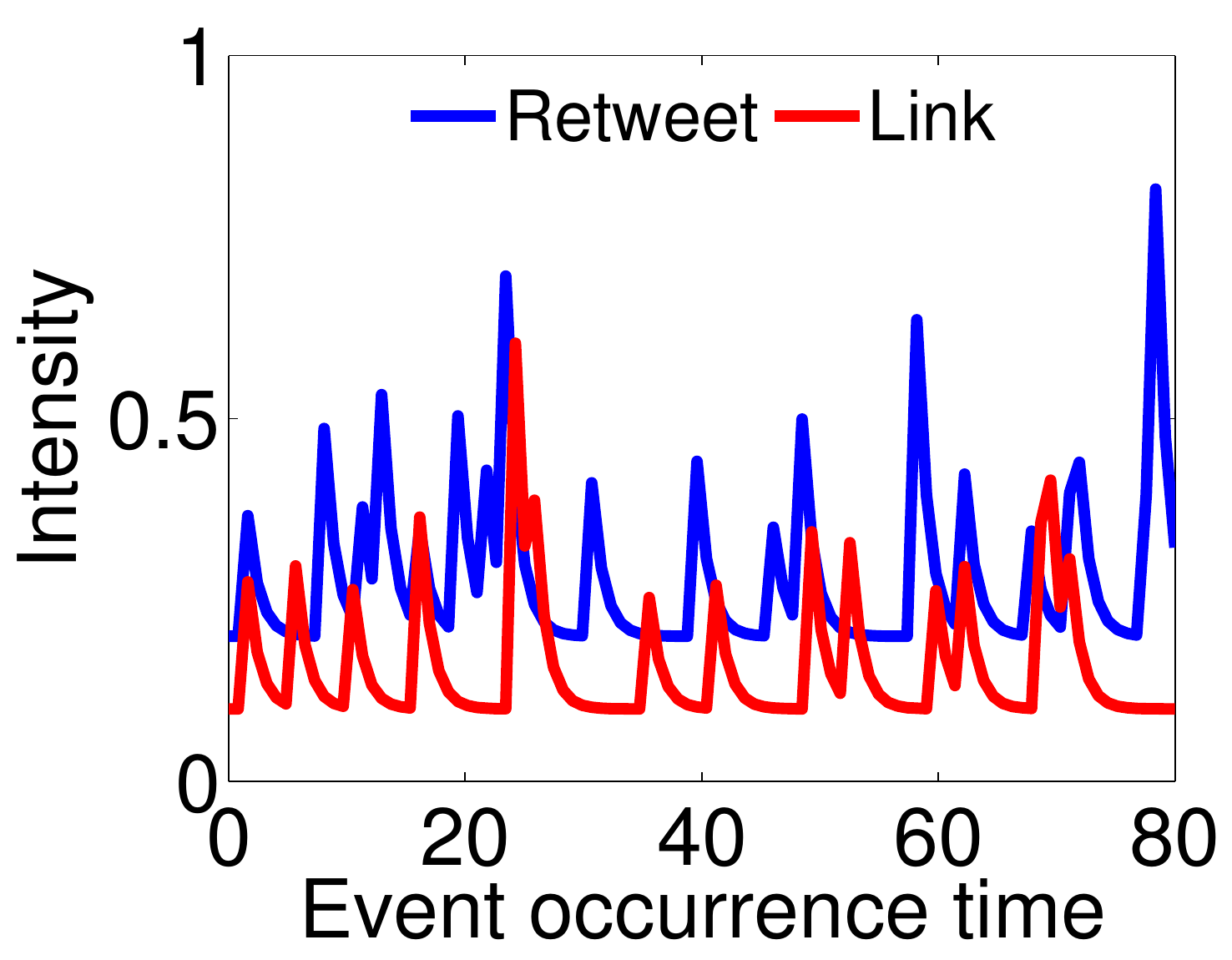}
        \vspace{3mm}
           \\
           (a) & (b) & (c) & (d) \\
        \hspace{-4mm}
         \includegraphics[width=0.23\textwidth]{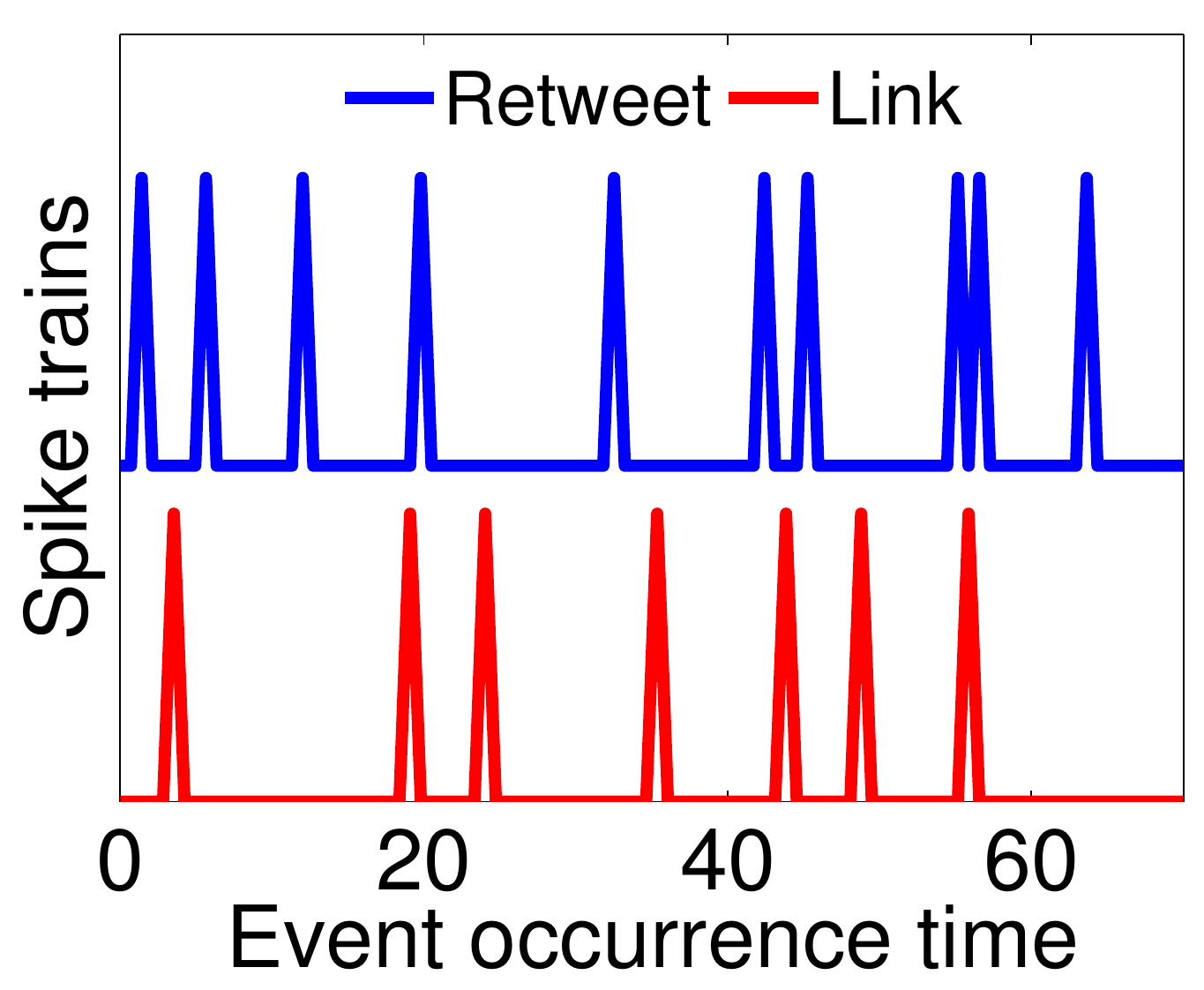} &
         \hspace{-5mm}
          \includegraphics[width=0.24\textwidth]{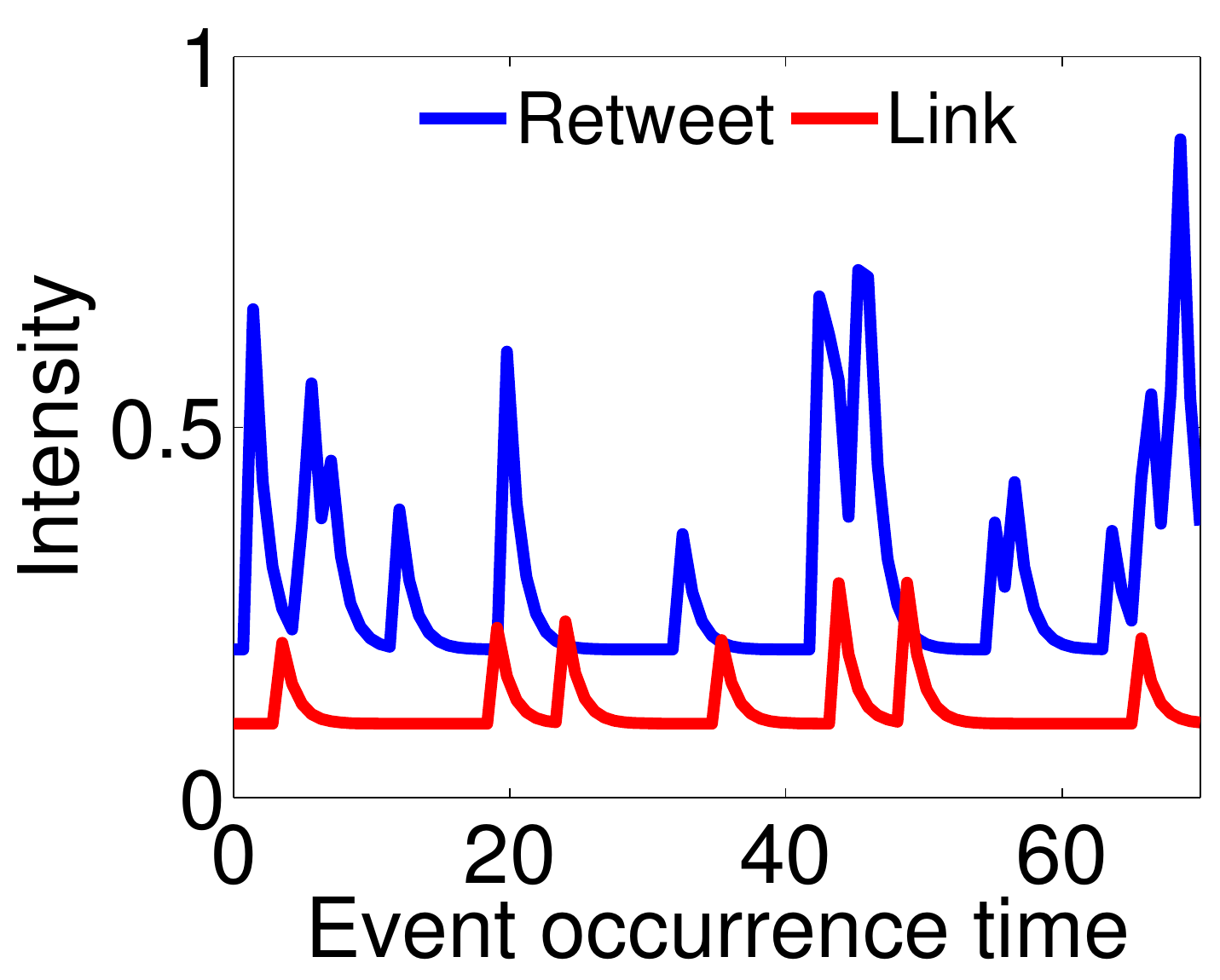} &
          \hspace{2mm}
          \includegraphics[width=0.23\textwidth]{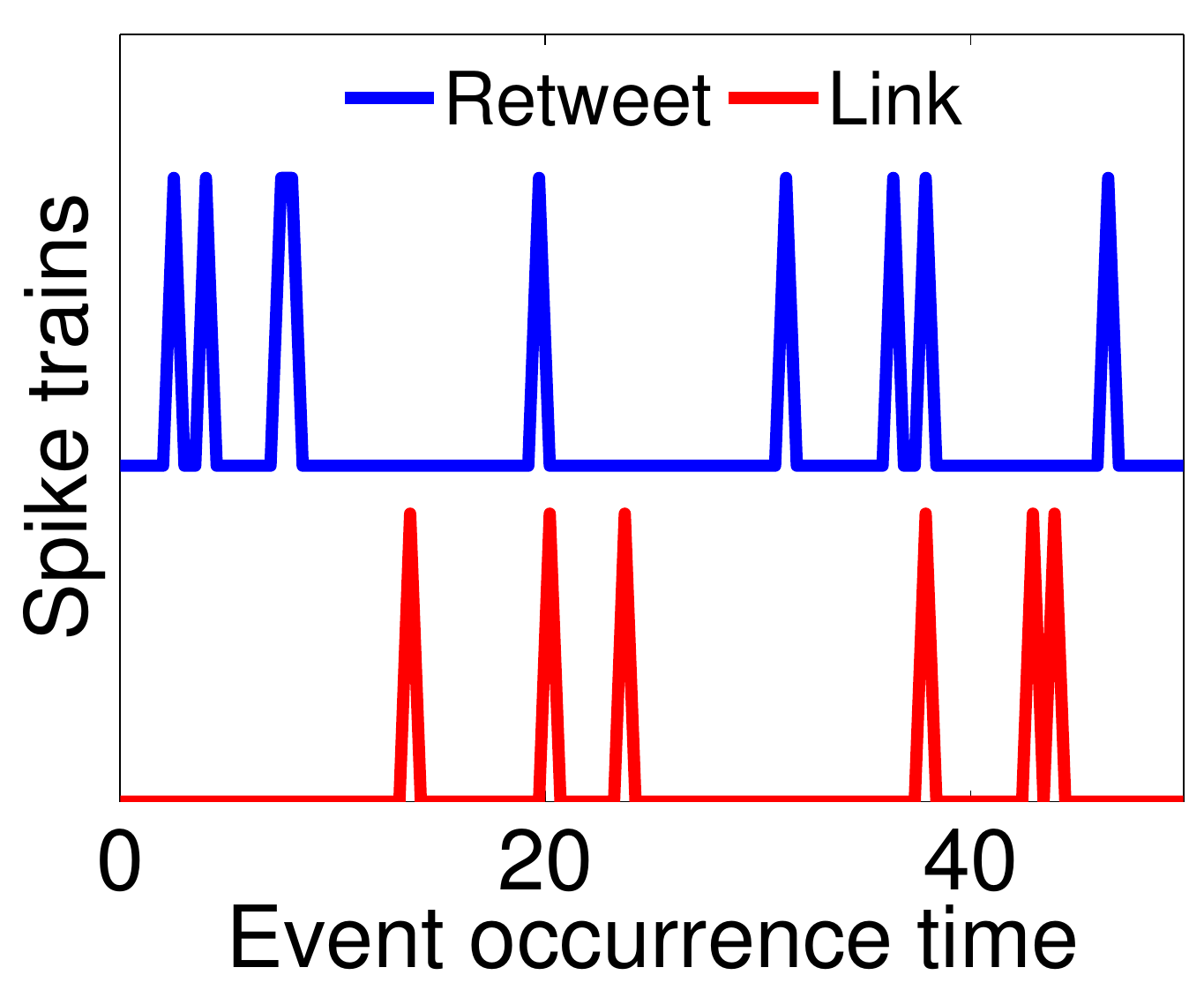} &
                   \hspace{-5mm}
          \includegraphics[width=0.24\textwidth]{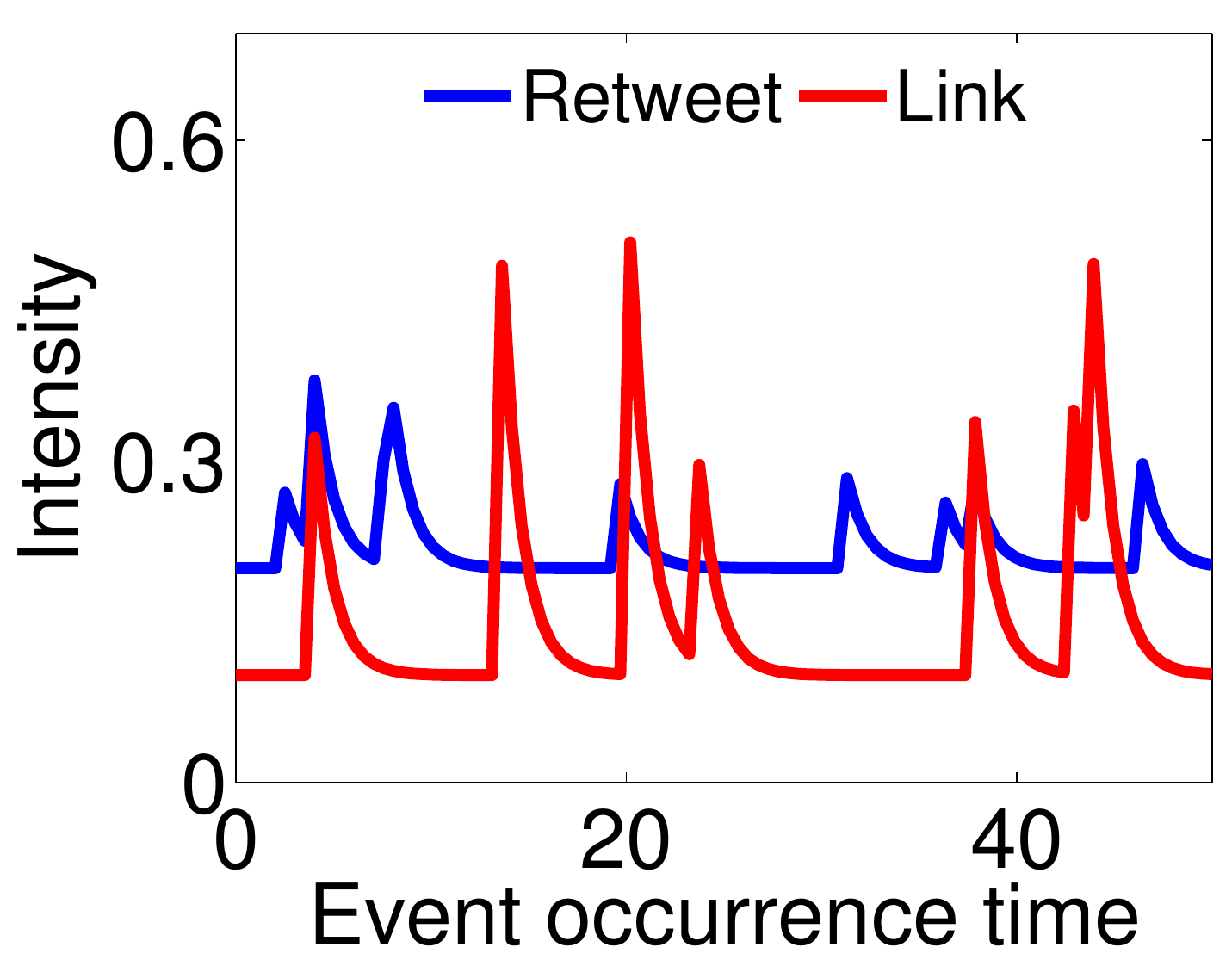}
          \vspace{3mm}
           \\
           (e) & (f) & (g) & (h) \\
        \end{tabular}
        \caption{Link and retweet behavior of 4 typical users in the real-world dataset. Panels (a,c,e,g) show the spike trains of link and retweet events and Panels (b,d,f,h) show the estimated link and retweet 
        intensities}
        \label{fig:coevolution-behavior-real}
\end{figure}
\begin{figure} [t]
        \centering
        \begin{tabular}{cccc}
          \includegraphics[width=0.24\textwidth]{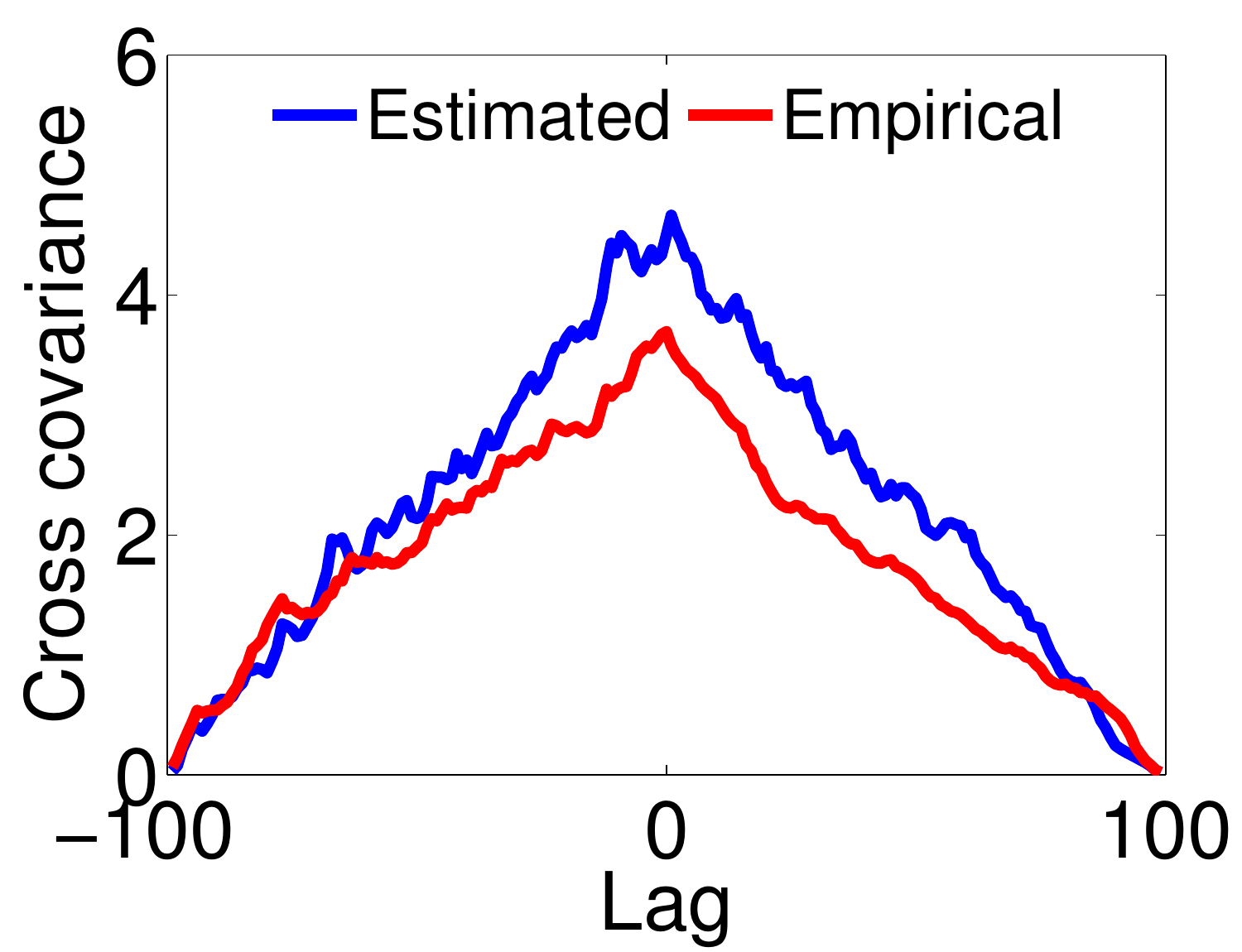} \hspace{-3mm} &
          \includegraphics[width=0.24\textwidth]{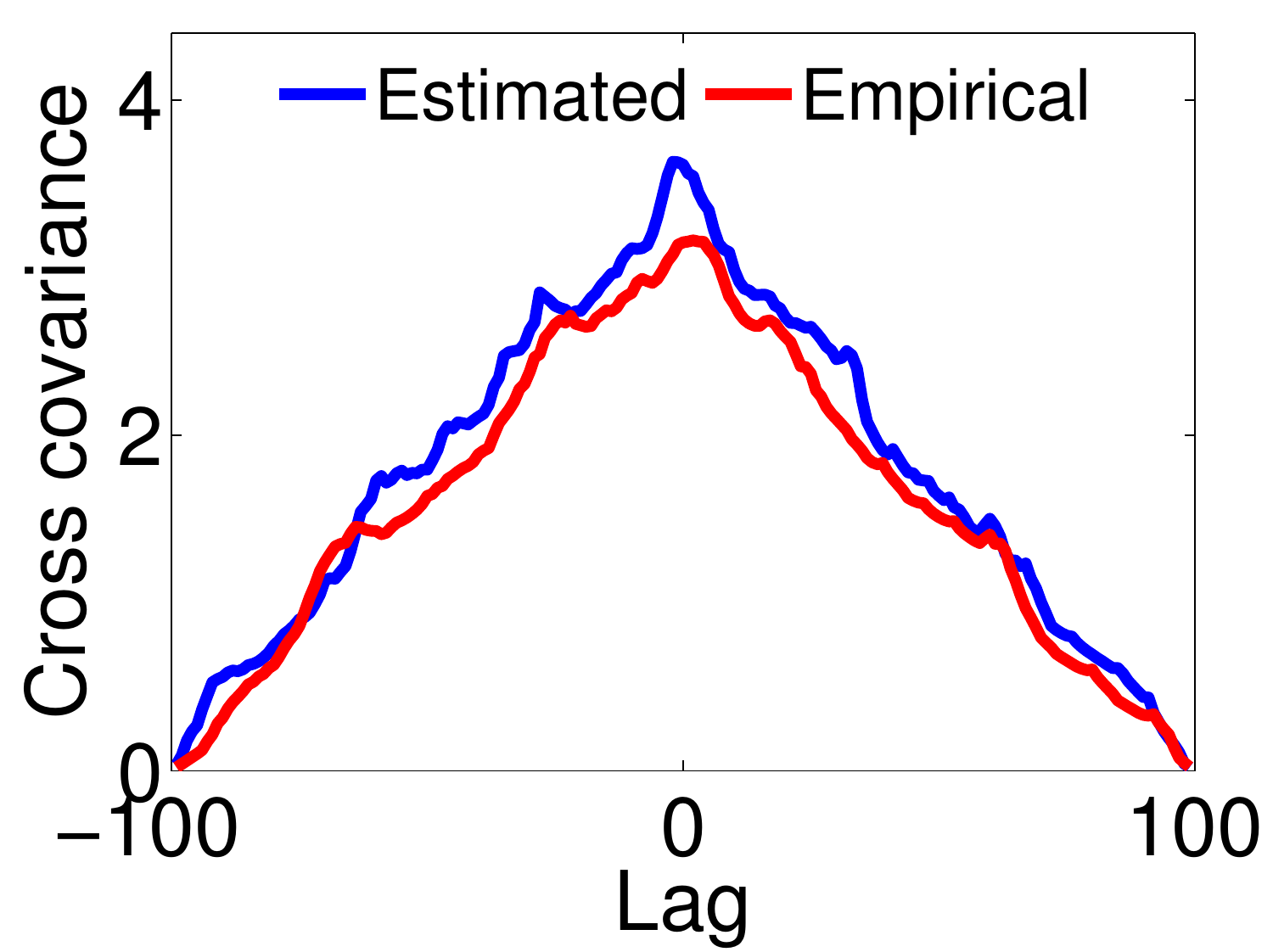}  \hspace{-3mm} &
           \includegraphics[width=0.24\textwidth]{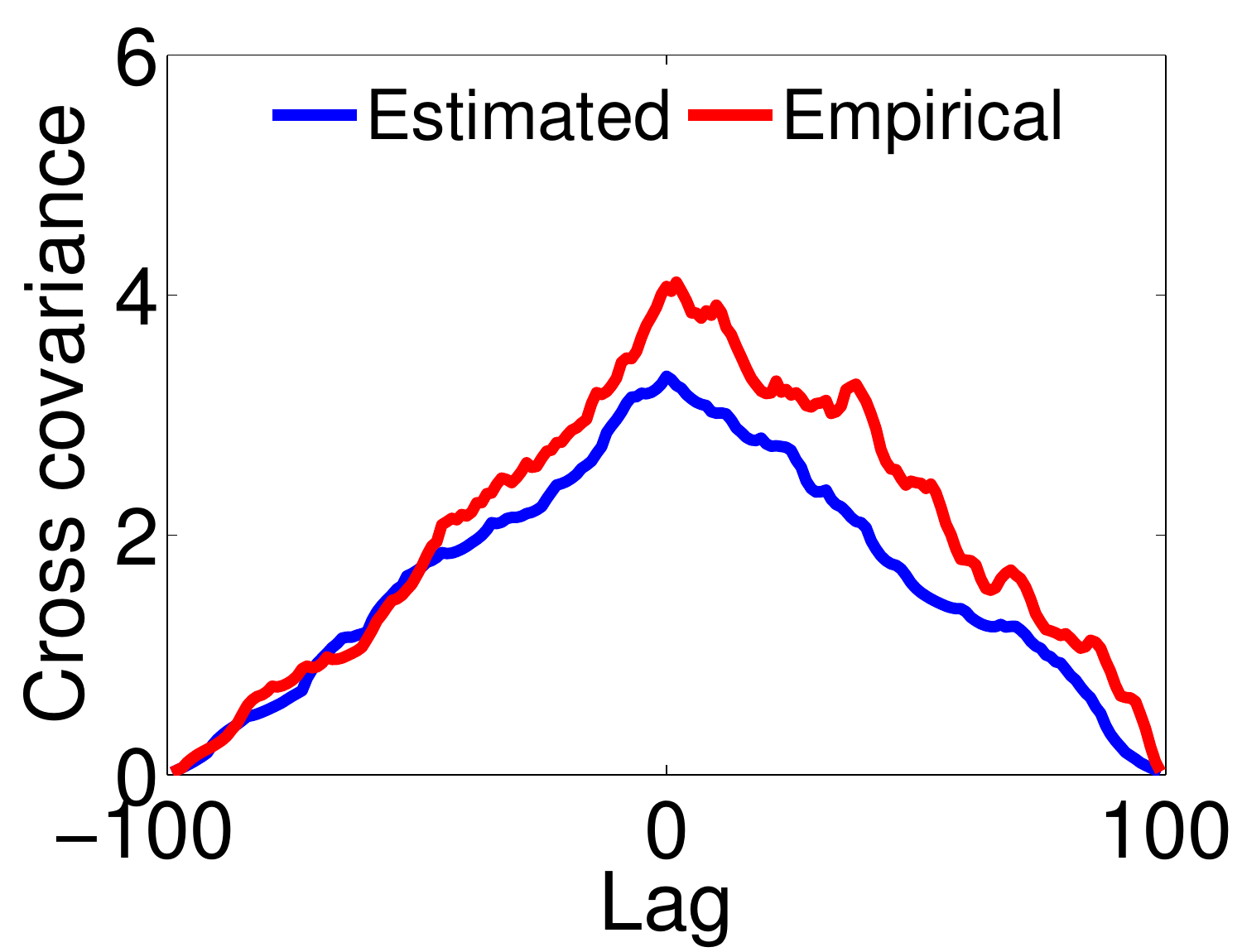} \hspace{-3mm}  &
          \includegraphics[width=0.24\textwidth]{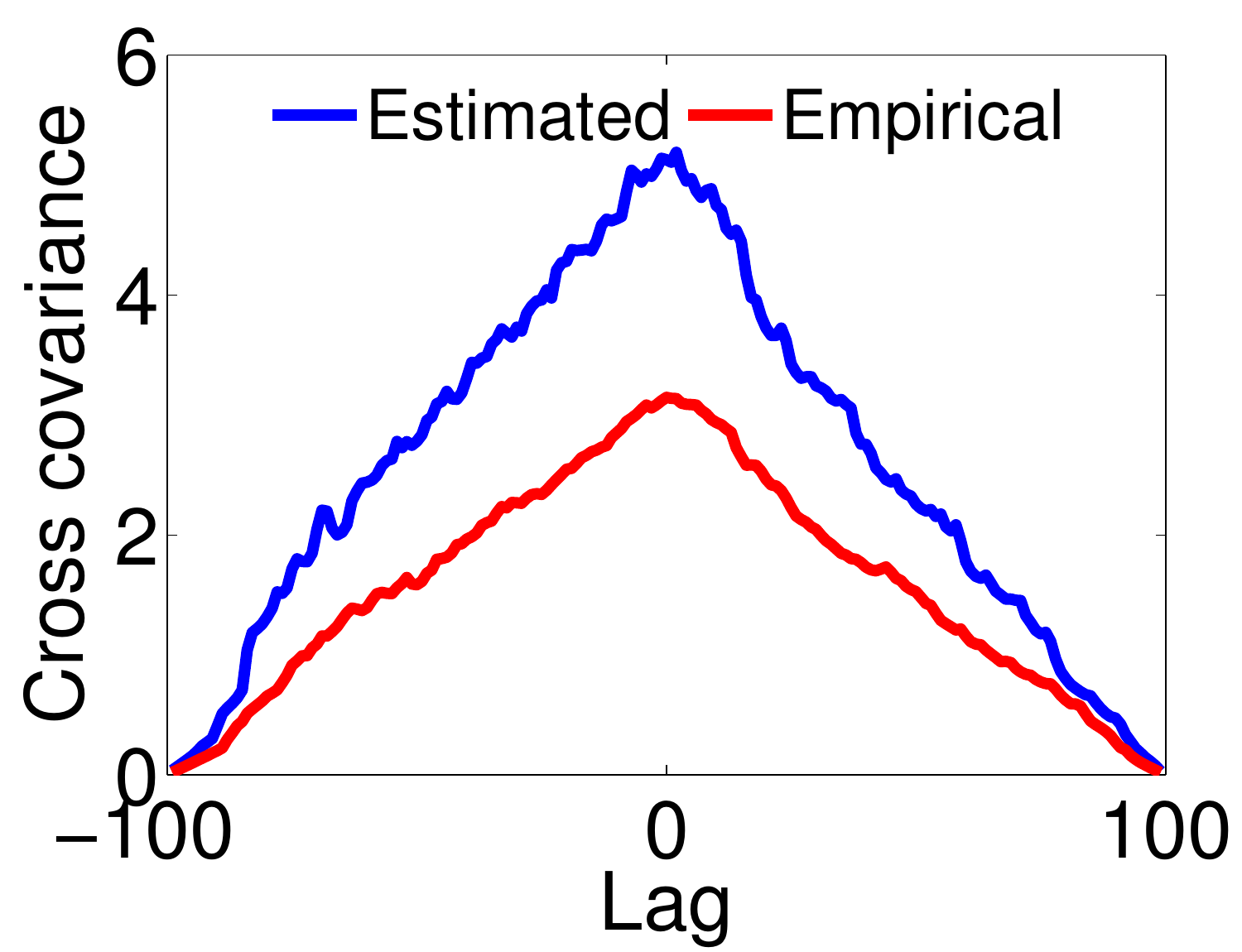} \\    
          (a) & (b) & (c) & (d) \\   
        \end{tabular}
        \caption{Empirical and simulated cross covariance of link and retweet intensities for 4 typical users. }
        \label{fig:crosscovariance-realdata}
\end{figure}

\subsection{Retweet and Link Coevolution} Figures~\ref{fig:coevolution-behavior-real} visualizes the retweet and link events, aggregated across different targets, and the corresponding intensities given by 
our trained model for four source nodes, picked at random.
Here, it is already apparent that retweets (of his posts) and link creations (to him) are clustered in time and often follow each other, and our fitted model intensities successfully track such behavior.
Further, Figure~\ref{fig:crosscovariance-realdata} compares the cross-covariance between the empirical retweet and link creation intensities and between the retweet and link creation intensities
given by our trained model, computed across multiple realizations, for the same nodes. 
For all nodes, the similarity between both cross-covariances is striking and both has their peak around 0, $\ie$, retweets and link creations are highly correlated and co-evolve over time. 
For ease of exposition, as in Section~\ref{sec:properties}, we illustrated co-evolution using four nodes, however, we found consistent results across nodes.

To further verify that our model can capture the coevolution, we compute the average value of the empirical cross covariance function, denoted by $m_{cc}$, per user. 
Intuitively, one could expect that our model estimation method should assign higher $\alpha$ and/or $\beta$ values to users with high $m_{cc}$. Figure~\ref{fig:cross-vs-params} confirms
this intuition on 1,000 users, picked at random. Whenever a user has high $\alpha$ and/or $\beta$ value, she exhibits a high cross covariance between her created links and retweets.

\begin{figure} [t]
        \centering
        \begin{tabular}{c}
         \includegraphics[width=0.25\textwidth]{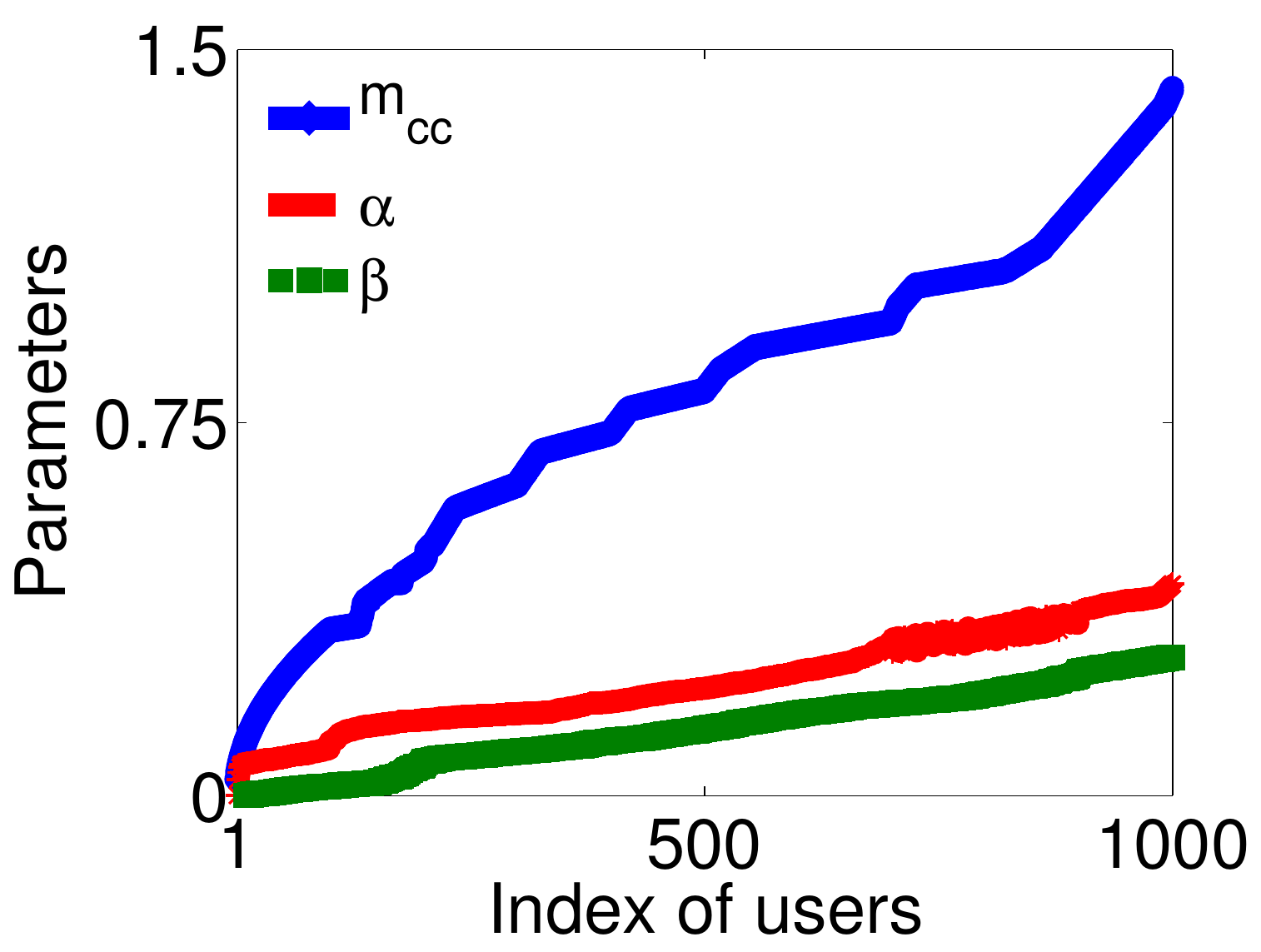}      
        \end{tabular}
        \caption{Empirical cross covariance and learned model parameters for 1,000 users, picked at random}
        \label{fig:cross-vs-params}
\end{figure}

\subsection{Link prediction}
We use our model to predict the identity of the source for each test link event, given the historical (link and retweet) events before the time of the prediction, and compare its performance with the same 
two state of the art methods as in the synthetic experiments, TRF~\cite{AntDov13} and WENG~\cite{WenRatPerGonCasBonSchMenFla13}. 
%
%
%

We evaluate the performance by computing the pro\-ba\-bility of all potential links using different methods, and then compute (i) the average rank of all true (test) events (AvgRank) and, (ii) the success probability 
(SP) that the true (test) events rank among the top-1 potential events at each test time (Top-1).
We summarize the results in Figure~\ref{fig:prediction-realdata}(a-b), where we consider an increasing number of training retweet/tweet events.
Our model outperforms TRF and WENG consistently. For example, for $8 \cdot 10^{4}$ training events, our model achieves a SP $2.5$x times larger than TRF and WENG.

\subsection{Activity prediction}
We use our model to predict the identity of the node that generates each test diffusion event, given the historical events before the time of the prediction, and compare its performance with a baseline consisting of 
a Hawkes process without network evolution. For the Hawkes baseline, we take a snapshot of the network right before the prediction time, and use all historical retweeting events to fit the model. 
Here, we evaluate the performance the via the same two measures as in the link prediction task and summarize the results in Figure~\ref{fig:prediction-realdata}(c-d) against an increasing number of 
training events.
The results show that, by modeling the co-evolutionary dynamics, our model performs significantly better than the baseline.
\begin{figure} [t]
        \centering
        \begin{tabular}{cccc}
          \includegraphics[width=0.23\textwidth]{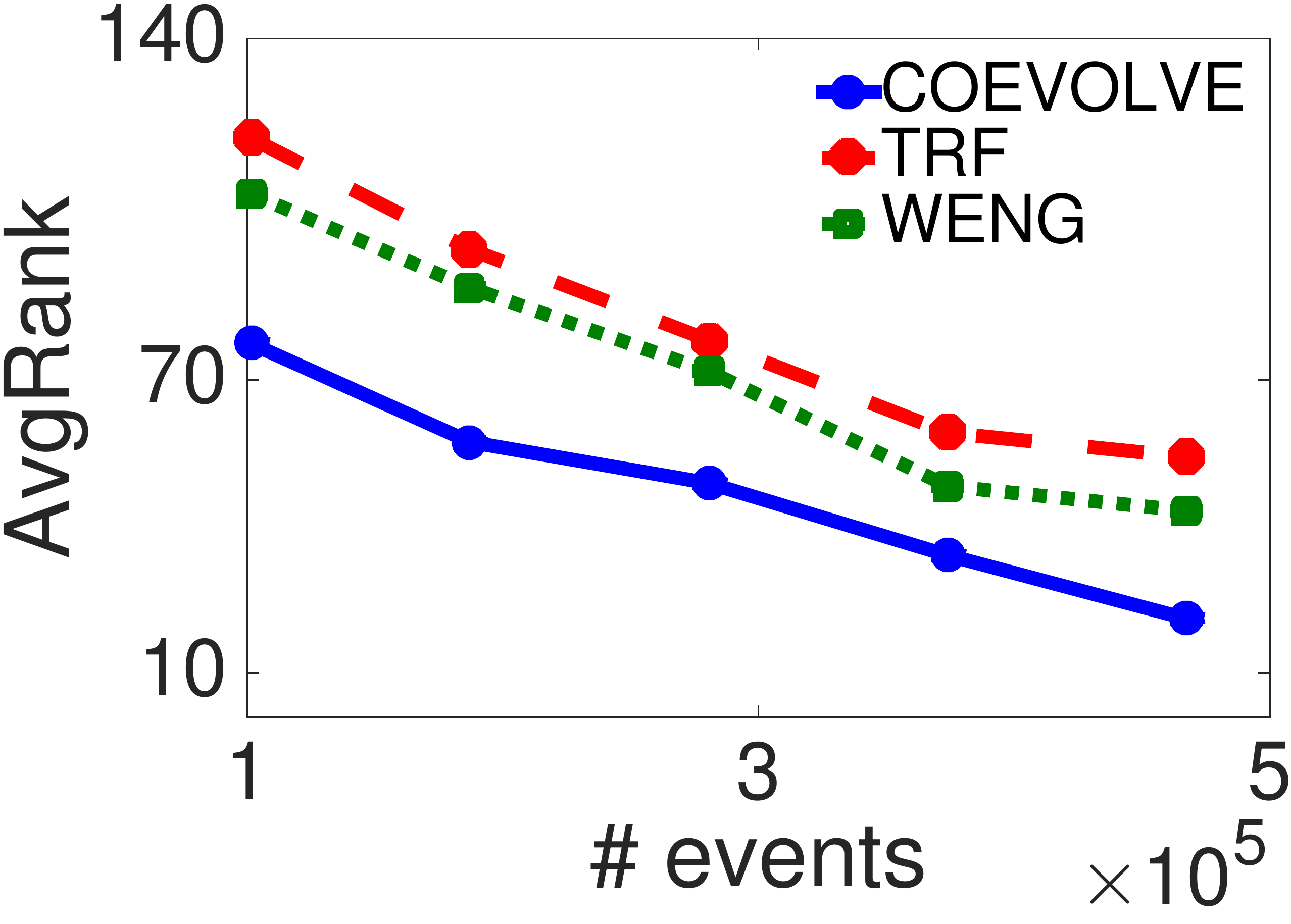} \hspace{-2mm} &
          \includegraphics[width=0.23\textwidth]{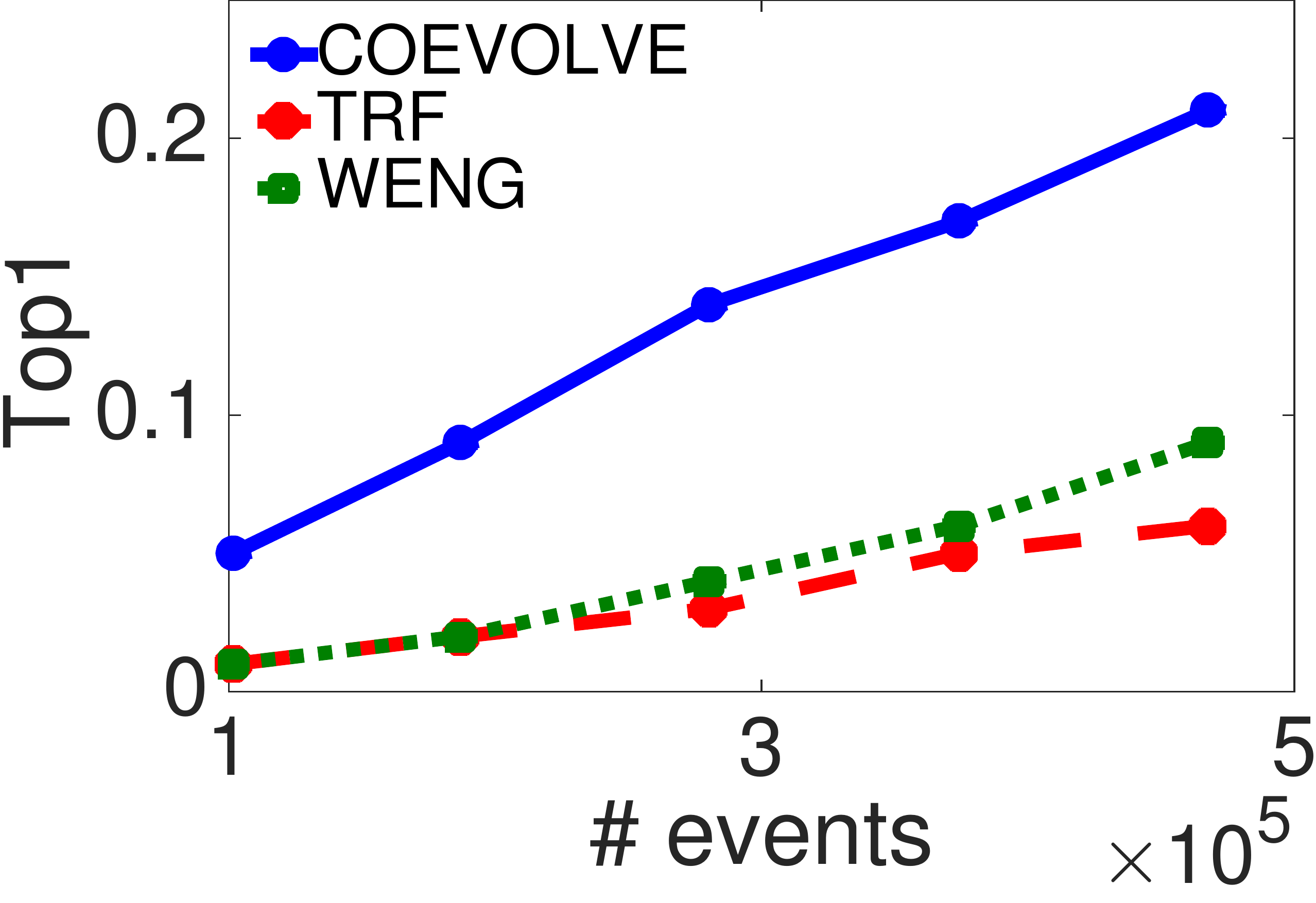}  \hspace{-2mm} &
           \includegraphics[width=0.23\textwidth]{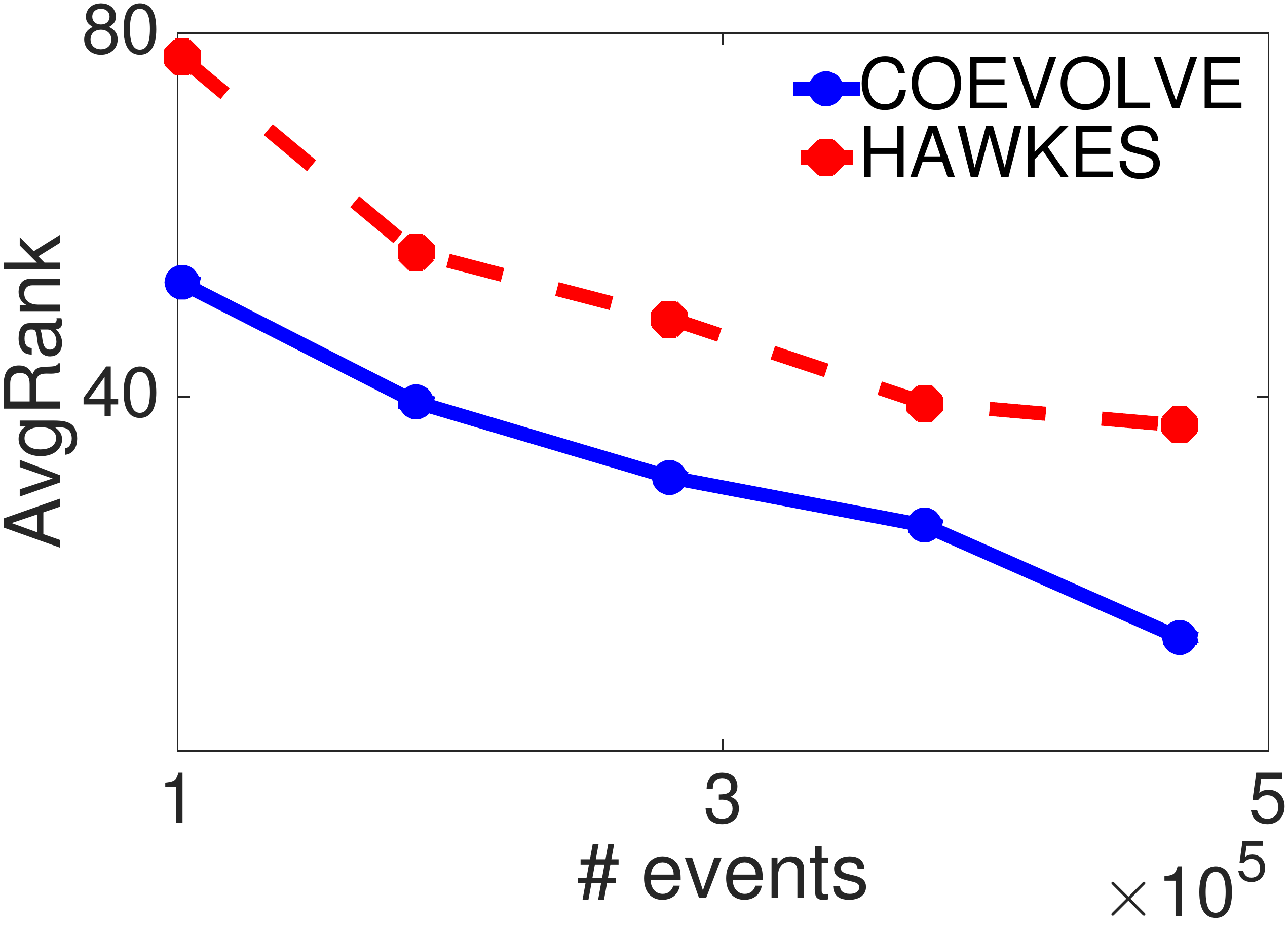} \hspace{-2mm}  &
          \includegraphics[width=0.23\textwidth]{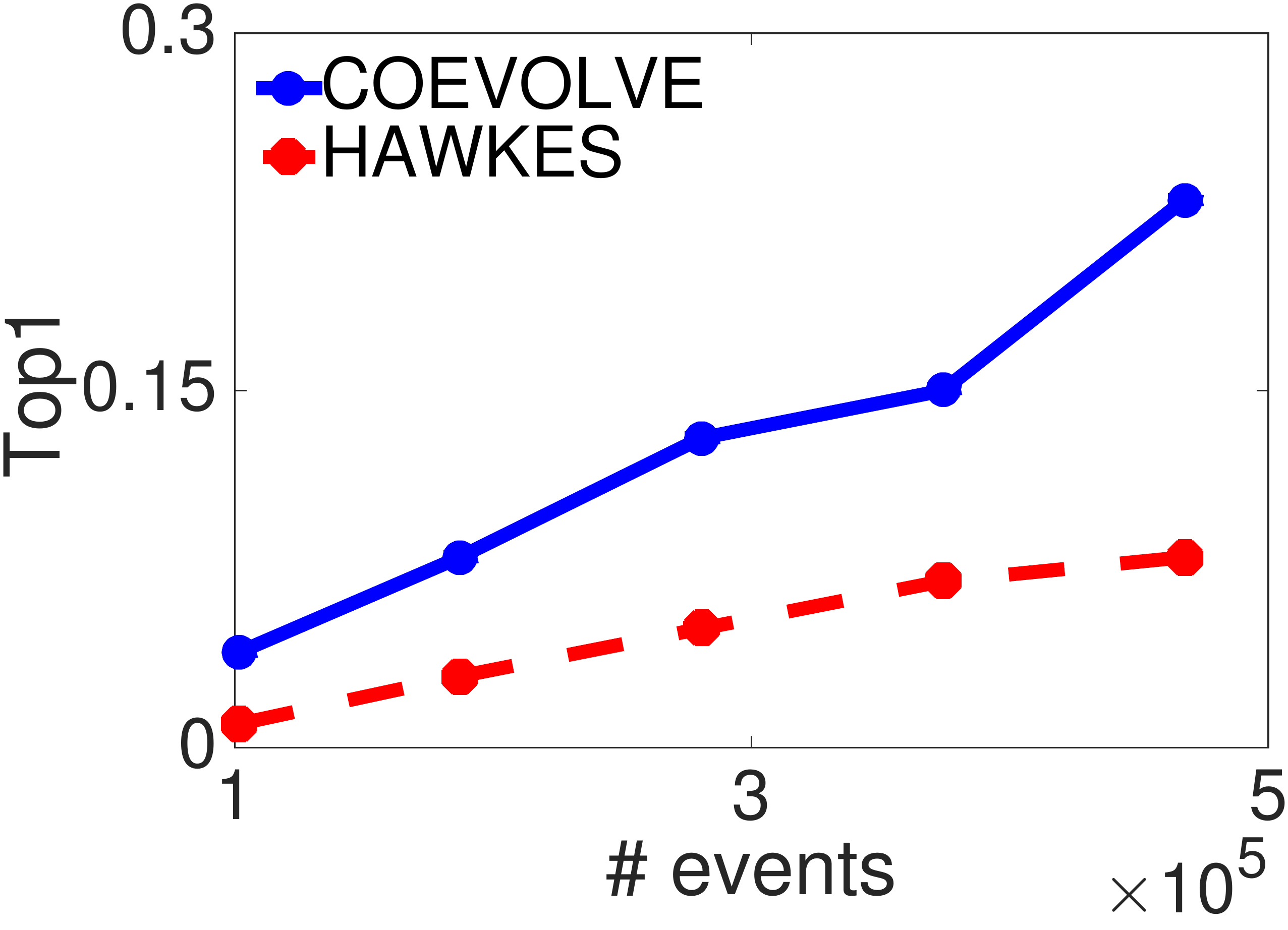} \\
          (a) Links: AR &  (b) Links: Top-1 & (c) Activity: AR & Activity: Top-1 \\         
        \end{tabular}
        \caption{Prediction performance in the Twitter dataset by means of average rank (AR) and success probability that the true (test) events rank among the top-1 events (Top-1).}
        \label{fig:prediction-realdata}
\end{figure}

\subsection{Model Checking}
Given all the subsequent event times generated using a Hawkes process, \ie, $t_i$ and $t_{i+1}$, according to the time changing theorem~\cite{DalVer2007}, the intensity integrals
$\int_{t_i}^{t_{i+1}} \lambda(t) \, dt$ should conform to the unit-rate exponential distribution.
Figure~\ref{fig:qqplots} presents the quantiles of the intensity integrals computed using intensities with the parameters estimated from the real Twitter data against the quantiles 
of the unit-rate exponential distribution. 
It clearly shows that the points approximately lie on the same line, giving empirical evidence that a Hawkes process is the right model to capture the real dynamics.

\begin{figure} [t]
	\centering
	\begin{tabular}{c c}
		\includegraphics[width=0.23\textwidth]{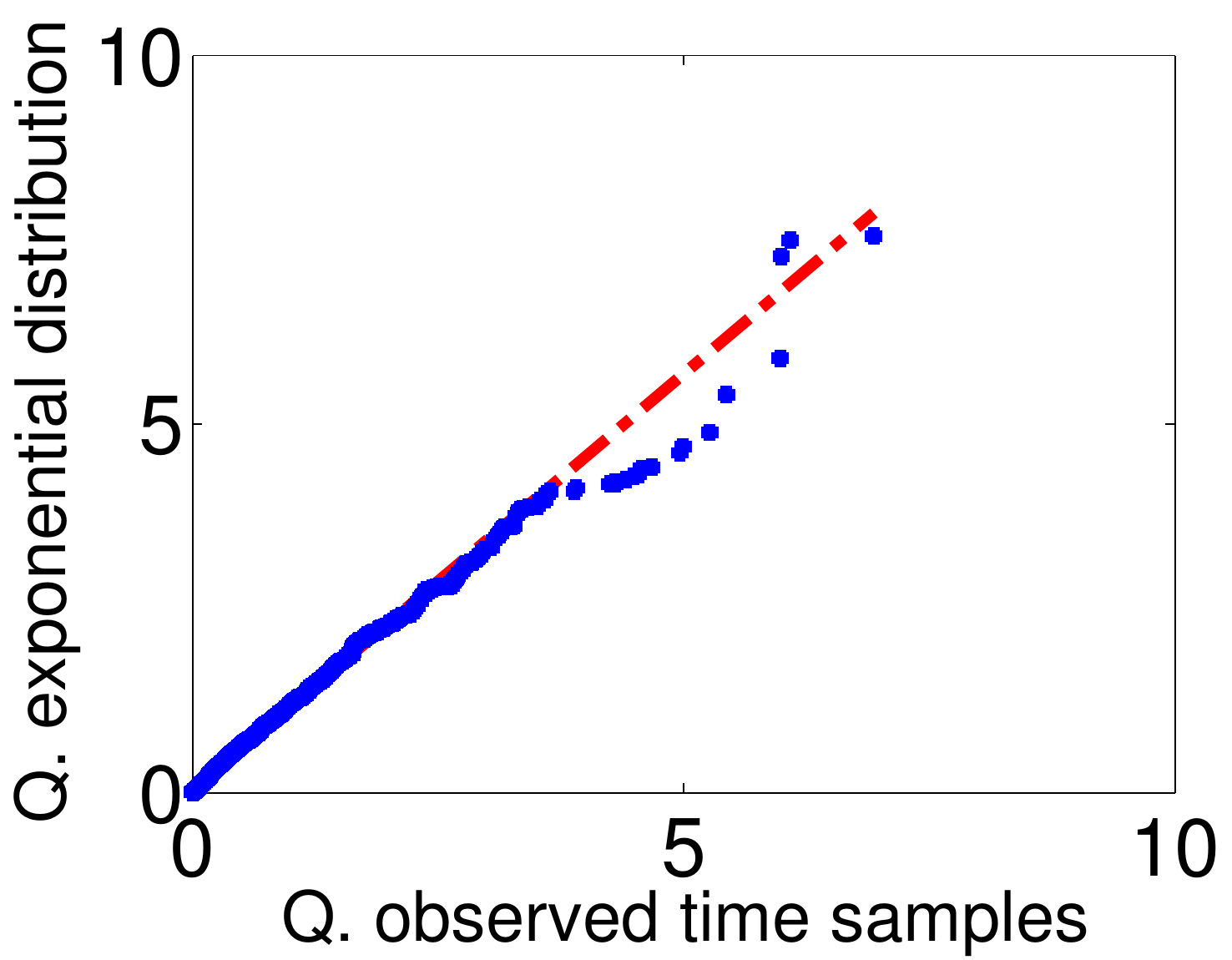} &
		\includegraphics[width=0.23\textwidth]{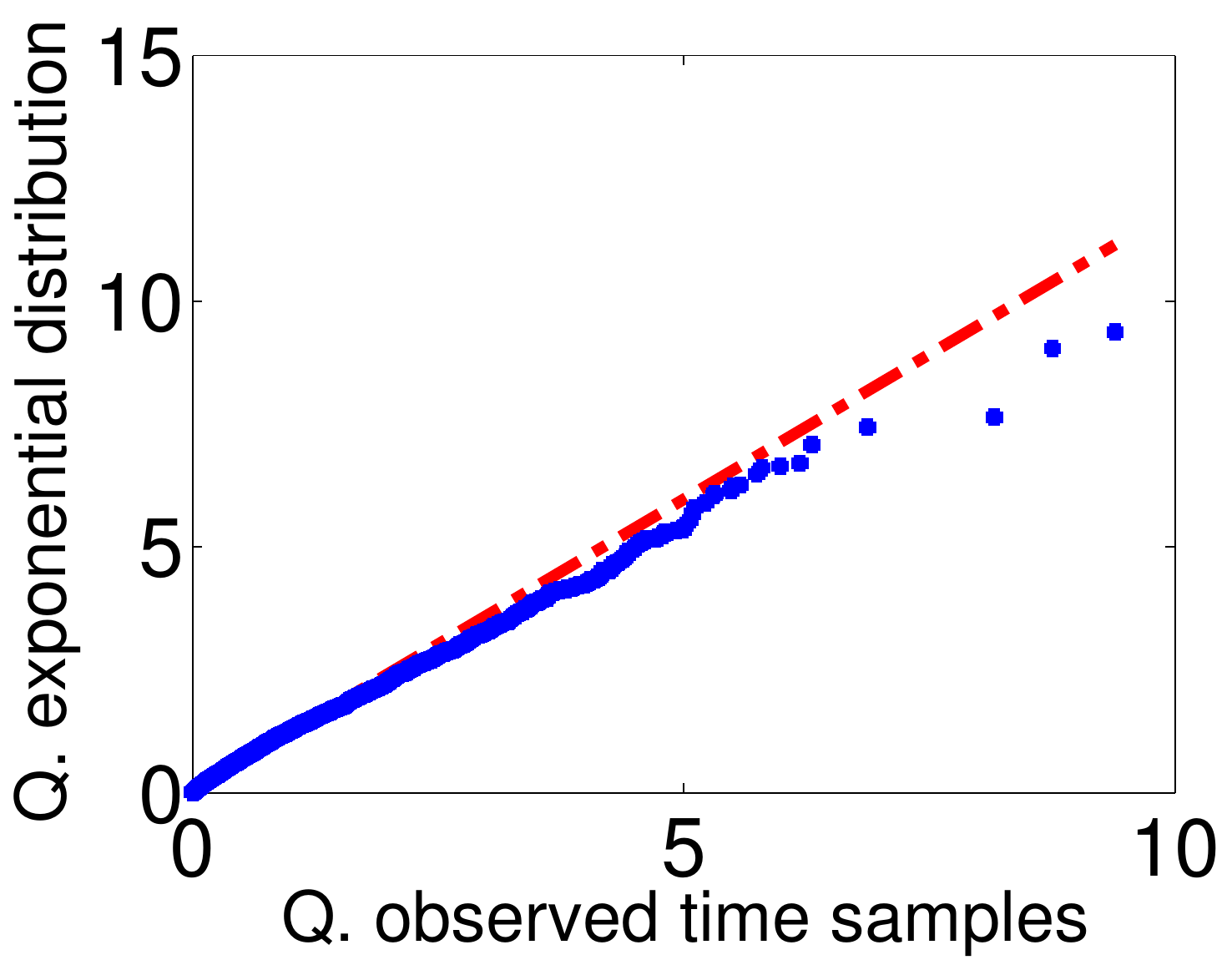} \\
		(a) Link process & (b) Retweet process   
	\end{tabular}
	\caption{Quantile plots of the intensity integrals from the real link and retweet event time}
	\label{fig:qqplots}
\end{figure}




\section{Related Work} 
\label{sec:related}
In this section, we survey related works in modeling temporal networks followed by a subsection on co-evolution dynamics. Next, we review the literature on information diffusion models. Finally, we conclude this section by works that are closely related and are developed for almost the same goal.

\xhdr{Temporal Networks}
Much effort has been devoted to modeling the evolution of social networks~\cite{phan2015natural, Doreian13, wang2011time, Newman10, Barrat08}. Of the proposed methods in characterizing link creation, triadic closure~\cite{Gra73} is a simple but powerful principle to model the evolution based on shared friends.
Modeling timing and rich features of social interactions has been attracting increasing interest in the social network modeling community~\cite{GolZheFieAir10}. However, most of these models use timing information as discrete indices. 
The dynamics of the resulting time-discretized model can be quite sensitive to the chosen discretization time steps;
Too coarse a discretization will miss important dynamic features of the process, and too fine a discretization will increase the computational and inference costs of the algorithms.
In contrast, the events we try to model tend to be asynchronous with a number of different time scales.
\cite{Heise89} used rule-based methods to model the evolution of the graph over time. \cite{Girvan02} analyzed community structure over time and \cite{LesBacKle09} studied the interaction of the friendship graph among group members and group growth.
Recently, \cite{SniLuc06} used a Cox-intensity Poisson model with exponential random graphs to model friendship dynamics. \cite{BraLerSni} extended this model to the temporal sequence of interactions that take place in the social network, but with insufficient model flexibility, 
and limited scalability. 
Modeling temporal dynamics of interactions in this way provides new opportunities for identifying network topology at multiple scales~\cite{GhoLer12} and for early detection of popular resources~\cite{HogLer12,LerGalVerSteetal11}. 
However, these works largely fail to model the interdependency between events generated by different users, which is one of the focuses of our proposed framework. 
Most of this line of work is summarized in a recent survey~\cite{Holme15}, with a short section devoted to point process based approaches.

\xhdr{Co-evolution Dynamics}
In machine learning and several other communities, both the dynamics on the network and the dynamics of the network have been extensively studied, and
combining the two is a natural next step. For example, \cite{bhattacharya2015analyzing} claimed that content generation in social networks is influenced not  just  by  their  personal features  like  age  and  gender,  but  also  by  their social network structure. Furthermore, research has been done to address the co-evolution problems,
for example, in the complex network literature, under the name of {\it adaptive system}~\cite{GroBla08,GroSay09,SayPesSchBusetal13}. 
The main premise is that the evolution of the topology depends on the dynamics of the nodes in the network, and a feedback loop can be created between the two, which allows dynamical exchange of information. 
It has been shown that adaptive networks are capable of self-organizing towards dynamically critical states, like phase transitions by the interplay between the two processes on different time scales~\cite{BorRoh00}. 
In a different context, epidemiologists have found that nodes may rewire their links to try to avoid contact with the infected ones~\cite{GroDliBla06,ZanRis08}. 
Co-evolutionary models have been also developed for collective opinion formation, investigating whether the coevolutionary dynamics will eventually lead to consensus or fragmentation of the population~\cite{ZscBohSeiHueetal12}. However, this line of research tends to be less data-driven.
Moreover, although the general nonlinear dynamic-system based methods usually address co-evolutionary phenomena that are macroscopic in nature, they lack the inference power of statistical 
generative models which are more adapted to teasing out microscopic details from the data. 
Finally, we would also like to mention a different line of research exemplified by the actor-oriented models developed by~\cite{Snijders14}, where a continuous-time Markov chain on the space of directed networks is specified by local node-centric probabilistic link change rules, and MCMC and method of moments are used for parameter estimation. Hawkes processes we used are generally non-Markovian and making use of event history far into the past.

\xhdr{Information Diffusion} The presence of timing information in event data and the ability to model such information bring up the interesting question of how to use the learned model for time-sensitive inference or decision making. Furthermore, the development of online social networks has attracted a lot of empirical studies of the online influence patterns of online communities~\cite{AgaLiuTanYu08, GruGuhLibTom04, SunCheLiuWanetal11, GuoBluWalHel15}, micro blogs~\cite{WenLimJiaHe10,BakHofMasWat11} and so on. However, these works usually consider only relatively simple models for the influence, which may not be very predictive. For more mathematically oriented works, based on information cascades (a special case of asynchronous event data) from social networks, discrete-time diffusion models have been fitted to the cascades~\cite{SaiNakKim08,GoyBonLak10} and used for decision making, such as identifying influencer~\cite{AgaLiuTanYu08}, maximizing information spread~\cite{KemKleTar03,RodSch12}, and 
marketing planing~\cite{RicDom02, DomRic01,BhaGoyLak12,BhaChaPar10}. 
Several recent experimental comparisons on both synthetic and real world data showed that continuous-time models yield significant improvement in settings such as recovering hidden diffusion network topologies from cascade data~\cite{DuSonYuaSmo12,GomBalSch11,YanZha13}, predicting the timings of future events~\cite{DuSonWooZha13,GomLesSch13}, finding source of information cascades~\cite{FarGomDuZamZhaSon15}. Besides this, Point process modeling of activity in network is becoming increasingly popular \cite{lian2015multitask, parikh2012conjoint, hall2014tracking}. 
These time-sensitive modeling and decision making problems can usually be framed into optimization problems and are usually difficult to solve. This brings up interesting optimization problems, such as efficient submodular function optimization with provable guarantees~\cite{GoyBonLakVen10,KemKleTar03}, sampling methods~\cite{Lian2014, gunawardana2011model} for inference and prediction, and convex framework proposed in~\cite{FarDuGomValZhaSon14} to make decisions to shape the activity to a variety of objectives. Furthermore, the high dimensional nature of modern event data makes the evaluation of objective function of the optimization problem even more expensive. Therefore, more 
accurate modeling and sophisticated algorithm needed to be designed to tackle the challenges posed by modern event data applications.

The work most closely related to ours is the empirical study of information diffusion and network evolution~\cite{GroBla08,SinWagStr12,WenRatPerGonCasBonSchMenFla13,AntDov13,MyeLes14}.
Among them,~\cite{WenRatPerGonCasBonSchMenFla13} was the first to show experimental evidence  that information diffusion influences network evolution in microblogging sites both at system-wide and individual levels. 
In particular, they studied \emph{Yahoo! Meme}, a social micro-blogging site similar to \-Twitter, which was active between 2009 and 2012, and showed that 
the likelihood that a user $u$ starts following a user $s$ increases with the number of messages from $s$ seen by $u$. 
\cite{AntDov13} investiga\-ted the temporal and statistical characteristics of retweet-driven connec\-tions within the Twitter network and then identified the number of retweets as a key factor to infer such connections.
\cite{MyeLes14} showed that the Twitter network can be characterized by steady rates of change, interrupted by sudden bursts of new connections, triggered by retweet cascades. 
They also developed a method to predict which retweets are more likely to trigger these bursts.
Finally, \cite{TraFarSonZha15} utilized multivariate Hawkes process to establish a connection between temporal properties of activities and the structure of the network. In contrast to our work they studied the static properties, $\eg$, community structure and inferred the latent clusters using the observed activities.

However, there are fundamental differences between the above-mentioned studies and our work. 
First, they only cha\-rac\-te\-rize the effect that information diffusion has on the network dynamics, but not the bidirectional
influence. 
%
%
In contrast, our probabilistic generative model takes into account the bidirectional influence between information diffusion and network dynamics. 
Second, previous studies are mostly empirical and only make binary predictions on link creation events. For example, the work of ~\cite{WenRatPerGonCasBonSchMenFla13, AntDov13} 
predict whether a new link will be created based on the number of retweets; and, \cite{MyeLes14} predict whether a burst of new links will occur based on the number of retweets and users'{} similarity.
However, our model is able to learn parameters from real world data, and predict the precise timing of both diffusion and new link events.


\section{Extensions}
\label{sec:extensions}
The basic model presented in Section~\ref{sec:model} is just a show-case of the potential of point processes in modeling networks and processes over them. In this section, we extend our model in a variety of ways. 
More specifically, we explain how the model can be augmented to support link removal, node birth and death, and connection specific parameters. 
We did not perform experiments with these extensions because our real-world dataset does not contain information regarding to link removal and node birth and death. Curating a comprehensive dataset that can be 
used in modeling all these aspects of networks is left as interesting future work.

\subsection{Link deletion} 
We can generalize our model to support link deletion by introducing an intensity matrix $\Xib^{*}(t) = \rbr{\xib_{us}^*(t)}_{u,s \in [m]}$ 
and model each individual intensity as a survival process. 
Assume $\Ab^+(t)$ is the previously defined counting matrix $\Ab(t)$, which indicates the existence of an edge at time $t$. Then, we introduce a new counting matrix $\Ab^-(t) = \rbr{A^-_{us}(t)}_{u,s \in [m]}$, which indicates the lack of an edge at time $t$, and we define it via 
its intensity function as
\begin{align}
\EE[d\Ab^-(t)\, |\,  \Hcal^r(t) \cup \Hcal^l(t)] = \Xib^*(t) \, dt,
\end{align}
Then, we define the intensity as 
\begin{align}
\label{eq:removal-matrix-intensity}
\xib_{us}^*(t) = A^+_{us}(t)(\zeta_u + \nu_s \, \sum_{v \in \Fcal{u}} \kappa_{\omega_3}(t)\star dA^{-}_{vs}(t) ),
\end{align} 
where the term $A^+_{us}(t)$ guarantees that the link has positive intensity to be removed only if it already exists, just like the term $1-A_{us}(t)$ in Equation \eqref{eq:network-evolve-intensity},
%
%
the parameter $\zeta_u $ is the base rate of link deletion and $\nu_s \, \sum_{v \in \Fcal{u}} \kappa_{\omega_3}(t)\star dA^{-}_{vs}(t)$
is the increased link deletion intensity due to increased number of followees of $u$ who decided to unfollow $s$. This is an excitation term due to deleted links to source $s$; given $s$ is unfollowed 
by some followees of $u$, then $u$ may find $s$ not a good source of information too.

Given a pair of nodes $(u, s)$, the process starts with $A^+_{us}(t) = 0$. Whenever a link is created this process ends and a removal process  $A^-_{us}(t)$ starts. Similarly, when the removal process fires, 
the connection is removed and a new link creation process is instantiated. These two processes interleave until the end.



\subsection{Node birth and death} 
We can augment our model to consider the number of nodes $m(t)$ to change over time: 
\begin{align}
m(t) = m_b(t) - m_d(t)
\end{align}
where $m_b(t)$ and $m_d(t)$ are counting processes modeling the numbers of nodes that join and left the network till time $t$, respectively. The way we construct $m_b(t)$ and $m_d(t)$ guarantees that $m(t)$ is always non-negative.

The birth process, $m_b(t)$, is characterized by a conditional intensity function $\phi^*(t)$:
\begin{align}
\EE[d m_b(t)\, |\,  \Hcal^r(t) \cup \Hcal^l(t)] = \phi^*(t) \, dt,
\end{align}
where
\begin{align}
\phi^*(t) = \epsilon + \theta \sum_{u,s \in [m(t)]} \kappa_{\omega_4} (t) \star d N_{us}(t),
\end{align} 
Here, $\epsilon$ is the constant rate of arrival and $\theta \sum_{u,s \in [m(t)]} \kappa_{\omega_4} (t) \star d N_{us}(t)$ is the increased rate of node arrival due to the increased activity of nodes. Intuitively, the higher the overall activity in 
the existing network, the larger the number of new users. 

The construction of the death process, $m_d(t)$, is more involved. 
Every time a new user joins the network, we start a survival process that controls whether she leaves the network. Thus, we can stack all these survival processes in a vector, $\lb(t) = \rbr{l_u(t)}_{u \in [m]}$, characterized by a multidimensional conditional 
intensity function $\sigmab^*(t) = \rbr{\sigma_u(t)}_{u \in [m_b(t)]}$:
\begin{align}
\EE[d \lb(t) | \Hcal^r(t) \cup \Hcal^l(t)]  = \sigmab^{*}(t) \, dt,
\end{align} 

Intuitively, we expect the nodes with lower activity to be more likely to leave the network and thus its conditional intensity function to adopt the following form:
\begin{align}
\sigma_{u}^*(t) = (1-l_u(t)) \left[ \sum_{j=1}^J \pi_{j} g_j(t) + \rbr{h(t) - \sum_{s \in [m(t)]} \kappa_{\omega_5} (t) \star d N_{us}(t) }_+ \right],
\end{align} 
where the term $\rbr{1-l_u(t)}$ ensures that a node is deleted only once, $ \sum_{j=1}^J \pi_{j} g_j(t) $ is the history-independent typical rate of death, shared across nodes, 
which we represent by a grid of known temporal kernels, $\cbr{g_j(t)}$ with unknown coefficients, $\cbr{\pi_{j}}$, and the second term is capturing the effect of activity on the
probability of leaving the network. 
More specifically, if a node is not active, we assume its intensity is upper bounded by $h(t)$ and the most active she becomes, the lower its probability of leaving the network and
the larger the term $\sum_{s \in [m(t)]} \kappa_{\omega_5} (t) \star d N_{us}(t)$. 
The hinge function $(\cdot)_{+}$ guarantees the intensity is always positive.

Then, given the individual death processes the total death process is
\begin{align}
m_d(t) = \sum_{u=1}^{m_b(t)} l_u(t),
\end{align}
which completes the modeling of the time-varying number of nodes.

\subsection{Incorporating features} 
One can simply enrich the model by taking into account the longitudinal or static information of the networked data, \eg,
by conditioning the intensity on additional external features, such as node attributes or edge types.
Let us assume each user $u$ comes with a $K$-dimensional feature vector $\xb_u$ including properties such as her age, job, location, number of followers, number of tweets, etc. 

Then, we can augment the information diffusion intensity as follows. We introduce a $K$-dimensional link intensity parameter $\etab_u$ in which each dimension reflects the contribution of the corresponding element in the feature vector to the intensity 
and replace the baseline rate $\eta_u$ by $\etab_u^{\top} \xb_u$. 
Similarly, we introduce a $K$-dimensional vector $\betab_s$ where each dimension has a corresponding element in the feature vector $\xb_s$ and substitute $\beta_s$ by $\betab_s \xb_s$. Therefore, one can rewrite the original information diffusion intensity given by 
Equation~\eqref{eq:activity-occurrence} as:
\begin{equation}
  \gamma_{us}^*(t) = \II[u = s] \, \etab_u^{\top} \xb_u + \II[u \neq s]\,              \betab_{s}^{\top} \xb_s \sum\nolimits_{v \in \Fcal_u(t)}  \kappa_{\omega_1}(t) \star  \rbr{A_{uv}(t)\,\, dN_{vs}(t)}, 
\end{equation}  

Similarly, we can parameterize the coefficients of the link creation intensity by a $K$-dimensional vector and write the counter-part of Equation~\eqref{eq:network-evolve-intensity-2} incorporating features of the node for computing the intensity:
\begin{equation}
  \lambda_{us}^*(t) = (1-A_{us}(t))(\mub_u^{\top} \xb_u + \alphab_u^{\top} \xb_u  \sum_{v \in \Fcal_u(t)} \kappa_{\omega_2}(t)\star dN_{vs}(t))
\end{equation}

Surprisingly enough, all the results for convexity for parameter learning, and efficient simulation techniques are still valid for this case too. As far as the features contribute to the intensity linearly, the log-likelihood is concave and we can simulate the model as efficiently as the 
original model.

\subsection{Connection specific parameters} 
%
%
Up to this point, the parameters of the link creation and removal, node birth and death and the information diffusion intensities depend on one end point of the interactions. For example $\beta_s$ and $\eta_u$ in the information diffusion intensity 
given by Equation~\eqref{eq:activity-occurrence} only depend on the source and the actor, respectively.
%
However, proceeding with this example, parameters can be made connection specific, \ie, Equation \eqref{eq:activity-occurrence} can be restated as
\begin{equation}
  \gamma_{us}^*(t) = \II[u = s] \, \eta_{us} + \II[u \neq s]\,              \beta_{us} \sum\nolimits_{v \in \Fcal_u(t)}  \kappa_{\omega_1}(t) \star  \rbr{A_{uv}(t)\,\, dN_{vs}(t)}, 
\end{equation}  
where $\eta_{us}$ is the base intensity of $u$ retweeting a tweet originated by $s$ and $\beta_{us}$ is the coefficient of excitement of $u$ to retweet $s$ when one of her followees retweets something 
from $s$.

Given enough computational resources and large amounts of historical data, one can take into account more complex scenarios and larger and more flexible models. For example, the middle user, say $v$, who is along the path of diffusion and forwards the 
tweet originated from $s$ to $u$ can also be taking into consideration, \ie, defining $\beta_{svu}$ as the amount of increase in intensity of user $u$ retweeting from $s$ when user $v$ has just retweeted a post from $s$. All desirable properties of simulation algorithm and parameter estimation method still hold.


\section{Conclusion and Future Works}
\label{sec:discussion}
In this work, we proposed a joint continuous-time model of information diffusion and network evolution, which can capture the coevolutionary dynamics, can mimic the most common static and temporal network patterns 
observed in real-world networks and information diffusion data, and can predict the network evolution and information diffusion more accurately than previous state-of-the-arts.
Using point processes to model intertwined events in information and social networks opens up many interesting venues for future. Our current model is just a show-case of a rich set of possibilities offered by a point process framework, 
which have been rarely explored before in large scale social network mo\-de\-ling. 
There are quite a few directions that remain as future work and are very interesting to explore. For example:
\begin{itemize}
\item A large and diverse range of point processes can also be used instead in the framework and augment the current model without changing the efficiency of simulation and the convexity of parameter estimation.
\item We can incorporate features from previous state of the diffusion or network structure. For example, one can model information overload by adding a nonlinear transfer function on top of the diffusion intensity, or model peer pressure by adding 
a nonlinear transfer function depending on the number of neighbors. 
\item There are situations that the processes are naturally evolve in different time scales. For example, link dynamics is meaningful in the scale of days, however, the resolution in which information propagation occurs is usually in hours or even minutes. 
Developing an efficient mechanism to account for heterogeneity in time resolution would improve the model's ability to predict.

\item We may augment the framework to allow time-varying parameters. The simulation would not be affected and the estimation of time-varying interaction can still be carried out via a convex optimization problem~\cite{ZhoZhaSon13b}.

\item Alternatively, one can use different triggering kernels for the Hawkes processes and learn them to capture finer details of temporal dynamics. 
\end{itemize}
%



\xhdr{Acknowledgement}
The authors would like to thank Demetris Antoniades and Constantine Dovrolis for providing them with the dataset. The research was supported in part by NSF/NIH BIGDATA 1R01GM108341, ONR N00014-15-1-2340, NSF IIS-1218749, NSF CAREER IIS-1350983.


\appendix

\section{Ogata's Algorithm} \label{app:sec:simulation}
\usetikzlibrary{decorations.pathreplacing,calc}
\newcommand{\tikzmark}[1]{\tikz[overlay,remember picture] \node (#1) {};}

\newcommand*{\AddNote}[4]{%
    \begin{tikzpicture}[overlay, remember picture]
        \draw [decoration={brace,amplitude=0.5em},decorate,ultra thick,black]
            ($(#3)!(#1.north)!($(#3)+(0,1)$)$) --  
            ($(#3)!(#2.south)!($(#3)+(0,1)$)$)
                node [align=center, text width=2.9cm, pos=0.5, anchor=west] {#4};
    \end{tikzpicture}
}%

\begin{algorithm}[h!]
\caption{Ogata'{}s Algorithm} \label{alg:ogata}
\begin{algorithmic}[2]
\State {\bf Input:} $U$ dimensional Hawkes process $\cbr{\lambda^*_u(t)}_{u=1\ldots U}$, Due time: $T$
\State {\bf Output:} Set of events: $\Hcal = \cbr{(t_1, u_1), \ldots, (t_n, u_n)}$
\State $t \gets 0$
\State $i \gets 0$
\While{$t < T $} \hspace{2cm}\tikzmark{right}
\State $\lambda^*_{sum}(\tau)  \gets \sum_{u=1}^U \lambda^*_u(\tau)$  \label{while_start}  \tikzmark{top1}
\State $\lambdahat  \gets \max_{t \le \tau \le T} \lambda^*_{sum}(\tau)$ 
\State  $s \sim Exponential( \lambdahat)$ 
\State $t'  \gets t + s$  
\If{$t' \ge T$} 
\State break
\EndIf   \tikzmark{bottom1}
\State $\lambdabar   \gets \lambda^*_{sum}(t')$   \tikzmark{top2}
\State  $d \sim Uniform(0,1)$
\If {$d \times \lambdahat  > \lambdabar $}
\State $t \gets t'$
\State Goto \ref{while_start}
\EndIf   \tikzmark{bottom2}
\State $S \gets 0$   \tikzmark{top3}
\State  $d \sim Uniform(0,1)$
\For{ $u \gets 1$ to $U$}
\State $S \gets S + \lambda^*_u(t')$
\If{$S \ge d$}
\State $i \gets i+1$
\State $u_i \gets u$
\State $t_i \gets t'$
\State $t \gets t'$
\State Goto \ref{while_start}
\EndIf
\EndFor  \tikzmark{bottom3}
\State Given the new event just sampled update intensity functions $\lambda^*_{u}(\tau)$ 
\EndWhile
\end{algorithmic}
\AddNote{top1}{bottom1}{right}{Sampling next event time}
\AddNote{top2}{bottom2}{right}{Rejection test}
\AddNote{top3}{bottom3}{right}{Attribution test}
\end{algorithm}

In this section, we revisit Ogata'{}s algorithm in more details. 
Consider a $U$-dimensional point process in which each dimension $u$ is characterized by a conditional intensity function $\lambda^*_u(t)$.

Ogata'{}s algorithm starts with summing the intensities, $\lambda^*_{sum}(\tau) = \sum_{u=1}^U \lambda^*_u(\tau)$.
Then, assuming we have simulated up to time $t$, the next sample time, $t'$, is the first event drawn from the non-homogenous Poisson process with intensity $\lambda^*_{sum}(\tau)$ which begins at time $t$.
%
%
Here, the algorithm exploits that, given a fixed history, the Hawkes Process is a non-homogenous Poisson process, which runs until the next event happens.
Then, the new event will result in an update of the intensities and a new non-homogenous Poisson process starts.

It can be shown that the waiting time of a non-homogeneous Poisson process is an exponentially distributed random variable with rate equal to integral of the intensity~\cite{Ross06}, $\ie$  
$
s \sim Exponential \rbr{\int_t^{t+s} \lambda^*_{sum} (\tau) \, d \tau}
$.
Thus, the next sample time can be computed as
\begin{align}
t' = \underbrace{t}_{\text{current time}} + \underbrace{s.}_{\text{waiting time for the first event}}
\end{align}
Sampling from a non-homogenous Poisson process is not straight-forward, therefore, Ogata'{}s algorithm uses rejection sampling with a homogenous Poisson process as the proposal distribution. 
More in detail, given
$
\lambdahat = \max_{t \le \tau \le T} \lambda^*_{sum}(\tau),
$
$t'$ is the time of first event of homogenous Poisson Process with rate $\lambdahat$.
Then, we accept the sample time with probability 
$
\lambda^*_{sum}(t')/ \lambdahat.
$
Finally, the dimension firing the event is determined by sampling proportionally to the contribution of the intensity of that user to the total intensity, \ie, $\lambda^*_u(t')/ \lambda^*_{sum}(t')$ for $1\le u \le U$.
%
This procedure is iterated until we reach the end of simulation time $T$. Algorithm~\ref{alg:ogata} presents the complete procedure.

Ogata'{}s algorithm would scale poorly with the dimension of the process, because, after each sample, we would need to re-evaluate the affected intensities and find the upper bound. 
%
As a consequence, a naive implementation to draw $n$ samples require $O(U n^2)$ time complexity, where $U$ is the number of dimensions. This is because for each sample we need 
to find the new summation of intensities, which involves $O(U)$ individual ones, each taking $O(n)$ time to accumulate over this history. In our social networks application, we have $m^2-m$ 
point processes for link creation and $m^2$ ones for retweeting, \ie, $U=O(m^2)$. Therefore, Ogata{}'s algorithm takes $O(m^2 n^2)$ time complexity.

\vskip 0.2in
\bibliographystyle{unsrt}
\bibliography{bibfile}

\begin{thebibliography}{10}

\bibitem{KwaLeeParMoo10}
Haewoon Kwak, Changhyun Lee, Hosung Park, and Sue Moon.
\newblock {W}hat is {T}witter, a social network or a news media?
\newblock In {\em Proceedings of the 19th International Conference on World
  Wide Web}, pages 591--600, New York, NY, USA, 2010. ACM.

\bibitem{CheAdaDowKleLes14}
Justin Cheng, Lada Adamic, P~Alex Dow, Jon~Michael Kleinberg, and Jure
  Leskovec.
\newblock Can cascades be predicted?
\newblock In {\em Proceedings of the 23rd international conference on World
  wide web}, pages 925--936, 2014.

\bibitem{AntDov13}
Demetris Antoniades and Constantine Dovrolis.
\newblock Co-evolutionary dynamics in social networks: A case study of twitter.
\newblock {\em arXiv preprint arXiv:1309.6001}, 2013.

\bibitem{WenRatPerGonCasBonSchMenFla13}
Lilian Weng, Jacob Ratkiewicz, Nicola Perra, Bruno Gon{\c{c}}alves, Carlos
  Castillo, Francesco Bonchi, Rossano Schifanella, Filippo Menczer, and
  Alessandro Flammini.
\newblock The role of information diffusion in the evolution of social
  networks.
\newblock In {\em Proceedings of the 19th ACM SIGKDD international conference
  on Knowledge discovery and data mining}, pages 356--364. ACM, 2013.

\bibitem{MyeLes14}
Seth~A Myers and Jure Leskovec.
\newblock The bursty dynamics of the twitter information network.
\newblock In {\em 23rd International Conference on the World Wide Web}, pages
  913--924, 2014.

\bibitem{GomLesKra10}
Manuel Gomez-Rodriguez, Jure Leskovec, and Andreas Krause.
\newblock Inferring networks of diffusion and influence.
\newblock In {\em Proceedings of the 16th ACM SIGKDD international conference
  on Knowledge discovery and data mining}, pages 1019--1028. ACM, 2010.

\bibitem{GomBalSch11}
Manuel Gomez-Rodriguez, David Balduzzi, and Bernhard Sch{\"o}lkopf.
\newblock Uncovering the temporal dynamics of diffusion networks.
\newblock In {\em Proceedings of the International Conference on Machine
  Learning}, 2011.

\bibitem{DuSonRodZha13}
Nan Du, Le~Song, Manuel Gomez-Rodriguez, and Hongyuan Zha.
\newblock Scalable influence estimation in continuous-time diffusion networks.
\newblock In {\em Advances in Neural Information Processing Systems 26}, 2013.

\bibitem{FarGomDuZamZhaSon15}
Mehrdad Farajtabar, Manuel Gomez-Rodriguez, Nan Du, Mohammad Zamani, Hongyuan
  Zha, and Le~Song.
\newblock Back to the past: Source identification in diffusion networks from
  partially observed cascades.
\newblock In {\em Proceedings of the 18th International Conference on
  Artificial Intelligence and Statistics (AISTATS)}, 2015.

\bibitem{ChaZhaFal2004}
Deepayan Chakrabarti, Yiping Zhan, and Christos Faloutsos.
\newblock R-mat: A recursive model for graph mining.
\newblock {\em Computer Science Department}, page 541, 2004.

\bibitem{LesBacKumTom08}
Jure Leskovec, Lars Backstrom, Ravi Kumar, and Andrew Tomkins.
\newblock Microscopic evolution of social networks.
\newblock In {\em Proceedings of the 14th ACM SIGKDD international conference
  on Knowledge discovery and data mining}, pages 462--470. ACM, 2008.

\bibitem{LesChaKleFaletal10}
Jure Leskovec, Deepayan Chakrabarti, Jon Kleinberg, Christos Faloutsos, and
  Zoubin Ghahramani.
\newblock Kronecker graphs: An approach to modeling networks.
\newblock {\em Journal of Machine Learning Research}, 11(Feb):985--1042, 2010.

\bibitem{Liniger09}
Thomas~Josef Liniger.
\newblock {\em Multivariate Hawkes Processes}.
\newblock PhD thesis, Swiss Federal Institute of Technology Zurich, 2009.

\bibitem{BluBecHelKat12}
Charles Blundell, Jeff Beck, and Katherine~A Heller.
\newblock Modelling reciprocating relationships with hawkes processes.
\newblock In {\em nips}, 2012.

\bibitem{IwaShaGha13}
Tomoharu Iwata, Amar Shah, and Zoubin Ghahramani.
\newblock Discovering latent influence in online social activities via shared
  cascade poisson processes.
\newblock In {\em Proceedings of the 19th ACM SIGKDD international conference
  on Knowledge discovery and data mining}, pages 266--274. ACM, 2013.

\bibitem{ZhoZhaSon13}
Ke~Zhou, Hongyuan Zha, and Le~Song.
\newblock Learning social infectivity in sparse low-rank networks using
  multi-dimensional hawkes processes.
\newblock In {\em Artificial Intelligence and Statistics (AISTATS)}, 2013.

\bibitem{ZhoZhaSon13b}
Ke~Zhou, Hongyuan Zha, and Le~Song.
\newblock Learning triggering kernels for multi-dimensional hawkes processes.
\newblock In {\em International Conference on Machine Learning (ICML)}, 2013.

\bibitem{FarDuGomValZhaSon14}
Mehrdad Farajtabar, Nan Du, Manuel Gomez-Rodriguez, Isabel Valera, Hongyuan
  Zha, and Le~Song.
\newblock Shaping social activity by incentivizing users.
\newblock In {\em Advances in Neural Information Processing Systems (NIPS)},
  2014.

\bibitem{LinAdaRya14}
Scott~W Linderman and Ryan~P Adams.
\newblock Discovering latent network structure in point process data.
\newblock In {\em International Conference on Machine Learning (ICML)}, 2014.

\bibitem{DuFarAhmSmoSon15}
Nan Du, Mehrdad Farajtabar, Amr Ahmed, Alexander~J Smola, and Le~Song.
\newblock Dirichlet-hawkes processes with applications to clustering
  continuous-time document streams.
\newblock In {\em KDD}. ACM, 2015.

\bibitem{ValGom15}
Isabel Valera and Manuel Gomez-Rodriguez.
\newblock Modeling adoption of competing products and conventions in social
  media.
\newblock {\em IEEE International Conference on Data Mining}, 2015.

\bibitem{HunSmyVuAsu11}
David Hunter, Padhraic Smyth, Duy~Q Vu, and Arthur~U Asuncion.
\newblock Dynamic egocentric models for citation networks.
\newblock In {\em Proceedings of the 28th International Conference on Machine
  Learning}, pages 857--864, 2011.

\bibitem{VuHunSmyAsu11}
Duy~Q Vu, David Hunter, Padhraic Smyth, and Arthur~U Asuncion.
\newblock Continuous-time regression models for longitudinal networks.
\newblock In {\em Advances in Neural Information Processing Systems}, pages
  2492--2500, 2011.

\bibitem{LesKleFal05}
Jure Leskovec, Jon Kleinberg, and Christos Faloutsos.
\newblock Graphs over time: densification laws, shrinking diameters and
  possible explanations.
\newblock In {\em Proceedings of the eleventh ACM SIGKDD international
  conference on Knowledge discovery in data mining}, pages 177--187. ACM, 2005.

\bibitem{GoeWatGol12}
Sharad Goel, Duncan~J Watts, and Daniel~G Goldstein.
\newblock The structure of online diffusion networks.
\newblock In {\em Proceedings of the 13th ACM conference on electronic
  commerce}, pages 623--638, 2012.

\bibitem{AalBorGje08}
Odd Aalen, Ornulf Borgan, and Hakon Gjessing.
\newblock {\em Survival and event history analysis: a process point of view}.
\newblock Springer, 2008.

\bibitem{KemKleTar03}
David Kempe, Jon Kleinberg, and {\'E}va Tardos.
\newblock Maximizing the spread of influence through a social network.
\newblock In {\em SIGKDD}, pages 137--146. ACM, 2003.

\bibitem{Ogata81}
Yosihiko Ogata.
\newblock On lewis' simulation method for point processes.
\newblock {\em Information Theory, IEEE Transactions on}, 27(1):23--31, 1981.

\bibitem{Ross06}
Sheldon~M. Ross.
\newblock {\em Introduction to Probability Models, Tenth Edition}.
\newblock Academic Press, Inc., 2011.

\bibitem{BoyVan04}
Stephen Boyd and Lieven Vandenberghe.
\newblock {\em Convex Optimization}.
\newblock Cambridge University Press, Cambridge, England, 2004.

\bibitem{HunLan04}
David~R Hunter and Kenneth Lange.
\newblock A tutorial on mm algorithms.
\newblock {\em The American Statistician}, 58(1):30--37, 2004.

\bibitem{ErdRen60}
Paul Erdos and A~R{\'e}nyi.
\newblock On the evolution of random graphs.
\newblock {\em Publ. Math. Inst. Hungar. Acad. Sci}, 5:17--61, 1960.

\bibitem{BacBolRosUgaVig12}
Lars Backstrom, Paolo Boldi, Marco Rosa, Johan Ugander, and Sebastiano Vigna.
\newblock Four degrees of separation.
\newblock In {\em Proceedings of the 4th Annual ACM Web Science Conference},
  pages 33--42, 2012.

\bibitem{Gra73}
Mark Granovetter.
\newblock The strength of weak ties.
\newblock {\em American journal of sociology}, pages 1360--1380, 1973.

\bibitem{RomKle10}
Daniel~Mauricio Romero and Jon Kleinberg.
\newblock The directed closure process in hybrid social-information networks,
  with an analysis of link formation on twitter.
\newblock In {\em ICWSM}, 2010.

\bibitem{UgaBacKle13}
Johan Ugander, Lars Backstrom, and Jon Kleinberg.
\newblock Subgraph frequencies: Mapping the empirical and extremal geography of
  large graph collections.
\newblock In {\em Proceedings of the 22nd international conference on World
  Wide Web}, pages 1307--1318. International World Wide Web Conferences
  Steering Committee, 2013.

\bibitem{WatStr98}
Duncan~J Watts and Steven~H Strogatz.
\newblock Collective dynamics of small-world networks.
\newblock {\em Nature}, 393(6684):440--442, June 1998.

\bibitem{DalVer2007}
Daryl~J Daley and David Vere-Jones.
\newblock {\em An introduction to the theory of point processes: volume II:
  general theory and structure}.
\newblock Springer Science \& Business Media, 2007.

\bibitem{phan2015natural}
Tuan~Q Phan and Edoardo~M Airoldi.
\newblock A natural experiment of social network formation and dynamics.
\newblock {\em Proceedings of the National Academy of Sciences},
  112(21):6595--6600, 2015.

\bibitem{Doreian13}
Patrick Doreian and Frans Stokman.
\newblock {\em Evolution of social networks}.
\newblock Routledge, 2013.

\bibitem{wang2011time}
Eric Wang, Jorge Silva, Rebecca Willett, and Lawrence Carin.
\newblock Time-evolving modeling of social networks.
\newblock In {\em Acoustics, Speech and Signal Processing (ICASSP), 2011 IEEE
  International Conference on}, pages 2184--2187. IEEE, 2011.

\bibitem{Newman10}
Mark Newman.
\newblock {\em Networks: an introduction}.
\newblock Oxford University Press, 2010.

\bibitem{Barrat08}
Alain Barrat, Marc Barthelemy, and Alessandro Vespignani.
\newblock {\em Dynamical processes on complex networks}.
\newblock Cambridge University Press, 2008.

\bibitem{GolZheFieAir10}
Anna Goldenberg, Alice~X Zheng, Stephen~E Fienberg, and Edoardo~M Airoldi.
\newblock A survey of statistical network models.
\newblock {\em Foundations and Trends{\textregistered} in Machine Learning},
  2(2):129--233, 2010.

\bibitem{Heise89}
David~R Heise.
\newblock Modeling event structures*.
\newblock {\em Journal of Mathematical Sociology}, 14(2-3):139--169, 1989.

\bibitem{Girvan02}
Michelle Girvan and Mark~EJ Newman.
\newblock Community structure in social and biological networks.
\newblock {\em Proceedings of the national academy of sciences},
  99(12):7821--7826, 2002.

\bibitem{LesBacKle09}
J.~Leskovec, L.~Backstrom, and J.~Kleinberg.
\newblock Meme-tracking and the dynamics of the news cycle.
\newblock In {\em Proceedings of the 15th ACM SIGKDD international conference
  on Knowledge discovery and data mining}, pages 497--506. ACM, 2009.

\bibitem{SniLuc06}
Tom~AB Snijders and SR~Luchini.
\newblock Statistical methods for network dynamics.
\newblock In {\em Proceedings of the XLIII Scientific Meeting, Italian
  Statistical Society}. CLEUP, 2006.

\bibitem{BraLerSni}
Ulrik Brandes, J{\"u}rgen Lerner, and Tom~AB Snijders.
\newblock Networks evolving step by step: Statistical analysis of dyadic event
  data.
\newblock In {\em Social Network Analysis and Mining, 2009. ASONAM'09.
  International Conference on Advances in}, pages 200--205. IEEE, 2009.

\bibitem{GhoLer12}
Rumi Ghosh and Kristina Lerman.
\newblock The role of dynamic interactions in multi-scale analysis of network
  structure.
\newblock {\em CoRR}, 2012.

\bibitem{HogLer12}
Tad Hogg and Kristina Lerman.
\newblock Social dynamics of digg.
\newblock {\em EPJ Data Science}, 1(1):1--26, 2012.

\bibitem{LerGalVerSteetal11}
Kristina Lerman, Aram Galstyan, Greg Ver~Steeg, and Tad Hogg.
\newblock Social mechanics: An empirically grounded science of social media.
\newblock In {\em Fifth International AAAI Conference on Weblogs and Social
  Media}, 2011.

\bibitem{Holme15}
Petter Holme.
\newblock Modern temporal network theory: A colloquium.
\newblock {\em arXiv preprint arXiv:1508.01303}, 2015.

\bibitem{bhattacharya2015analyzing}
Prasanta Bhattacharya, Tuan~Q Phan, and Edoardo~M Airoldi.
\newblock Analyzing the co-evolution of network structure and content
  generation in online social networks.
\newblock {\em ECIS 2015 Completed Research Papers}, page~18, 2015.

\bibitem{GroBla08}
Thilo Gross and Bernd Blasius.
\newblock Adaptive coevolutionary networks: a review.
\newblock {\em Journal of The Royal Society Interface}, 5(20):259--271, 2008.

\bibitem{GroSay09}
Thilo Gross and Hiroki Sayama.
\newblock {\em Adaptive networks}.
\newblock Springer, 2009.

\bibitem{SayPesSchBusetal13}
Hiroki Sayama, Irene Pestov, Jeffrey Schmidt, Benjamin~James Bush, Chun Wong,
  Junichi Yamanoi, and Thilo Gross.
\newblock Modeling complex systems with adaptive networks.
\newblock {\em Computers \& Mathematics with Applications}, 65(10):1645--1664,
  2013.

\bibitem{BorRoh00}
Stefan Bornholdt and Thimo Rohlf.
\newblock Topological evolution of dynamical networks: Global criticality from
  local dynamics.
\newblock {\em Physical Review Letters}, 84(26):6114, 2000.

\bibitem{GroDliBla06}
Thilo Gross, Carlos J~Dommar DÕLima, and Bernd Blasius.
\newblock Epidemic dynamics on an adaptive network.
\newblock {\em Physical review letters}, 96(20):208701, 2006.

\bibitem{ZanRis08}
Dami{\'a}n~H Zanette and Sebasti{\'a}n Risau-Gusm{\'a}n.
\newblock Infection spreading in a population with evolving contacts.
\newblock {\em Journal of biological physics}, 34(1-2):135--148, 2008.

\bibitem{ZscBohSeiHueetal12}
Gerd Zschaler, Gesa~A B{\"o}hme, Michael Sei{\ss}inger, Cristi{\'a}n Huepe, and
  Thilo Gross.
\newblock Early fragmentation in the adaptive voter model on directed networks.
\newblock {\em Physical Review E}, 85(4):046107, 2012.

\bibitem{Snijders14}
Tom~AB Snijders.
\newblock Siena: Statistical modeling of longitudinal network data.
\newblock In {\em Encyclopedia of Social Network Analysis and Mining}, pages
  1718--1725. Springer, 2014.

\bibitem{AgaLiuTanYu08}
Nitin Agarwal, Huan Liu, Lei Tang, and Philip~S Yu.
\newblock Identifying the influential bloggers in a community.
\newblock In {\em Proceedings of the 2008 international conference on web
  search and data mining}, pages 207--218. ACM, 2008.

\bibitem{GruGuhLibTom04}
Daniel Gruhl, Ramanathan Guha, David Liben-Nowell, and Andrew Tomkins.
\newblock Information diffusion through blogspace.
\newblock In {\em Proceedings of the 13th international conference on World
  Wide Web}, pages 491--501. ACM, 2004.

\bibitem{SunCheLiuWanetal11}
Tao Sun, Wei Chen, Zhenming Liu, Yajun Wang, Xiaorui Sun, Ming Zhang, and
  Chin-Yew Lin.
\newblock Participation maximization based on social influence in online
  discussion forums.
\newblock In {\em Proceedings of the International AAAI Conference on Weblogs
  and Social Media}, 2011.

\bibitem{GuoBluWalHel15}
Fangjian Guo, Charles Blundell, Hanna Wallach, Katherine Heller, and UCL
  Gatsby~Unit.
\newblock The bayesian echo chamber: Modeling social influence via linguistic
  accommodation.
\newblock In {\em Proceedings of the Eighteenth International Conference on
  Artificial Intelligence and Statistics}, pages 315--323, 2015.

\bibitem{WenLimJiaHe10}
Jianshu Weng, Ee-Peng Lim, Jing Jiang, and Qi~He.
\newblock Twitterrank: finding topic-sensitive influential twitterers.
\newblock In {\em Proceedings of the third ACM international conference on Web
  search and data mining}, pages 261--270. ACM, 2010.

\bibitem{BakHofMasWat11}
Eytan Bakshy, Jake~M. Hofman, Winter~A. Mason, and Duncan~J. Watts.
\newblock Everyone's an influencer: Quantifying influence on twitter.
\newblock In {\em WSDM}, pages 65--74, 2011.

\bibitem{SaiNakKim08}
Kazumi Saito, Ryohei Nakano, and Masahiro Kimura.
\newblock Prediction of information diffusion probabilities for independent
  cascade model.
\newblock In {\em Knowledge-based intelligent information and engineering
  systems}, pages 67--75. Springer, 2008.

\bibitem{GoyBonLak10}
Amit Goyal, Francesco Bonchi, and Laks~VS Lakshmanan.
\newblock Learning influence probabilities in social networks.
\newblock In {\em Proceedings of the third ACM international conference on Web
  search and data mining}, pages 241--250. ACM, 2010.

\bibitem{RodSch12}
M.G. Rodriguez and B.~Sch{\"o}lkopf.
\newblock Influence maximization in continuous time diffusion networks.
\newblock In {\em Proceedings of the International Conference on Machine
  Learning}, 2012.

\bibitem{RicDom02}
Matthew Richardson and Pedro Domingos.
\newblock Mining knowledge-sharing sites for viral marketing.
\newblock In {\em Proceedings of the eighth ACM SIGKDD international conference
  on Knowledge discovery and data mining}, pages 61--70. ACM, 2002.

\bibitem{DomRic01}
Pedro Domingos and Matt Richardson.
\newblock Mining the network value of customers.
\newblock In {\em Proceedings of the seventh ACM SIGKDD international
  conference on Knowledge discovery and data mining}, pages 57--66. ACM, 2001.

\bibitem{BhaGoyLak12}
Smriti Bhagat, Amit Goyal, and Laks~VS Lakshmanan.
\newblock Maximizing product adoption in social networks.
\newblock In {\em Proceedings of the fifth ACM international conference on Web
  search and data mining}, pages 603--612. ACM, 2012.

\bibitem{BhaChaPar10}
Rushi Bhatt, Vineet Chaoji, and Rajesh Parekh.
\newblock Predicting product adoption in large-scale social networks.
\newblock In {\em Proceedings of the 19th ACM international conference on
  Information and knowledge management}, pages 1039--1048. ACM, 2010.

\bibitem{DuSonYuaSmo12}
Nan Du, Le~Song, Ming Yuan, and Alex~J Smola.
\newblock Learning networks of heterogeneous influence.
\newblock In {\em Advances in Neural Information Processing Systems}, pages
  2780--2788, 2012.

\bibitem{YanZha13}
Shuang-Hong Yang and Hongyuan Zha.
\newblock Mixture of mutually exciting processes for viral diffusion.
\newblock In {\em Proceedings of the 30th International Conference on Machine
  Learning (ICML-13)}, pages 1--9, 2013.

\bibitem{DuSonWooZha13}
Nan Du, Le~Song, Hyenkyun Woo, and Hongyuan Zha.
\newblock Uncover topic-sensitive information diffusion networks.
\newblock In {\em Proceedings of the sixteenth international conference on
  artificial intelligence and statistics}, pages 229--237, 2013.

\bibitem{GomLesSch13}
Manuel~Gomez Rodriguez, Jure Leskovec, and Bernhard Sch{\"o}lkopf.
\newblock Modeling information propagation with survival theory.
\newblock {\em arXiv preprint arXiv:1305.3616}, 2013.

\bibitem{lian2015multitask}
Wenzhao Lian, Ricardo Henao, Vinayak Rao, Joseph Lucas, and Lawrence Carin.
\newblock A multitask point process predictive model.
\newblock In {\em Proceedings of the 32nd International Conference on Machine
  Learning (ICML-15)}, pages 2030--2038, 2015.

\bibitem{parikh2012conjoint}
Ankur~P Parikh, Asela Gunawardana, and Christopher Meek.
\newblock Conjoint modeling of temporal dependencies in event streams.
\newblock In {\em UAI Bayesian Modelling Applications Workshop}. Citeseer,
  2012.

\bibitem{hall2014tracking}
Eric~C Hall and Rebecca~M Willett.
\newblock Tracking dynamic point processes on networks.
\newblock {\em arXiv preprint arXiv:1409.0031}, 2014.

\bibitem{GoyBonLakVen10}
Amit Goyal, Francesco Bonchi, Laks~VS Lakshmanan, and Suresh
  Venkatasubramanian.
\newblock Approximation analysis of influence spread in social networks.
\newblock {\em arXiv preprint arXiv:1008.2005}, 2010.

\bibitem{Lian2014}
W.~Lian, V.~A. Rao, B.~Eriksson, and L.~Carin.
\newblock Modeling correlated arrival events with latent semi-markov processes.
\newblock In {\em Proceedings of the International Conference on Machine
  Learning}, 2014.

\bibitem{gunawardana2011model}
Asela Gunawardana, Christopher Meek, and Puyang Xu.
\newblock A model for temporal dependencies in event streams.
\newblock In {\em Advances in Neural Information Processing Systems}, pages
  1962--1970, 2011.

\bibitem{SinWagStr12}
Philipp Singer, Claudia Wagner, and Markus Strohmaier.
\newblock Factors influencing the co-evolution of social and content networks
  in online social media.
\newblock In {\em Modeling and Mining Ubiquitous Social Media}, pages 40--59.
  Springer, 2012.

\bibitem{TraFarSonZha15}
Long Tran, Mehrdad Farajtabar, Le~Song, and Hongyuan Zha.
\newblock Netcodec: Community detection from individual activities.
\newblock In {\em SDM}, 2015.

\end{thebibliography}

\end{document}